\documentclass[11pt]{article}
\usepackage{setspace}
\usepackage{amsmath,amssymb}
\usepackage{amsthm}
\usepackage{algorithm, algpseudocode}
\usepackage{url}
\usepackage{fullpage}
\usepackage{color}
\usepackage{graphicx}

\newtheorem{claim}{Claim}[section]
\newtheorem{lemma}[claim]{Lemma}

\newtheorem{theorem}{Theorem}
\newtheorem{proposition}[claim]{Proposition}
\newtheorem{corollary}[claim]{Corollary}
\newtheorem{definition}[claim]{Definition}

\theoremstyle{definition}
\newtheorem{remark}[claim]{Remark}

\def\cI{{\cal I}}
\def\supp{{\rm supp}}

\def\cH{{\cal H}}
\def\oM{\overline{M}}
\def\<{\langle}
\def\>{\rangle}
\def\prob{{\mathbb P}}

\def\naturals{{\mathbb N}}
\def\E{{\mathbb E}} 

\def\reals{\mathbb{R}}

\def\sT{{\sf T}}

\def\bh{\mathbf{h}}
\def\br{\mathbf{r}}
\def\bs{\mathbf{s}}
\def\bx{\mathbf{x}}
\def\by{\mathbf{y}}
\def\bu{\mathbf{u}}
\def\btu{\mathbf{\tilde{u}}}
\def\bts{\mathbf{\tilde{s}}}
\def\bv{\mathbf{v}}
\def\hbv{\mathbf{\widehat{v}}}

\def\bg{\mathbf{g}}
\def\bh{\mathbf{h}}
\def\one{\mathbf{1}}
\def\bz{\mathbf{z}}

\def\bw{\mathbf{w}}
\def\hbv{\widehat{\mathbf{v}}}
\def\hbu{\widehat{\mathbf{u}}}
\def\bvz{\mathbf{v_0}}
\def\buz{\mathbf{u_0}}
\def\bX{\mathbf{X}}

\def\rec{^\text{rec}}
\def\sym{^\text{sym}}
\def\bZ{\mathbf{Z}}
\def\bG{{G}}

\def\bV{{V}}
\def\bW{\mathbf{W}}
\def\cW{\mathcal{W}}

\def\E{\mathbb{E}}
\def\normal{{\sf N}}
\def\F{{\sf F}}
\def\H{{\sf H}}
\def\cF{{\cal F}}
\def\cP{{\cal P}}
\def\T{{\sf T}}

\def\S{{\sf S}}
\def\D{{\sf D}}
\def\Y{{\sf Y}}

\def\q{{\sf q}}
\def\R{{\sf R}_V}
\def\G{{\sf G}}
\def\id{{\rm I}}
\def\eps{\varepsilon}
\def\val{r}
\def\trans{^\sT}

\def\de{{\rm d}}
\def\tons{{\sf \tilde{b}}}
\def\ons{{\sf b}}
\def\onsd{{\sf d}}
\def\onsD{{\sf D}}
\def\onsa{{\sf a}}
\def\beps{{\bar \eps}}

\def\vth{\vartheta}
\def\cX{{\cal X}}
\def\cY{{\cal Y}}

\def\ind{{\mathbb{I}}}
\def\bDelta{{\bf \Delta}}

\def\brho{{\boldsymbol \rho}}
\def\bkappa{{\boldsymbol \kappa}}

\author{Andrea~Montanari\footnote{Department of Electrical
    Engineering and Department of Statistics, Stanford University}
\, and\, Emile Richard\footnote{Department of Electrical
    Engineering, Stanford University}}
\title{Non-negative Principal Component Analysis:\\
Message Passing Algorithms and Sharp Asymptotics}
\begin{document}
\maketitle

\begin{abstract}
Principal component analysis (PCA) aims at estimating the direction of
maximal variability of a high-dimensional dataset. A natural question is: does this task become
easier, and estimation more accurate, when we exploit additional knowledge on
the principal vector? We study the case in which the principal vector  is known to lie in the
positive orthant. Similar constraints arise in a number of
applications,  ranging from analysis of gene
expression data to spike sorting in neural signal processing. 

In the unconstrained case, the estimation performances of PCA 
has been precisely characterized using random matrix
theory, under a statistical model known as the `spiked model.'  
It is known that the estimation error undergoes
a phase transition as the signal-to-noise ratio crosses a certain threshold.
Unfortunately, tools from random matrix theory have no bearing on the constrained problem.
Despite this challenge, we develop an analogous characterization in
the constrained case, within a
one-spike model.

In particular: $(i)$~We prove that the estimation error undergoes a
similar phase transition, albeit at a different threshold in
signal-to-noise ratio that we determine exactly; $(ii)$~We prove that
--unlike in the unconstrained case-- estimation error depends on the
spike vector, and characterize the least favorable vectors; $(iii)$~We
show that a non-negative principal component can be approximately
computed --under the spiked model-- in nearly linear time. This
despite the fact that the problem is non-convex and, in general,
NP-hard to solve exactly. 
\end{abstract}

\section{Introduction}

Principal Component Analysis (PCA) is arguably the most successful of
dimensionality reduction techniques. 
Given samples $\bx_1,\bx_2,\dots,\bx_n$ from a $p$-dimensional
distribution, $\bx_i\in \reals^p$, PCA seeks the direction of maximum
variability. Assuming for simplicity the $\bx_i$'s to be centered
(i.e. $\E(\bx_i) = 0$),
and denoting by $\bx$ a random vector distributed  as $\bx_i$, the
objective is to estimate the solution of 
\begin{align}
\text{maximize} & \;\;\;\E\big(\<\bx,\bv\>^2\big)\, ,\label{eq:PopulationPCA}\\
\text{subject to} & \;\;\;\|\bv\|_2= 1\, .\nonumber
\end{align}
The solution of this problem is the principal eigenvector of the
covariance matrix $\E(\bx\bx^{\sT})$. This is normally estimated by
replacing expectation above by the sample mean, i.e. solving
\begin{align}\tag*{Classical PCA}
\text{maximize} & \;\;\;\sum_{i=1}^n\<\bx_i,\bv\>^2\, ,\label{eq:SamplePCA}\\
\text{subject to} & \;\;\;\|\bv\|_2= 1\, .\nonumber
\end{align}
Denoting by $\bX\in\reals^{n\times p}$ the matrix with rows $\bx_1,\bx_2,\dots,\bx_n$, the solution is of course given  by the principal eigenvector of the
sample covariance $\bX\bX^{\sT}/n=\sum_{i=1}^n\bx_i\bx_i^{\sT}/n$,
that we will denote by $\bv_1=\bv_1(\bX)$.

This approach is known to be consistent in low dimension. Let $\bvz$ be the solution of problem
(\ref{eq:PopulationPCA}). If $n/p\to\infty$, then $\|\bv_1-\bvz\|_2\to
0$ in probability \cite{anderson1963asymptotic}.
On the other hand, it is well understood that consistency can break
dramatically in the high-dimensional regime $n=O(p)$. This phenomenon is crisply captured by
the spiked covariance model \cite{johnstone2004sparse,johnstone2009consistency}, that postulates
\begin{align}
\bx_i = \sqrt{\beta}\, u_{0,i} \, \bvz + \, \bz_i\, ,
\end{align}
where $\bvz$ has unit norm, $\bz_1,\bz_2,\dots \bz_p$ are i.i.d. $p$-dimensional standard
normal vectors  $\bz_i\sim \normal(0,\id_{p}/n)$, and $\buz =
(u_{0,1},\dots,u_{0,n})^{\sT}$ is a unit-norm vector\footnote{The
  definition of \cite{johnstone2004sparse} assumes
  $u_i\sim_{\text{i.i.d.}}\normal(0,1/n)$ but for our purposes it is
  more convenient to consider $\bu$ as a given deterministic
  vector. Equivalently, we can condition on $\bu_0$.}.
The above model can also be written as
\begin{align}\tag*{Spiked Model}
\bX = \sqrt{\beta} \, \buz\,\bvz^{\sT} + \bZ\, ,\label{eq:SpikedMatrix}
\end{align}
where $\bZ\in\reals^{n}$ has i.i.d. entries $\bZ_{ij}\sim
\normal(0,1/n)$.

The spectral properties of the random matrix $\bX$ defined by
the \ref{eq:SpikedMatrix} have been studied in detail 
across statistics, signal processing and probability theory
\cite{baik2005phase,baik2006eigenvalues,baik2006eigenvalues,paul2007asymptotics,feral2009largest,benaych2012singular,capitaine2012central}. In the limit $n,p\to\infty$ with $p/n\to \alpha\in (0,\infty)$,
the leading eigenvector 
$\bv_1$ undergoes a phase transition:
\begin{align}
\lim_{n\to\infty}\big|\<\bv_1,\bvz\>\big|=  \begin{cases}
0 & \mbox{ if $\beta \leq \sqrt{\alpha}$,}\\
&\\
\sqrt{{\displaystyle\frac{1-\alpha/\beta^2}{1+\alpha/\beta}}} & \mbox{ if $\beta>\sqrt{\alpha}$,}
\end{cases}\label{eq:PCAPhaseTransition}
\end{align}
In other words,  \ref{eq:SamplePCA} contains information about the signal $\bvz$ if and
only if the signal-to-noise ratio is above the threshold
$\sqrt{\alpha}$.
Below that threshold, the principal component is asymptotically
orthogonal to the signal.

The failure of PCA has motivated significant effort aimed at developing
better estimation methods. A recurring idea is to use additional
structural information about the principal eigenvector $\bvz$, such as
its sparsity \cite{johnstone2004sparse,zou2006sparse} or its distribution (within a Bayesian framework)
\cite{bishop1999bayesian,lawrence2009non}. Here we focus on the simplest type of structural information,
namely we assume $\bvz$ is known to be non-negative\footnote{Of
  course the case in which $\bvz \in Q$ with $Q$ an arbitrary, known,
  orthant, can be reduced to the present one.}. It is then natural to
replace the \ref{eq:SamplePCA} problem with the following one (whereby we
use the matrix $\bX$ to represent the data):
\begin{align}\tag*{Non-negative PCA}
\text{maximize} & \;\;\;\|\bX\bv\|_2^2\, ,\label{eq:PositivePCA}\\
\text{subject to} & \;\;\;\bv\ge 0\,, \;\;\;\; \|\bv\|_2= 1\,  .\nonumber
\end{align}
Notice  that this problem in non-convex and cannot be solved by
standard singular value decomposition. Indeed it is in general NP-hard
by reduction from maximum independent set
\cite{de2002approximation}. Two questions are therefore
natural: given the additional complexity induced by the non-negativity
  constraint, does this constraint reduce the statistical error significantly?
Are there efficient algorithms to solve the \ref{eq:PositivePCA} problem?

In this paper we answer \emph{positively} to both questions within the 
spiked covariance model. Namely denoting by $\bv^+$ the solution of
the \ref{eq:PositivePCA} problem, we provide the
following contributions: 
\begin{enumerate}
\item[$(i)$] We unveil a new phase transition phenomenon concerning
  $\bv^+$ that is analogous to the classical one, see Eq.
  (\ref{eq:PCAPhaseTransition}). Namely, for $\beta> \sqrt{\alpha / 2}$,
 $\<\bv^+,\bvz\>$ stays bounded away from $0$, while,
 for  $\beta< \sqrt{\alpha / 2}$, there exists vectors $\bvz$ such
 that $\<\bv^+,\bvz\>\to 0$ as $n,p\to\infty$.

Non-negative PCA is superior to classical PCA in this respect since
$\sqrt{\alpha / 2} <\sqrt{\alpha}$ strictly.
\item[$(ii)$] We prove an explicit formula for the asymptotic scalar
  product  $\lim_{n\to\infty} \<\bv^+,\bvz\>$.
Non-negative PCA is superior to \ref{eq:SamplePCA} also in this respect. 
Namely $\<\bv^+,\bvz\>$ is strictly larger than $|\<\bv_1,\bvz\>|$
with high probability as $n,p\to\infty$.

 Note that the non-negativity
constraint breaks the rotational invariance of classical PCA (under
the spiked model). As a consequence, not all spikes $\bvz$ are equally
hard --or easy-- to estimate. We use our theory to characterize the
least favorable vectors $\bvz$.
\item[$(iii)$] We prove that (for any fixed $\delta>0$) a $(1-\delta)$ approximation to the
  non-convex optimization \ref{eq:PositivePCA} problem can be found efficiently
  with high probability with respect to the noise realization. Our
  algorithm has complexity of order $T_{\rm mult}\log(1/\delta)$, where
  $T_{\rm mult}$ is the maximum of the complexity of multiplying a vector by $\bX$
  or by $\bX^{\sT}$.
\end{enumerate}
Technically, our approach has two components. We use Sudakov-Fernique
inequality to upper bound the expected value of the 
\ref{eq:PositivePCA} optimization problem. We then define an iterative algorithm
to solve the optimization problem, and evaluate the value achieved by
the algorithm after any number $t$ of iterations. This provides a
sequence of lower bounds which we prove converge to the upper bound as
the number of iterations increase.

More precisely, we  use an approximate message
passing (AMP) algorithm of the type introduced in
\cite{DMM09,BM-MPCS-2011}.  Each iteration requires a multiplication
by $\bX$ and a multiplication by $\bX^{\sT}$ plus some lower
complexity operations. 
While AMP is not guaranteed to solve the
\ref{eq:PositivePCA} problem for arbitrary matrices $\bX$, we establish
the following properties:
\begin{enumerate}
\item After any number of iterations $t$, the algorithm produces a
  running estimate $\bv^t\in\reals^p$ that satisfies the constraints
  $\bv^t\ge 0$ and $\|\bv^t\|_2=1$.

Further the limit $\lim_{n,p\to\infty}\|\bX\bv^t\|^2_2= \val(t)$ exists almost
surely, and $\val(t)$ can be computed explicitly as a function of the
empirical law of entries of $\bvz$. Analogously, the asymptotic correlation
$\lim_{n,p\to\infty}\<\bv^t,\bvz\>= s(t)$
can be computed explicitly. 
\item Denoting by $\val_{*}$ the upper bound  on the value of the
  optimization \ref{eq:PositivePCA} problem  implied by Sudakov-Fernique
  inequality,
we prove that $\val(t)\ge (1-\delta)\val_*$ for all $t\ge t_0(\delta)$
for some dimension-independent $t_0(\delta)$.  This implies that Sudakov-Fernique inequality is
asymptotically tight in the high-dimensional limit.
\item The asymptotic correlation converges to a limit as 
the number of iteration tends to infinity $s_*=\lim_{t\to\infty}s(t)$
(the convergence is, again, exponentially fast).
Further, if we add the constraint $|\<\bv,\bvz\>-s_*|\ge \delta$ to the \ref{eq:PositivePCA} optimization problem, Sudakov-Fernique's upper bound on the resulting value
is asymptotically smaller than $\val_*$ for any $\delta>0$.

This implies that $\lim_{n\to\infty}\<\bv^+,\bvz\>  = s_*$. 
\end{enumerate}

Finally, we generalize our analysis to the case of symmetric matrices,
namely assuming that data consist of a $n\times n$ symmetric matrix $\bX$:
\begin{align}\tag*{Symmetric Spiked Model}
\bX = \beta\, \bvz\bvz^{\sT} + \bZ\label{eq:SymmetricModel}
\end{align}
with $\bvz\ge 0$, $\|\bvz\|_2 = 1$. Here $\bZ = \bZ^{\sT}$ is a noise
matrix such that $(\bZ_{ij})_{i\le j}$ are independent with
$\bZ_{ij}\sim\normal(0,1/n)$ for $i<j$ and $\bZ_{ii}
\sim\normal(0,2/n)$.

In this case we study the analogue of the \ref{eq:PositivePCA} problem, namely
\begin{align}\tag*{Symmetric non-negative PCA}
\text{maximize} & \;\;\;\, \<\bv,\bX\bv\>\, ,\label{eq:PositivePCASymm}\\
\text{subject to} & \;\;\;\bv\ge 0\,, \;\;\;\; \|\bv\|_2= 1\,  .\nonumber
\end{align}

\subsection{Related literature}

The non-negativity constraint on principal components arises naturally
in many situations: we briefly discuss a few related areas. Let us
emphasize that the theoretical understanding of the methods
discussed below is much more limited than for \ref{eq:SamplePCA}.

\vspace{0.25cm}

\noindent{\bf Microarray data.} Microarray measurements of gene expression result in a matrix
$\bX\in\reals^{n\times p}$ whereby $\bX_{ij}$ denotes the expression
level of gene $j$ in sample $i$. Several authors  
\cite{tanay2002discovering,kluger2003spectral,madeira2004biclustering,shabalin2009finding,sun2013maximal} seek for a
subset of genes that are simultaneously over-expressed (or
under-expressed) in a subset of
samples. Lazzeroni and Owen \cite{lazzeroni2002plaid} propose a model
of the form
\begin{align}
\bX_{ij} = \mu_0+\sum_{k=1}^{K}\mu_k\,\brho^{(k)}_i \bkappa^{(k)}_j\, ,\label{eq:Plaid}
\end{align}
where $k$ indexes such gene groups (or `layers'), and $\brho^{(k)}$,
$\bkappa^{(k)}$ indicate the level of participation of different
samples or different genes in group $k$. These authors assume
$\brho^{(k)}_i \bkappa^{(k)}_i\in \{0,1\}$, but it is natural to relax
this condition allowing for partial participation in group $k$, i.e. $\brho^{(k)}_i \bkappa^{(k)}_i\in [0,1]$,
By a change of normalization, this constraint can be simplified to
$\brho^{(k)}_i \bkappa^{(k)}_i\ge 0$. Note a few differences with
respect to our work:
\begin{itemize}
\item[$(i)$] We study a model with only one non-negative
  component. While Eq.~(\ref{eq:Plaid}) corresponds to a model with
  multiple $K\ge 1$ components, in practice several authors fit one
  `layer' at a time, hence 
  effectively reducing the problem to a single-component case.

Extending our analysis to the multiple component case will be the
object of future work.
\item[$(ii)$] The non-negativity constraint is imposed in the model (\ref{eq:Plaid})
  on both components. This is a relatively straightforward
  modification of our setting.
\item[$(iii)$] Several studies (e.g. \cite{lazzeroni2002plaid})  fit models of the form
  Eq.~(\ref{eq:Plaid}) using greedy optimization methods. Their
  conclusions are based on the unproven belief that these methods approximately
  solve the  optimization problem.
  Our results
  (establishing convergence, with high probability, of an iterative
  method) provide some mathematical justification for this approach.
\end{itemize}

\vspace{0.25cm}

\noindent{\bf Neural signal processing.}   Neurons' activity can be recorded through thin implanted
electrodes. The resulting signal is a superposition of localized
effects of single neurons (spikes). In order to reconstruct the
single neuron activity, it is necessary to assign each spike to a
specific neuron that created it, a process known as `spike sorting' \cite{lewicki1998review,quiroga2004unsupervised,quiroga2009extracting}.
Once spikes are aligned, the resulting data can be viewed as a matrix
$\bX = (\bX_{ij})_{i\in [n], j\in [p]}$, where $i$ indexes the spikes
and $j$ time (or a transform domain, e.g. wavelet domain).

In this context, principal component analysis is often used to 
project each row of $\bX$ (i.e. each recorded spike) in a low
dimensional space, or decomposing it as a sum of single neurons
activity, see e.g. 
\cite{brown2001independent,zhang2004spike,pavlov2007sorting}. Clustering may be carried out after dimensionality reduction.
Note that each spike is a sum of single neuron activity \emph{with
  non-negative} coefficients. In other words, the $i$-th row of $\bX$
reads
\begin{align}
\bx_{i} \approx \sum_{k=1}^K u_{0,ik} \bvz^{(k)}\, ,
\end{align}
where $\bvz^{(1)}$, \dots $\bvz^{(K)}$ are the signatures of $K$
neurons and $u_{0,ik}$ are non-negative coefficients. 

Again, this corresponds to a multiple component version of the problem
we study here. To the best of our knowledge, the non-negativity
constraint has not been exploited in this context.

\vspace{0.25cm}

\noindent{\bf Non-negative matrix factorization.} Initially introduced
in the context of chemometrics
\cite{paatero1994positive,paatero1997least}, non-negative matrix
factorization attracted considerable interest because of its
applications in computer vision and topic modeling. In particular, Lee
and Seung \cite{lee1999learning} demonstrated empirically that
non-negative matrix
factorization successfully identifies parts of images, or topics in
documents' corpora.
 
A  mathematical model to
understand these findings was put forward in \cite{donoho2003does} and most
recently studied, for instance, in \cite{arora2012computing}. Note
that these results only apply under a no-noise or very-weak noise
conditions, but for multiple components. Further, the aim is to
approximate the original data matrix, rather than estimating the
principal components.

In this sense, non-negative matrix
factorization is the farther among all related areas to the scope of
our work.

\vspace{0.25cm}

\noindent{\bf Approximate Message Passing.} Approximate Message
Passing algorithms  proved successful as a fast
first-order method for compressed
sensing reconstruction  \cite{DMM09}. Their definition is
inspired by ideas from statistical mechanics and coding theory
\cite{thouless1977solution,SpinGlass,RiU08}, see also \cite{MontanariChapter} for
further background. One attractive feature of AMP algorithms is that their
high-dimensional asymptotics can be characterized exactly and in close
form, through `state-evolution'  \cite{BM-MPCS-2011,javanmard2013state,BM-Universality}.
Several applications and generalizations were developed by Rangan
\cite{RanganGAMP}, Schniter \cite{SchniterEM} and collaborators. 

In
particular Schniter and Cevher
\cite{schniter2011approximate,parker2013bilinear} apply AMP the
problem of reconstructing a vector from bilinear noisy
observations, a problem that is mathematically equivalent to the one
explored here. These authors consider however more complex Bayesian models, and
evaluate performances through empirical simulations, while we
characterize a fundamental threshold phenomenon in a worst case
setting. Similar ideas were applied in \cite{krzakala2013phase} to the
problem of dictionary learning, and in \cite{vila2013hyperspectral} to
hyperspectral imaging. Finally, Kabashima and collaborators
\cite{kabashima2014phase} study low-rank matrix reconstruction using a
similar approach, but focus on the case in which the rank scales
linearly with the matrix dimensions.

\subsection{Organization of the paper}

In Section \ref{sec:Main} we present formally our results, both for
symmetric matrices and rectangular matrices.
As mentioned above, the proof is obtained by establishing an upper bound on
the value of the \ref{eq:PositivePCA} optimization  problem using
Sudakov-Fernique inequality, and a lower bound by analyzing an AMP algorithm.
The upper bound is outlined in  Section
\ref{sec:UpperBoundSlepianNonnegative}. Section \ref{sec:AMP}
introduces formally AMP and its analysis, hence establishing the
desired lower bound as well as the convergence properties of this
algorithm. 
Section \ref{sec:Numerical} presents a numerical illustration of
the phase transition phenomenon, and of the behavior of our
algorithm. Finally, Section \ref{sec:proofUpperBounds} contains 
proofs, with some technical details deferred to the appendices.
%
%
\section{Main results}
\label{sec:Main}

In this section we present formally our results. For the sake of clarity ,
we consider first the case of symmetric (Wigner) matrices, and then
the case of rectangular (or sample covariance, Wishart) matrices. Indeed formul\ae~for symmetric
matrices are somewhat simpler.  Before doing that, it is convenient to
introduce some definitions.
(For basic notations, we invite the reader to consult Section \ref{sec:Notations}.)

\subsection{Definitions}

Our results concern sequences of matrices $\bX$ with diverging
dimensions $n,p$, and are expressed in terms of the asymptotic
empirical distribution of the entries of $\bvz$. This is formalized
through the following definition.
\begin{definition}\label{def:convergesEmpirical}
Let $\{\bx(n)\}_{n\ge 0}$ be a sequence of vectors with, for each $n$,
$\bx(n)\in \reals^n$, and $\mu$ be a (Borel) probability measure on the real
line $\reals$. Then we say that  $\bx(n)$ \emph{converges in
empirical distribution} to $\mu$ if the probability measure
\begin{align}
\mu_{\bx(n)} \equiv \frac{1}{n}\sum_{i=1}^n\delta_{\bx(n)_i}\, ,
\end{align}
converges weakly to $\mu$ and the
second moment  of $\mu_{\bx(n)}$ converges as well or, equivalently,
$\|\bx(n)\|_2^2/n\to \int x^2\mu(\de x)$.  

With an abuse of
terminology, we will say that $\{\bx(n)\}_{n\ge 0}$ converges in
empirical distribution to $X$ if $X$ is a random variable with law $\mu$. 
\end{definition}

Given a random variable $X$, we let $\mu_X$ denote its law. We next
define a few functions of such a law.
\begin{definition}\label{def:functionsFGRTS}
Let $\bV$ be a real non-negative random variable independent of $\bG\sim \normal(0,1)$
and $x\in\reals_{\ge 0}$ be a real number. We define the two functions
\begin{equation}\label{eq:defResponseFunctionF}
\F_\bV( x) \equiv \frac{ \E~ \bV \left ( x \bV +\bG\right
  )_+}{\sqrt{\E \left ( x\bV+\bG\right
    )_+^2}}\quad\text{and}\quad\G_\bV(x) \equiv  \frac{\E~ \bG \left (  x \bV+\bG \right )_+}{\sqrt{\E \left ( x\bV+\bG\right )_+^2}} ~~.
\end{equation}
Using $\F_V$ and $\G_V$ define the following  `Rayleigh functions'
\begin{align}\label{eq:definitionEnergy}
\R\sym(x) & \equiv \beta~ \F^2_V(x) + 2~ \G_\bV(x)\\
\R\rec(x,\alpha) &\equiv  \sqrt { 1 + \beta~  \F_\bV(x / \sqrt \alpha )^2}+ \sqrt \alpha~ \G_\bV(x / \sqrt \alpha )~~.
\end{align}

For $\beta\ge 0$, we also define $\T_V(\beta)$ as the unique non-negative 
solution of $x = \beta\F_V(x)$ and $\S_V(\beta,\alpha)$ as
 the unique non-negative solution of $x^2(1+\beta\F_V(x/\sqrt{\alpha})^2)
 = \beta^2 \F_V(x/\sqrt\alpha)^2$.
\end{definition}
Note that  the above functions depend on the random variable $\bV$
only through its law $\mu_{\bV}$, but we prefer the notation 
--say-- $\F_{\bV}$ to the 
more indirect $\F_{\mu_\bV}$.
Existence and well-definedness of $\T_V$ and $\S_V$ are proved in
Lemma \ref{lem:Twelldefined}
below. Further in Lemma \ref{lem:Rvariations} we prove that the functions $\R\sym$ respectively $\R\rec(\cdot,\alpha)$ have a unique maximum reached respectively at $\T_V(\beta)$ and at $\S_V(\beta,\alpha)$.

Our results become particularly explicit in case $\bvz$ is sparse
which (in the asymptotic setting) is equivalent to 
$\prob(\bV\neq 0)$  small. We introduce some
terminology to address this case.
\begin{definition}\label{def:UniformConvergence}
Given a real random variable $\bV$, we let $\eps(\bV) \equiv
\prob(\bV\neq 0)$ denote its sparsity level. We let
$\cP$ be the set of probability measures $\mu$ supported on
$\reals_{\ge 0}$, with second moment equal to one, and, for $\eps\ge 0$, $\cP_{\eps}\equiv
\{\mu\in\cP:\, \mu(\{0\})\ge 1-\eps\}$. 

Given a function $Q:\cP\to\reals$, $\mu_{\bV}\mapsto Q_{\bV}$, and a
number $q\in\reals$, we write
that $\lim_{\eps(\bV)\to 0} Q_{\bV} = q$ \emph{uniformly over $\cP$} if
\begin{align}
\lim_{n\to\infty}\inf_{\mu_{\bV}\in\cP_{\eps}} Q_{\bV} = 
\lim_{n\to\infty}\sup_{\mu_{\bV}\in\cP_{\eps}} Q_{\bV} = q\, .
\end{align}
\end{definition}

In the following, we will often state  that an event holds almost surely 
as the dimensions of the random matrix $\bX$ tend to infinity. It is
understood that such statements hold with respect to the law of a
sequence $\{\bX_n\}_{n\ge 1}$ of \emph{independent} random matrices
distributed according to the \ref{eq:SpikedMatrix} or  the \ref{eq:SymmetricModel}. 

\subsection{Symmetric matrices}

For the sake of comparison, we begin by recalling some asymptotic
properties of \ref{eq:SamplePCA}. Given $\bX\in\reals^{n\times n}$ symmetric
distributed according to the \ref{eq:SymmetricModel},
we denote by $\bv_1=\bv_1(\bX)$ its principal eigenvector, and by
$\lambda_1=\lambda_1(\bX)$  the corresponding eigenvalue.

This model has been studied in probability theory under the name of
`low rank deformation of a Wigner matrix'. The following is a
simplified version of the main theorem in  \cite{capitaine2009largest}.
\begin{theorem}[\cite{capitaine2009largest}]
Let $\bX = \beta\bvz\bvz^{\sT}+\bZ$ be a rank-one deformation of the
Gaussian symmetric matrix $\bZ$, with $\bZ_{ij}\sim\normal(0,1/n)$
independent for $i<j$, and $\|\bvz\|_2 =1$. Then we have, almost surely
\begin{align}
\lim_{n\to\infty}\lambda_{1}(\bX) = \begin{cases}
2 & \mbox{ if $\beta\le 1$,}\\
\beta + 1/\beta & \mbox{ if $\beta>1$.}
\end{cases}
\end{align}
Further
\begin{align}
\lim_{n\to\infty}|\<\bv_1,\bv_0\>| = \begin{cases}
0 & \mbox{ if $\beta\le 1$,}\\
\sqrt{1-\beta^{-2}} & \mbox{ if $\beta>1$.}
\end{cases}
\end{align}
\end{theorem}
Numerous refinements exist on this basic result,
see for instance
\cite{capitaine2009largest,peche2009universality,benaych2011eigenvalues,benaych2011fluctuations,capitaine2011free,benaych2012large,knowles2013isotropic,pizzo2013finite}.

Our analysis provides a version of this theorem that holds for
non-negative PCA, and is intriguingly similar to the original one. Its proof  can be found in Appendix \ref{sec:ProofOfMainTheorems}. 
\begin{theorem}\label{th:mainSym}
Let $\bX = \beta\bvz\bvz^{\sT}+\bZ$ be a rank-one deformation of the symmetric
Gaussian matrix $\bZ$  with $\bZ_{ij}\sim\normal(0,1/n)$
independent for $i<j$, and $\|\bvz\|_2 =1$. Further let $\lambda^+=\lambda^+(\bX)$ be the
value of the \ref{eq:PositivePCASymm} problem, and
$\bv^+=\bv^+(\bX)$ be any of the optimizers.
Finally assume that $\bv_0=\bv_0(n)\in \reals^n$ is such that
$\{\sqrt{n}\bvz(n)\}$  converges in empirical
distribution to $\mu_V$.

Then (with the notation introduced in Definition
\ref{def:functionsFGRTS}), we have almost surely
\begin{align}
\lim_{n\to\infty}\lambda^+(\bX) &= \R\sym(\T_V(\beta))\, \label{eq:lambdaPlusSim},\\
\lim_{n\to\infty}\<\bv^+,\bv_0\>&= \F_V(\T_V(\beta))\,   \label{eq:innerProductSim}.
\end{align}
Further, uniformly over $\cP$,
\begin{align}\label{eq:RlimitEpsZeroSym}
\lim_{\eps(\bV)\to 0}\R\sym(\T_V(\beta))&=\begin{cases}
\sqrt 2 & \mbox{ if $\beta\le 1/\sqrt{2}$,}\\
\beta  + 1/(2\beta) & \mbox{ otherwise.}
\end{cases} 
\end{align}
and
\begin{align}\label{eq:FlimitEpsZeroSym}
\lim_{\eps(V)\to 0}\F_V(\T_V(\beta))&=\begin{cases}
0 & \mbox{ if $\beta\le 1/\sqrt{2}$,}\\
\sqrt{ 1 - 1/(2\beta^2)} & \mbox{ otherwise.}
\end{cases} 
\end{align}
\end{theorem}
The statement in Theorem \ref{th:mainSym} is dependent on the
empirical distribution of the entries of $\bvz$. It is of special
interest to characterize the least favorable situation, i.e. the
distribution corresponding to the smallest scalar product
$\<\bv^+,\bvz\>$. 
This has two  motivations: $(i)$ to guarantee the minimum value of
$\<\bv^+,\bv_0\>$ achieved by a solution $\bv^+$ of the optimization
problem \ref{eq:PositivePCASymm} and $(ii)$ to describe the least
favorable signal $\bvz$.

The worst-case scenario is realized for a particularly simple
distribution, namely 2-atoms distribution, with an atom at $0$.
However, unlike in classical denoising  \cite{Donoho94}, the worst
case mixture is not obtained by setting all the allowed coordinates to
non-zero. In the following Theorem we are interested in the worst case
among $\beps$-sparse signals, or equivalently in vector sequences
$\{\bvz(n)\}_{n\geq0}$ such that $\lim_{n \to \infty} \|\bvz(n)\|_0/n
\leq \beps$, or $V \in \cP_{\beps}$ since
  sparse signals are naturally interesting for applications. 
\begin{theorem}\label{th:Fworstcase}
Consider  the \ref{eq:SymmetricModel} with the \ref{eq:PositivePCASymm}
estimator.

If $\beta \le 1/\sqrt 2$, then there exists a sequence of vectors $\{\bv_0(n)\}_{n \geq 0}$
such that  $\lim_{n \to \infty} \|\bvz(n)\|_0/n = 0$ and, almost surely,
\begin{align}
\lim_{n \to \infty } \langle \bv^+ , \bvz\rangle  = 0~.\label{eq:WorstCaseLowSNR}
\end{align}

For any $\beta> 1/\sqrt{2}$, there exists $\eps_{*}(\beta,\beps)\in (0,\beps]$
such that the following is true. Let $\bV_*$ be the random variable
with law
\begin{align}
\mu_{\bV_{*}} = (1-\eps_{*})\delta_0 + \eps_{*}\,
\delta_{1/\sqrt{\eps_{*}}}\, .
\end{align}
Then for any sequence of vectors $\{\bvz(n)\}_{n \geq 0}$ such that $\|\bvz(n)\|_0\le
n\beps$ we have, almost surely,
\begin{align}
\lim\inf_{n \to \infty } \langle \bv^+ , \bvz\rangle \geq
\F_{V_{*}}(\T_{V_{*}}(\beta))~. \label{eq:WorstCaseHighSNR}
\end{align}
Equality holds if $\bvz(n)$ is the vector with $n\eps_{*}$ non-zero
entries, all equal to $1/\sqrt{n\eps_*}$.
\end{theorem}
We defer this proof to Section \ref{sec:worstCase}. The worst case mixture $\eps_{\#}(\beta)$ as well as the function
$\F_{V_{*}}(\T_{V_{*}}(\beta))$ can be expressed explicitly in terms
of the Gaussian distribution function, see Section \ref{sec:worstCase}.

\subsection{Rectangular matrices}

We develop a very similar theory for the case of rectangular matrices.
Our first result characterizes the value of the \ref{eq:PositivePCA} problem,  and the estimation error, in
analogy with Theorem \ref{th:mainSym}. The proof can be found in Appendix \ref{sec:ProofOfMainTheorems}.
\begin{theorem}\label{th:mainRec}
Let $\bX = \sqrt \beta\buz\bvz^{\sT}+\bZ$ be a rank-one deformation of the
Gaussian matrix $\bZ$ with $\bZ_{ij}\sim\normal(0,1/n)$
independent, and $\|\buz\|_2=\|\bvz\|_2 =1$. Further let $\sigma^+=\sigma^+(\bX)$ be the
expected value of the  \ref{eq:PositivePCA} problem, and
$\bv^+=\bv^+(\bX)$ be any of the optimizers.

 Assume that $n,p\to\infty$ with  convergent aspect ratio 
 $p/n\to\alpha \in(0,\infty)$,
 and that $\bv_0=\bv_0(p)\in \reals^p$ converges in empirical
 distribution to $\mu_V$.

Then (with the notation introduced in Definition
\ref{def:functionsFGRTS}), we have almost surely
\begin{align}
\lim_{n\to\infty}\sigma^+(\bX) &= \R\rec(\S_V(\beta, \alpha),\alpha)\,
,\label{eq:MainRrec}\\
\lim_{n\to\infty}\<\bv^+,\bv_0\>&= \F_V(\S_V(\beta,\alpha) / \sqrt \alpha)\, .\label{eq:MainFrec}
\end{align}
Further, uniformly over $\cP$,
\begin{align}\label{eq:RlimitEpsZeroNonSym}
\lim_{\eps(V)\to 0}\R\rec(\S_V(\beta,\alpha),\alpha) = 
\begin{cases}
1+\sqrt {\alpha/2} &\mbox{if }\beta \le \sqrt{\alpha / 2}\, ,\\
 \sqrt{ \left (\sqrt\beta +\frac{\alpha}{ 2  \sqrt\beta } \right )
   \left (\sqrt\beta +\frac{1}{   \sqrt\beta } \right ) }
& \mbox{otherwise,}
\end{cases}
\end{align}
and
\begin{align}\label{eq:FlimitEpsZeroNonSym}
\lim_{\eps(V)\to 0} \F_V(\S_V(\beta, \alpha) / \sqrt \alpha)&= 
\begin{cases}
0 & \mbox{ if} \quad \beta\le \sqrt{\alpha / 2},\\
 \sqrt{ (\beta^2 - \alpha / 2)(\beta^2 + \beta \alpha / 2)^{-1}} &\mbox{ otherwise.}
\end{cases}. 
\end{align}
\end{theorem}
Finally, in the same fashion as Theorem \ref{th:Fworstcase}, 
we can characterize the worst case signals $\bvz$. 
\begin{theorem}\label{th:FworstcaseRec}
Consider  the \ref{eq:SpikedMatrix}, with the \ref{eq:PositivePCA}
estimator.

If $\beta \le \sqrt{\alpha/ 2}$, then there exists a sequence of vectors $\{\bv_0(p)\}_{p \geq 1}$
such that $\lim_{p \to \infty}\|\bvz(p)\|_0 /p  = 0$ and, almost surely,
\begin{align}
\lim_{p \to \infty } \langle \bv^+ , \bvz\rangle  = 0~.
\end{align}
For any $\beta> \sqrt{\alpha/2}$, there exists
$\eps_{rec,*}(\alpha,\beta,\beps)\in (0,\beps]$
such that the following is true. Let  $\bV_{*}$ be the random variable
with law $(1-\eps_{rec,*})\delta_0 + \eps_{rec,*}\,
\delta_{1/\sqrt{\eps_{rec,*}}}$.
Then for any sequence of vectors $\{\bv_0(p)\}_{p\geq 1}$, $\|\bv_0(p)\|_0\le
p\beps$, we have. almost surely,
\begin{align}
\lim\inf_{p \to \infty } \langle \bv^+ , \bvz\rangle \geq
\F_{V_{*}}(\S_{V_{*}}(\beta,\alpha) / \sqrt \alpha)~.\label{eq:WorstCaseRecS}
\end{align}
Equality holds if $\bvz(p)$ is the vector with $p\eps_{*}$ non-zero
entries, all equal to $1/\sqrt{p\eps_{*}}$.
\end{theorem}
For the proof we refer to  Section \ref{sec:worstCase} which also
contains explicit expressions to compute  $\eps_{\rec,\#}$.

\subsection{Additional notations}\label{sec:Notations}

We use capital  boldface for matrices, e.g. $\bX$, $\bZ$,\dots and
lowercase boldface for vectors,
e.g. $\bx$ or $\by$. The ordinary scalar product between
$\bx,\by\in\reals^m$ is denoted by $\<\bx,\by\> =
\sum_{i=1}^m\bx_i\by_i$.  The $\ell_p$ norm of a vector is denoted by
$\|\bx\|_p$, and we will occasionally omit the subscript for the case
$p=2$.
The $\ell_2$ operator norm of the matrix $\bX$ is denoted by
$\|\bX\|_2$.

As usual, we write $\phi(x) = e^{-x^2/2}/\sqrt{2\pi}$ for the standard
Gaussian density, and $\Phi(x) = \int_{-\infty}^x \phi(z) \,\de z$ for
the Gaussian distribution function. 
Finally we will say that a function $\psi : \reals^d \to \reals$ is
pseudo-Lipschitz if there exists a constant $L>0$ such that 
\begin{align}
 \big |\psi(\bx) - \psi(\by) \big| \leq L (1+\|\bx\|_2 + \|\by\|_2)
 \|\bx-\by\|_2\, .
\end{align}
%

\section{Upper bounds on non-negative PCA values}\label{sec:UpperBoundSlepianNonnegative}

As mentioned above, Theorems \ref{th:mainSym} and \ref{th:mainRec} 
are proved in two steps. We establish an upper bound on the value of
the optimization problem by using Sudakov-Fernique inequality and prove that the
bound is tight by analyzing an iterative algorithm that solves the optimization
problem.

The first statement concerns the \ref{eq:SymmetricModel}.
\begin{lemma}\label{th:upperBoundsSymmetric}
Consider the \ref{eq:SymmetricModel}, and let $\bv^+=\bv^+(\bX)$ be the
\ref{eq:PositivePCASymm} estimator, with $\lambda^+=\lambda^+(\bX)$ the 
value
of the corresponding optimization problem.

Then, under the assumptions of Theorem \ref{th:mainSym}, we have
\begin{align} \label{eq:upperBoundSymmetricCase}
\lim\sup_{n \to \infty}  \E\,\lambda^+(\bX) \leq  \R\sym(\T_V(\beta))\, .
\end{align}
Further, there exists a deterministic function  $\Delta:\reals_{\ge
  0}\to \reals$, with $\lim_{x\to 0}\Delta(x) = 0$  such that, almost surely,
\begin{equation}\label{eq:vplusvzeroCloseToFTbeta}
\lim\sup_{n\to\infty}\left | \langle \bv^+, \bvz \rangle -
  \F_V(\T_V(\beta)) \right | \leq \Delta\left(
\R\sym(\T_V(\beta)) -
\lim\inf_{n \to \infty} \lambda^+(\bX)\right )\, .
\end{equation}
\end{lemma}

The second statement concern the (non-symmetric) \ref{eq:SpikedMatrix}.
\begin{lemma}\label{th:upperBounds}
Consider the \ref{eq:SpikedMatrix} and let $\bv^+=\bv^+(\bX)$ be the
\ref{eq:PositivePCA} estimator, with $\sigma^+=\sigma^+(\bX)$ the 
value
of the corresponding optimization problem.

Then, under the assumptions of Theorem \ref{th:mainRec}, we have
\begin{align} \label{eq:upperBoundGeneralProposition}
\lim\sup_{n \to \infty} \E\, \sigma^+(\bX) \leq  \R\rec(\S_V(\beta,\alpha) ,
\alpha)\, .
\end{align}
Further, there exists a deterministic function  $\Delta:\reals_{\ge
  0}\to \reals$, with $\lim_{x\to 0}\Delta(x) = 0$  such that, almost surely,
\begin{align}\label{eq:upperBoundNNinnerProductNonSymDense}
\lim\sup_{n \to \infty} \left |\langle \bv^+, \bvz \rangle  - \F_V
  (\S_V(\beta,\alpha) / \sqrt \alpha)\right | \le \Delta\left(
\R\rec(\S_V(\beta,\alpha) ,\alpha) -\lim\inf_{n \to \infty}  \sigma^+(\bX)  \right)\,.
\end{align}
\end{lemma}
The proof of Lemma \ref{th:upperBounds} can be found in Section
\ref{sec:ProofUpperBound}. The proof for the case of symmetric
matrices, cf. Lemma \ref{th:upperBoundsSymmetric},
  is completely analogous and we omit it.

\begin{remark} 
While the above upper bounds are stated in asymptotic form, 
the proofs in Section \ref{sec:proofUpperBounds} imply non-asymptotic
upper bounds. Roughly speaking, the above upper bounds hold
non-asymptotically up to an additive correction of order  $1 / \sqrt{n}$. 
\end{remark}


\section{Approximate message passing algorithm}\label{sec:AMP}

We use an algorithmic approach to prove a lower bound that matches the upper
bound in Lemmas \ref{th:upperBoundsSymmetric}, \ref{th:upperBounds}.
The algorithm is close in spirit to the usual power method that
computes the leading eigenvector of a symmetric matrix $\bX$ by iterating
\begin{align}
\bv^{t+1} = \bX \, \, \bv^t\, ,
\end{align}
from an arbitrary initialization $\bv^0\in\reals^n$. Of course the
power method is not well suited for the present problem, since it does
not enforce the non-negativity constraint $\bv\ge 0$.  We will enforce
this constraint iteratively by projecting on the feasible set. Similar
non-linear power methods were studied previously, for instance in the
context of sparse PCA
\cite{journee2010generalized,yuan2013truncated} and a statistical
analysis of a method of this type was developed in \cite{ma2013sparse}.

Our approach differs substantially from this line of work. We develop
an approximate message passing (AMP) algorithm that builds on ideas
from statistical physics and graphical models 
\cite{DMM09,MontanariChapter}. Remarkably, exact high-dimensional
asymptotics for these algorithms have been characterized in some
generality  using a method known as \emph{state evolution}
\cite{BM-MPCS-2011,BM-Universality}. We establish the desired lower
bounds by applying this theory
to our problem.

As before, we will start by considering the case of symmetric matrices
and then move to rectangular matrices. 
 
%
%

\subsection{Symmetric matrices}

\subsubsection{Algorithm definition}

The AMP algorithm is iterative and, after $t$ iterations, mantains a
state $\bv^t\in\reals^n$. We initialize it with
$\bv^0=(1,1,\dots,1)^{\sT}$, $\bv^{-1}=(0,0,\dots,0)^{\sT}$,  and use
the update rule, for $t\ge 0$,
\begin{align}\tag*{AMP-sym}\label{eq:AMPsym}
 \bv^{t+1}  & = \bX 
  f(\bv^t)   -
\ons_t\,f(\bv^{t-1})\, ,
\end{align}
where  $\ons_t\equiv \|(\bv^t)_+\|_0/\{\sqrt n \|(\bv^t)_+\|_2\}$ 
and
$f:\reals^n\to\reals^n$ is the normalized projection on the
positive orthant:
\begin{align}
f(\bx) = \sqrt n \,\frac{(\bx)_+}{\|(\bx)_+\|_2}\, .
\end{align}
(The factor $\sqrt{n}$ is introduced here for future convenience.)

If we neglect the memory term  $-\ons_t\,f(\bv^{t-1})$, the
algorithm \ref{eq:AMPsym} is extremely simple: It alternates between a
power iteration, and an orthogonal projection onto the constraint set
$\{ \bv~:~\bv\ge0~,~\|\bv\|\leq 1\} $. As proved in \cite{BM-MPCS-2011,BM-Universality} the
memory term (`Onsager term') plays a crucial role in allowing for an
exact high-dimensional characterization.

Note that $\bv^t$ does not satisfy --in general-- the positivity
constraint. Indeed it is not the algorithm estimate of $\bvz$. After any number $t$ of iteration we construct the estimate
\begin{align}
\hbv^t = \frac{(\bv^t)_+}{\|(\bv^t)_+\|_2}\, . \label{eq:AMP_Estimates_Symm}
\end{align}

\subsubsection{Asymptotic analysis}

State evolution \cite{DMM09,BM-MPCS-2011,javanmard2013state,BM-Universality}  is a
mathematical technique that provides an exact
distributional characterization of a class of algorithms that includes
\ref{eq:AMPsym}, under suitable probabilistic models for the matrix
$\bX$. In the present case, we will assume the
\ref{eq:SymmetricModel}, with $\sqrt{n}\bvz$ converging in empirical
distribution to a random variable $\bV$.

Informally, state evolution predicts that as $n\to\infty$, for any
fixed $t\ge 1$, the state vector $\bv^t$ is approximately normal with
mean $\sqrt{n}\tau_t\, \,\bvz$ and covariance $\id_{n\times n}$. In other words,
it can be viewed as a noisy version of the signal $\bvz$:
\begin{align}
\bv^t \approx \sqrt{n}\tau_t\, \bvz +\,\bg\, , \;\;\;\;\;
\bg\sim\normal(0,\id_{n\times n})\, .
\end{align}
The signal-to-noise ratio $\tau_t$ is determined recursively by
letting $\tau_1=\beta\E \bV$ and for all $t\ge 1$, $\tau_{t+1} =
\F_{\bV}(\tau_t)$. Explicitly:
\begin{align}
\tau_{t+1} = \beta\; \frac{ \E~ \bV \left ( \tau_t \bV +\bG\right
  )_+}{\sqrt{\E \left ( \tau_t\bV+\bG\right
    )_+^2}}\, , \label{eq:SymmetricStateEvolution}
\end{align}
with $\bG\sim\normal(0,1)$ independent of $\bV$. A formal statement is
given below.
\begin{proposition}\label{prop:AMPsym}
Consider  the \ref{eq:SymmetricModel}, and assume that $\left \{\sqrt
  n \bvz(n) \right \}_{n \geq 0} $ converges in  empirical
distribution to a random variable  $V$. Further, let $\{\tau_t\}_{t\ge
1}$ be defined by the state evolution recursion (\ref{eq:SymmetricStateEvolution}).

Then, for any pseudo-Lipschitz function  $\psi : \reals^{2} \to
\reals$ and any $t\ge 1$ we have, almost surely,
\begin{equation}\label{eq:limitPsiPseudoLipSym}
 \lim_{n \to \infty} \frac 1n \sum_{i=1}^n \psi( \bv_i^t ,\sqrt n
 (\bvz)_i) = \E \left \{ \psi ( \tau_t V + G , V)  \right \} \, ,
\end{equation}
where $G \sim \normal(0,1)$ is independent of $V$. 
Further,
the convergence in  Eq.~(\ref{eq:limitPsiPseudoLipSym}) also holds for
$\psi(x,y) = \ind(x\le a)$ and any $a\in\reals$.
\end{proposition}
The proof of this result is a direct application of the results of
\cite{BM-MPCS-2011,javanmard2013state}
and can be found in Appendix \ref{app:AMPsym}.

A second important result that follows from state evolution is
that the sequence $\{\bv^t\}_{t\ge 0}$ converges in the following
asymptotic sense.
 \begin{proposition}\label{lem:diffGoesToZeroSymm}
Under the assumptions of Proposition \ref{prop:AMPsym}, fix any
$\ell\ge 0$. Then, we have almost surely
\begin{align}
\lim_{t \to \infty} \lim_{n \to \infty} \frac 1n \|\bv^t -
\bv^{t+\ell}\|^2_2 = 0\, .
\end{align}
\end{proposition}
The proof of this statement is deferred to Appendix \ref{app:DiffGoesToZeroSymm}.

As $t\to\infty$, $\tau_t\to \tau$, with $\tau$ the unique positive
solution of  the fixed point equation $\tau=\beta\F_{\bV}(\tau)$.
By using the above two propositions, we then obtain the following
lower bound, whose proof can be found in Section
\ref{sec:ProofLowerBounds}.
\begin{theorem}\label{th:lowerBoundsAMPsymmetric} 
Consider  the \ref{eq:SymmetricModel}, and assume that $\left \{\sqrt
  n \bvz(n) \right \}_{n \geq 0} $ converges in  empirical
distribution to a random variable  $V$. Further, let $\{\hbv^t\}_{t\ge
  0}$ be the AMP iterates as defined by
\ref{eq:AMPsym}  and Eq.~(\ref{eq:AMP_Estimates_Symm}).
Finally, let $\tau$ be the
unique positive solution of the fixed point equation $\tau =
\beta\F_{\bV}(\tau)$ (equivalently $\tau= \T_{\bV}(\beta)$). 

Then we have, almost surely, 
\begin{align}\label{eq:limAMPRayleighSym} 
\lim_{t \to \infty}  \lim_{n \to \infty}  \langle \hbv^t , \bX \hbv^t
\rangle &=\R\sym(\tau )\, ,\\
\lim_{t \to \infty}
\lim_{n \to \infty}  \langle \hbv^t,\bvz\rangle &=  \F_\bV( \tau ) \label{eq:limAMPinnerFsym}~.
\end{align}
\end{theorem}
This provides the necessary lower bound that complements the upper
bound based on Sudakov-Fernique inequality, cf. Section \ref{sec:UpperBoundSlepianNonnegative}.

\subsection{Rectangular matrices}

\subsubsection{Algorithm definition}
In this case the algorithm keeps track --after $t$ iterations--  of
$\bu^t \in \reals^n$ 
and $\bv^t\in\reals^p$. These are initialized by setting 
 $\bv^0=(1,1,\dots,1)^{\sT}$, $\bu^{-1} =0$,
and updated by letting, for $t\ge 0$, 
\begin{equation}
\tag*{AMP-rec}\label{eq:AMPnonsym}
\left \{ \begin{aligned} 
\bu^{t} =& \bX   f(\bv^t) -  \ons_t    g(\bu^{t-1}) \, ,\\
\bv^{t+1} =&  \bX\trans  g(\bu^{t}) -\onsd_{t}   f(\bv^{t}) \, ,  
\end{aligned}\right.
\end{equation}
where $\onsd_{t} =\sqrt{n}/\|\bu^t\|_2$ and $\ons_t= \|(\bv^t)_+\|_0/( \sqrt n\|(\bv^t)_+\|_2)$
and  $f:\reals^p\to\reals^p$ and $g:\reals^n\to\reals^n$ are defined by:
\begin{align}
 f(\bx) = \sqrt { n} ~\frac{(\bx)_+ } { \|(\bx)_+\|_2}\, ,\quad \quad
g(\bx) =\sqrt n\frac{\bx}{\|\bx\|_2}\, .
\end{align} 

After any number $t$ of iteration we construct the estimates
\begin{align}
\hbu^t = \frac{\bu^t}{\|\bu^t\|_2}\, ,\;\;\;\;\;\;\;\;\;
\hbv^t = \frac{(\bv^t)_+}{\|(\bv^t)_+\|_2}\, . \label{eq:AMP_Estimates_Rec}
\end{align}
These satisfy the normalization and positivity constraints and are
used as estimates of $\buz$, $\bvz$. 

\subsubsection{Asymptotic analysis}

We consider the high dimensional setup where $n\to \infty$, and $p=p(n)\to\infty$ with
converging  aspect ratio 
$p/n\to\alpha \in (0,1)$. We  assume that $\ \{ \sqrt n~ \buz(n)  \}_{n \geq 0} $ converges in empirical distribution to $U$ and $ \{\sqrt {p }~ \bvz(p)  \}_{p \geq 0} $ converges in empirical distribution to $V$.

The high dimensional asymptotics of $\bu^t$, $\bv^t$ is characterized
--as in the symmetric case-- through state evolution. We introduce 
the real-valued \emph{state evolution sequences} $\{ \vartheta_t\}_{t \geq 0}$ and $\{ \mu_t
\}_{t \geq 1}$ through  the following recursion for $t\ge 0$
\begin{equation}\tag*{SE-rec}\label{eq:seqThetaMu}
\left \{ \begin{aligned}
 &  \mu_{t} = \sqrt{  \beta  }~\F_V\left ( \frac {\vartheta_t} {\sqrt
     \alpha} \right )\, ,\\
& \vartheta_{t+1} = \sqrt \beta \frac{\mu_t }{\sqrt{1+ \mu_t^2}}\, , 
\end{aligned}
\right .
\end{equation}
with initial conditions $\mu_0 = \sqrt \beta \E\bV$.
We refer to these as to  the \emph{  state evolution equations}. 
Roughly speaking, state evolution establishes that 
$\bu^t$ is approximately normal with mean $\sqrt{n}\, \mu_t\, \buz$
and unit covariance, and  $\bv^t$ is approximately normal with mean 
$\sqrt{n}\, \vartheta_t\, \bvz$ and unit covariance. This is
formalized below.
\begin{proposition}\label{prop:SEgeneral}
Consider the \ref{eq:SpikedMatrix} and assume that
$\ \{ \sqrt n~ \buz(n)  \}_{n \geq 0} $ converges in empirical
distribution to a random variable  $U$ and $ \{\sqrt {p }~ \bvz(p)
\}_{p \geq 0} $ converges in empirical
distribution to a random variable $V$. Further, let $\{\mu_t\}_{t\ge
  0}$, $\{\vth_t\}_{t\ge 1}$ be defined by the state evolution
recursion \ref{eq:seqThetaMu}.

Then, for any pseudo-Lipshitz function $\psi : \reals^2 \to \reals$
and any  $t\ge 1$ we have, almost surely 
\begin{equation}\label{eq:limitPsiPseudoLipRec}
\left \{\begin{aligned}
& \lim_{n \to \infty} \frac 1n \sum_{i=1}^n \psi( \bu_i^t ,\sqrt n (\buz)_i) = \E \left \{ \psi ( \mu_t U + G , U)  \right \} \\
& \lim_{p \to \infty} \frac 1p \sum_{i=1}^p \psi(\bv_i^t ,\sqrt p (\bvz)_i) = \E \left \{ \psi ( \vartheta_t/\sqrt{\alpha}\, V + G , V)  \right \}
\end{aligned}
\right .
\end{equation}
where $G \sim \normal(0,1)$ is independent of $U$ and $V$. Further,
the convergence in  Eq.~(\ref{eq:limitPsiPseudoLipRec}) also holds for
$\psi(x,y) = \ind(x\le a)$ and any $a\in\reals$.
 \end{proposition}
The proof is very similar to the one of Proposition \ref{prop:AMPsym}
and is again a direct application of the results of
\cite{BM-MPCS-2011,javanmard2013state}. We omit it to avoid redundancy.

We also have an analogous of Proposition \ref{lem:diffGoesToZeroSymm}.
 \begin{proposition}\label{lem:diffGoesToZero}
Under the assumptions of Proposition \ref{prop:SEgeneral}, let
$\ell\ge 0$ be a fixed integer. Then we have, almost surely,
\begin{align}
\lim_{t \to \infty} \lim_{n \to \infty} \frac 1n \|\bv^t -
\bv^{t+\ell}\|_2 = 0,\quad
\quad\lim_{t \to \infty} \lim_{n \to \infty} \frac 1n \|\bu^t - \bu^{t+\ell}\|_2 = 0\,.
\end{align}
\end{proposition}
We omit the proof, as it is very similar to the one of Proposition \ref{lem:diffGoesToZero}.

In the limit $t\to\infty$ (and assuming $\eps>0$), the sequence
defined in \ref{eq:seqThetaMu} converges to a nonzero  fixed point
$(\mu,\vartheta)$ satisfying the fixed point equations
\begin{equation}
\left \{ \begin{aligned}
 &  \mu = \sqrt{  \beta  }~\F_V\left ( \frac {\vartheta} {\sqrt
     \alpha} \right )\, ,\\
& \vartheta= \sqrt \beta \frac{\mu}{\sqrt{1+ \mu^2}}\, .
\end{aligned}
\right .\label{eq:FixedPointRectangular}
\end{equation}
We will prove that these equations admit a unique positive solution. 

Considering $t\to\infty$ (after $n\to\infty$) we can thus prove the following.
\begin{theorem}\label{th:lowerBoundsAMPnonsymmetric} 
Consider the \ref{eq:SpikedMatrix} and assume that
$\ \{ \sqrt n~ \buz(n)  \}_{n \geq 0} $ converges in empirical
distribution to a random variable  $U$ and $ \{\sqrt {p }~ \bvz(p)
\}_{p \geq 0} $ converges in empirical
distribution to a random variable $V$. 
Further, let $\{\hbu^t\}_{t\ge 0}$, $\{\hbv^t\}_{t\ge 0}$ be the AMP
estimates as defined by \ref{eq:AMPnonsym} and Eq.~(\ref{eq:AMP_Estimates_Rec})
Finally, let $(\mu,\vth)$ be the only positive solution of  
the fixed point equations (\ref{eq:FixedPointRectangular}).

Then we have, almost surely, 
\begin{align}\label{eq:limAMPRayleighNonSym}%
\lim_{t \to \infty}  \lim_{n \to \infty} \langle \hbu^t , \bX \hbv^t
\rangle &= \R\rec(\vartheta,\alpha)\, ,\\
\lim_{t \to \infty}  \lim_{n \to \infty}  \langle \hbu^t,\buz\rangle &
= \frac{\mu }{\sqrt{1+\mu ^2}} \, ,\label{eq:AMPScalProdURect}\\
 \lim_{t \to \infty}  \lim_{n \to \infty}  \langle \hbv^t,\bvz\rangle
 &=  \F_V\left ( \frac{\vartheta }{\sqrt \alpha} \right)\, . \label{eq:AMPScalProdVRect}
\end{align}
\end{theorem}
The proof of this theorem can be found in Section \ref{sec:ProofLowerBounds}.

\subsection{Computational complexity}

As a direct consequence of the characterization of AMP established
in Propositions \ref{prop:AMPsym} and \ref{prop:SEgeneral}, 
we can upper bound the number of iterations needed for Algorithms
\ref{eq:AMPnonsym} and \ref{eq:AMPsym} to converge.
We point out that the cost of each step of the AMP algorithms is dominated by a matrix vector multiplication. This operation can easily be parallelized and performed efficiently.

To be definite, we state the next result in the case of symmetric
matrices. A completely analogous statement holds for rectangular
matrices.
\begin{proposition} \label{prop:geometricConvergence} 
For any law $\mu_V\in\cP$ and any $\delta>0$ there exists a constant $t_0(V,\delta)<\infty$ such that the following holds true.
Under the assumptions of Proposition \ref{prop:AMPsym}, let
$\{\hbv^t\}_{t\ge 0}$ be the sequence of estimates produced by
AMP. Then,  for all fixed $t\ge t_0$ we have
\begin{align}
\lim_{n\to\infty}\prob\Big(\<\hbv^t,\bX\hbv^t\>\ge
(1-\delta)\max_{\bv\ge 0, \|\bv\|=1}\<\bv,\bX\bv\>\Big) = 1\, .
\end{align}
\end{proposition}
The proof of this statement follows immediately from Theorem
\ref{th:mainSym} and
\ref{th:lowerBoundsAMPsymmetric}. A more careful treatment of error
terms in the latter can be used to show that --indeed-- $t_0(V,\delta)
\le C(V)\log(1/\delta)$ for some finite constant $C(V)$.

Notice that the computational cost of AMP is dominated by the one of matrix
vector multiplications, call it $T_\text{mult}$. 
The above discussion indicates that  the average-case complexity of the algorithms
\ref{eq:AMPnonsym} and \ref{eq:AMPsym} is
 $O(T_\text{mult} \log 1/\delta)$.

\section{Numerical illustration}
\label{sec:Numerical}

We carried out numerical simulations on synthetic data generated
following \ref{eq:SymmetricModel}. 
We use a signal $\bvz$ that takes two values:
\begin{align}
(\bvz)_i = \begin{cases}
\frac{1}{\sqrt{n\eps}} & \mbox{ if $i\in S$,}\\
0 & \mbox{otherwise,}
\end{cases}
\end{align}
where $S\subseteq [n]$ is of size $|S|=n\eps$. 
It is immediate to see that the sequence $\{\sqrt{n}\bvz(n)\}_{n\ge
  0}$ converges in empirical distribution to a random variable with
distribution
\begin{align}
 V = \begin{cases}
\varepsilon^{-1/2}\quad & \text{with probability}\quad \varepsilon, \\ 
 0  \quad&  \text{with probability}\quad 1-\varepsilon\, \end{cases} 
\end{align}
In other words $\mu_V$ is the 2-points mixture
$\mu_V= (1-\eps)\delta_0 +\eps\, \delta_{1/\sqrt{\eps}}$.

The predictions of Theorem  \ref{th:mainSym} are stated in terms of
the function $\F_V\equiv\F_{\eps}$ that is rather explicit in this case. We have
\begin{align}\label{eq:FVepsilon} 
\F_{\eps}(x) &= 
\frac{\eps\,  B(x/\sqrt{\eps})/x}
{\sqrt{(1-\eps)/2+\eps(B(x/\sqrt{\eps})+\Phi(x/\sqrt{\eps}))} }\, ,\\
B(w) & \equiv w^2\Phi(w)+w\,\phi(w)\, .
\end{align}

\subsection{Comparison with classical PCA}

We implemented the algorithm \ref{eq:AMPsym}, and report in Figure
\ref{fig:symSpik} the results of numerical simulations with $n =10\,000$,
 sparsity level $\varepsilon \in\{ 0.001, 0.1,
0.8\}$, and signal-to-noise ratio $\beta\in\{0.05,0.10,\dots,1.5\}$.
In each case we run AMP for $t=50$ iterations and plot the empirical
average of $\<\hbv^t,\bvz\>$ over $32$ instances. 
The algorithm convergence is fast and --for our purposes-- this value
of $t$ is large enough so that $\tau_t\approx \T_V(\beta)$ and
$\hbv^t\approx \bv^+$. (See below for further evidence of this point.)

The results agree well with the asymptotic predictions of
Theorem \ref{th:mainSym}, namely with the curves reporting
$\F_V(\T_V(\beta))$.
The figure also illustrates that sparse vectors (small $\eps$)
correspond to the least favorable signal in small signal-to-noise
ratio.
The value $\beta = 1/\sqrt{2}$ corresponds to the phase transition.
\begin{figure}\label{fig:symSpik}
\begin{center}
\includegraphics[width = 12cm]{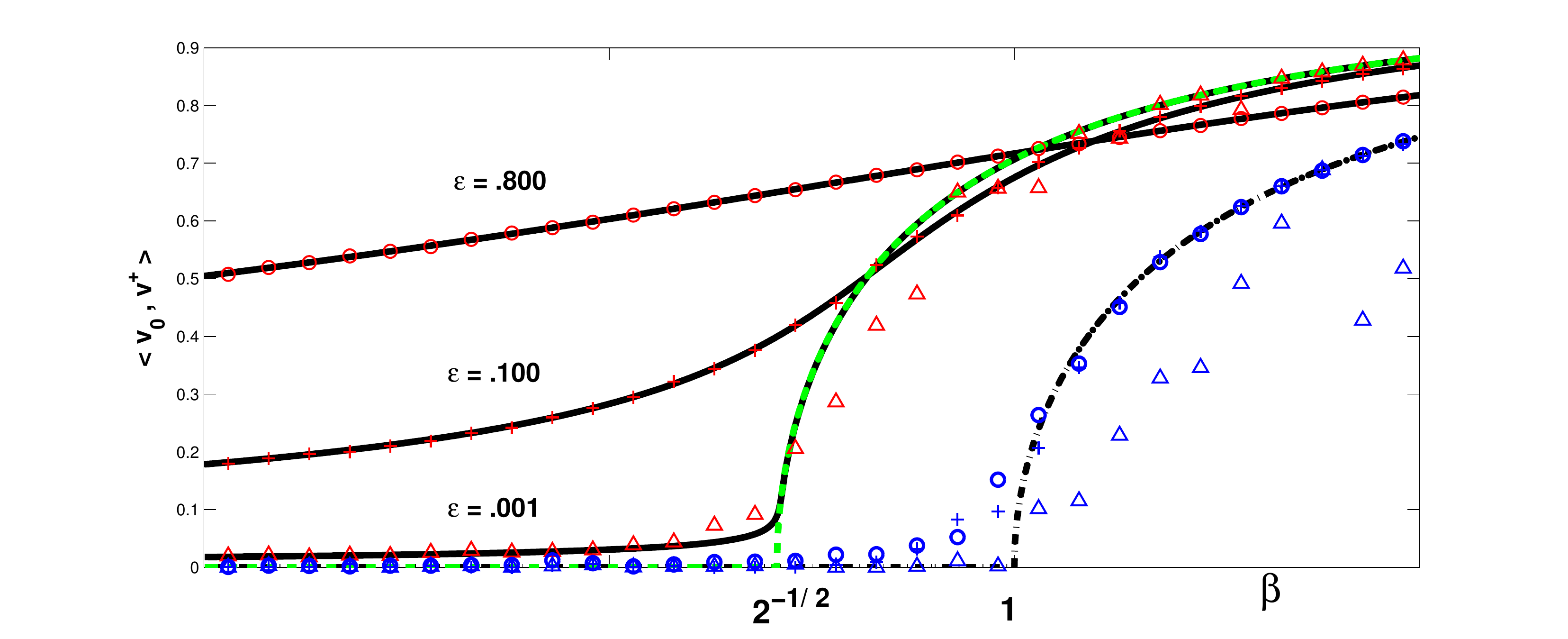}
\caption{Numerical simulations  with the Symmetric Spiked Model. Black
  lines represent the theoretical predictions of Theorem
  \ref{th:mainSym}, and dots represent empirical values of 
$\langle \hbv^t , \bvz \rangle$ for the AMP estimator  (in red) and
$\langle \bv_1 , \bvz \rangle$ for \ref{eq:SamplePCA} (in blue). The dashed
red line corresponds to the limit behavior as $\varepsilon \to 0$. 
 In the right hand side of the plot, blue dots and dashed black line
 correspond to $\langle \bv_1 , \bv_0 \rangle$ and the theoretical
 prediction }
\end{center}
\end{figure}

\subsection{Deviation from the asymptotic behavior}

Theorem \ref{th:mainSym} and
Proposition \ref{prop:AMPsym} predict the value of $\< \bvz , \bv^+\>$
and $\< \bvz , \hbv^t\>$ in the limit $n
\to \infty$. It is natural to question the  validity of the prediction for
moderate values of  $n$. 

In order to investigate this point, we performed numerical experiments
with AMP
by generating instances of the problem for several values of $n$
and compared the results  with the asymptotic
prediction of Eq. (\ref{eq:FVepsilon}).
The top left-hand frame   in Figure \ref{fig:asymptotic}
is obtained with $n = 50, 500, 5000$, $\eps = 0.05$ and several value
of $\beta$. For each point we plot the average of $\<\hbv^t,\bvz\>$
after $t=60$  iterations, over 32 instances.
 Already at $n=500$ the agreement is good, and improving with $n$. 

In the top-right plot we plot the deviation between the empirical
averages of $\<\hbv^t,\hbv\>\approx \<\hbv^+,\bvz\>$ (over 32 instances) and the
asymptotic prediction $\F_V(\T_V(\beta))$.
The data  suggest
\begin{align}
 \<\hbv^+,\bvz\> \approx \F_V(\T_V(\beta)) +\, A\, n^{-b}\, ,
\end{align}
with $b\approx 0.5$.

In the bottom frames we plot the theoretical and empirical (for $n =
1000$) values of $\<\hbv^t,\bvz\>$ for a grid of parameters $\beta,\eps$. 
The difference between the two has average $1\cdot10^{-3}$ and standard deviation $3\cdot 10^{-2}$. 
 
\begin{figure}
\begin{center}
\includegraphics[width = 8cm]{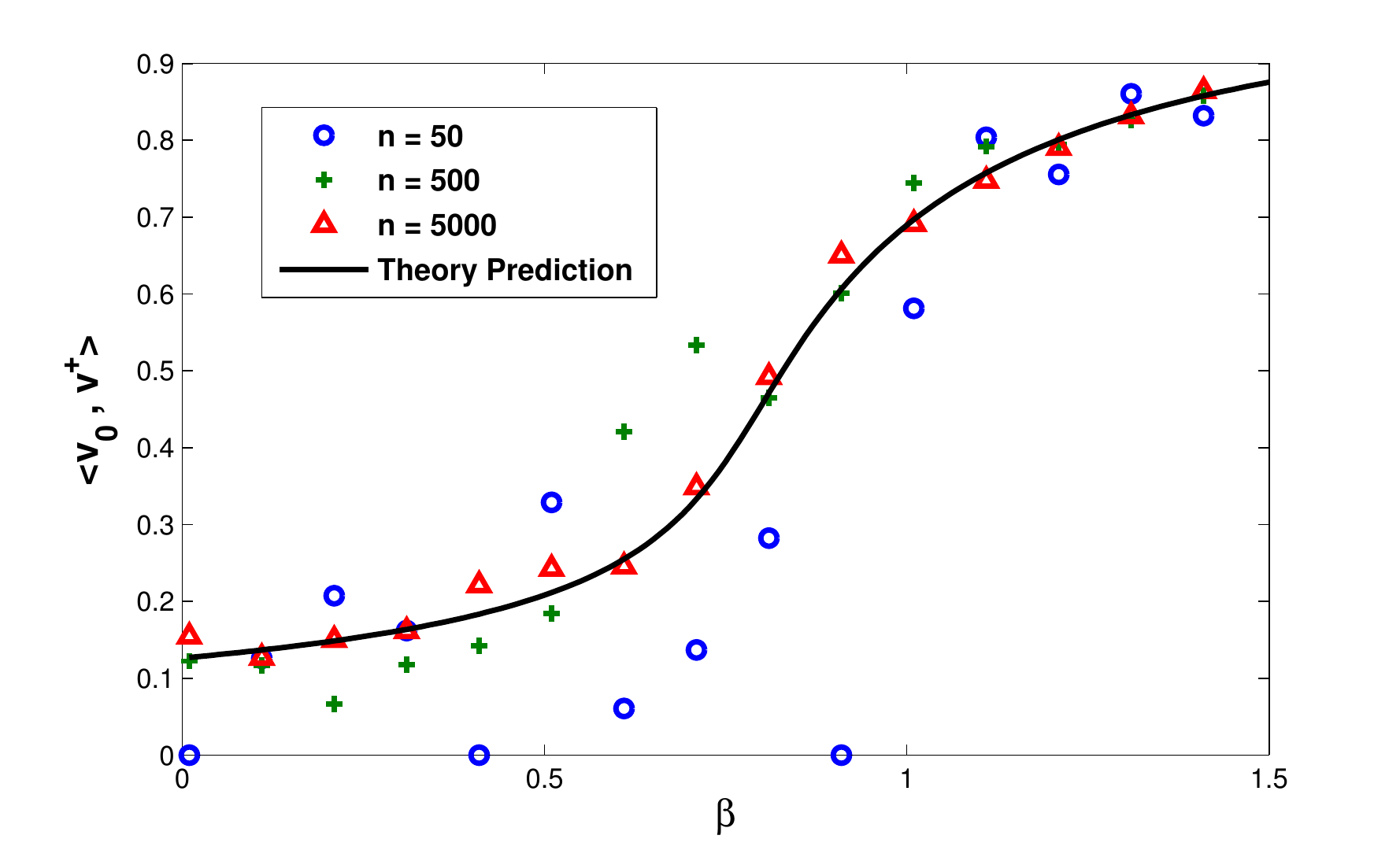} \includegraphics[width = 8cm]{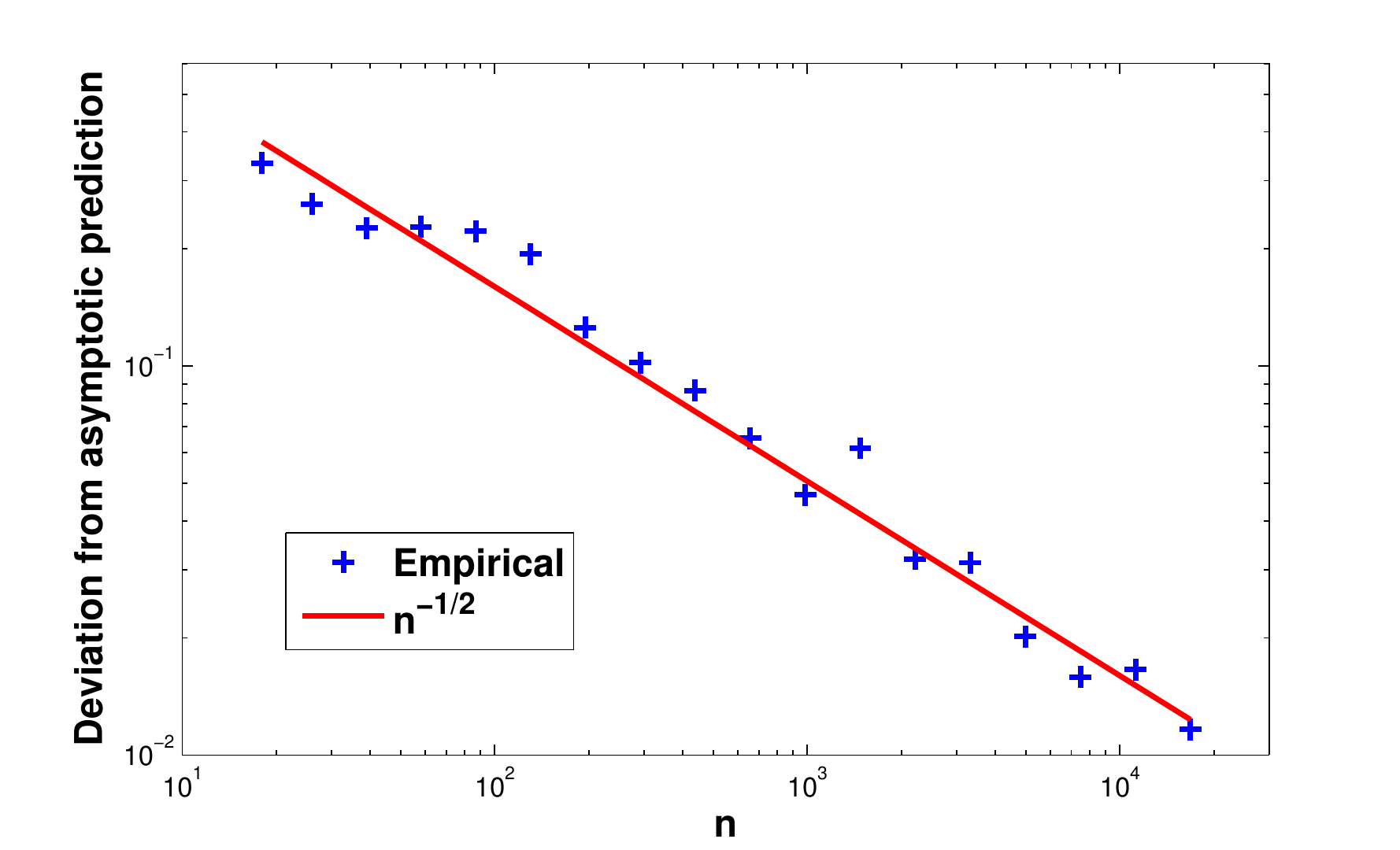}\\
\includegraphics[width = 8cm]{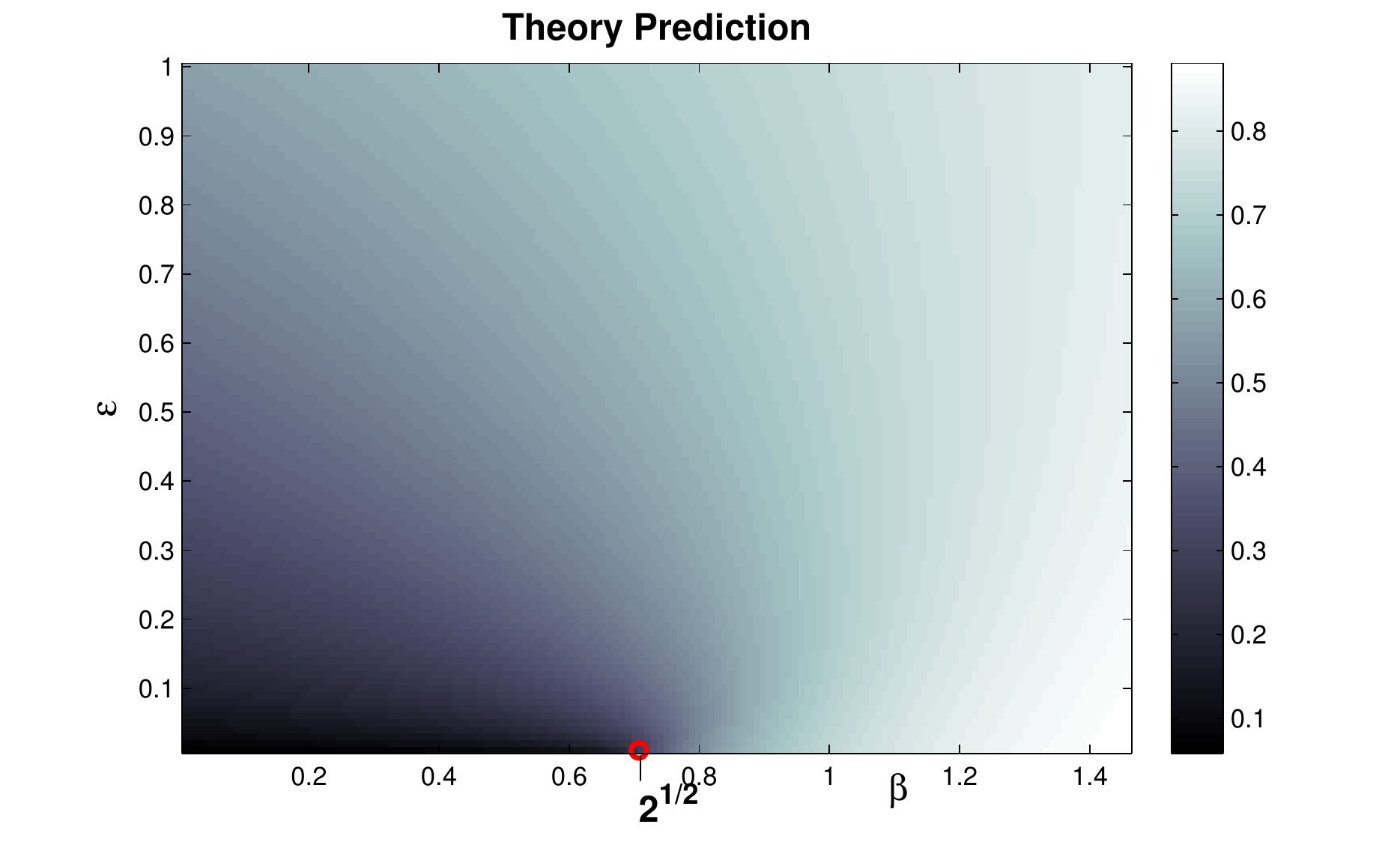} \includegraphics[width = 8cm]{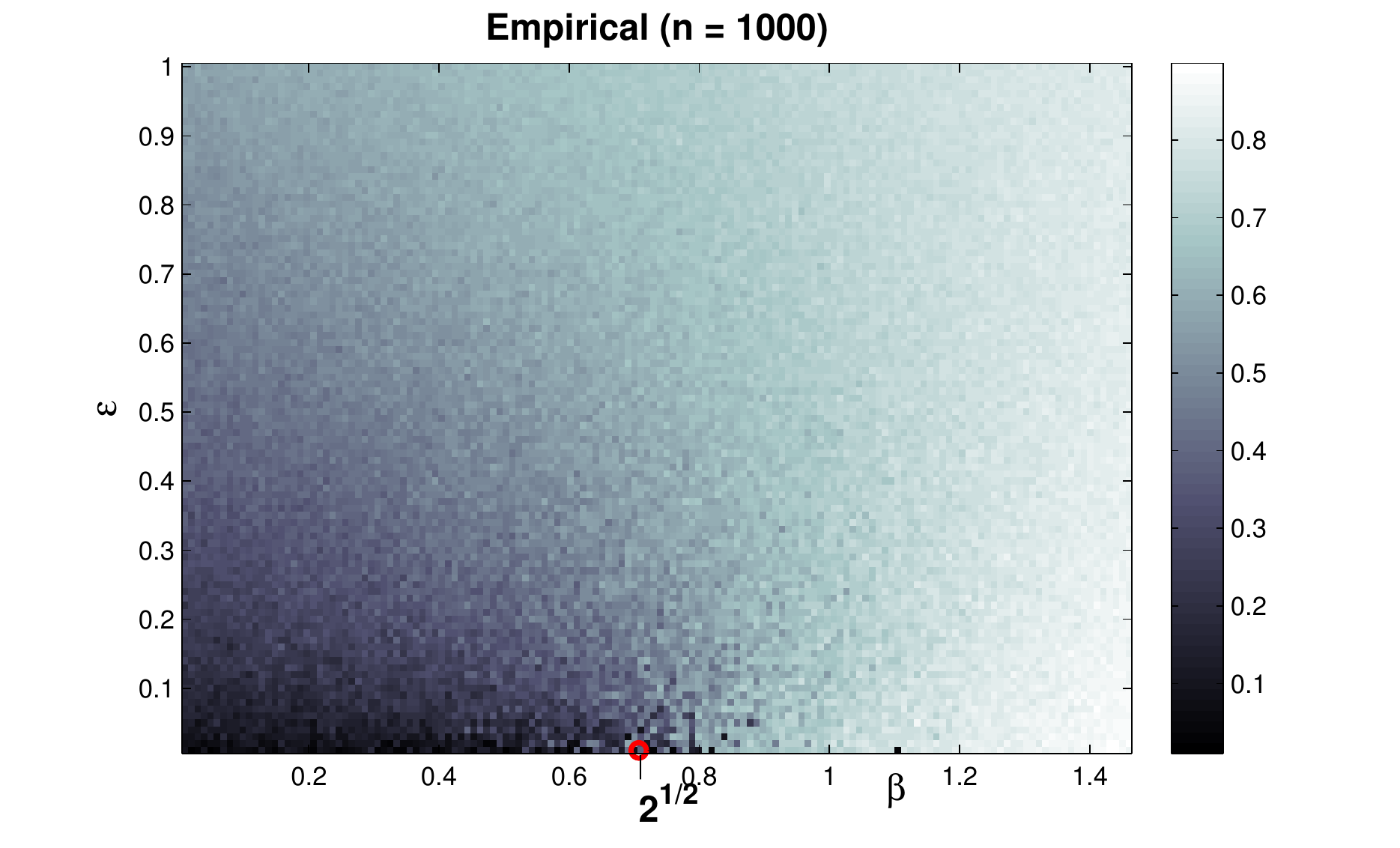}
\caption{Comparison of theoretical prediction and empirical results
  for $\<\hbv^t,\bvz\>\approx \<\bv^+,\bvz\>$ for moderate values of
  $n$ (see main text).}\label{fig:asymptotic}
\end{center}
\end{figure}

\subsection{Comparison with a convex relaxation}

A natural convex relaxation for the \ref{eq:PositivePCASymm} problem
is the semi-definite program 
\begin{equation}\tag*{SDP}\label{eq:SDPrelaxation}
\begin{aligned}
\text{maximize}&~~ \< \bX , \bW\> \, ,\\
 \text{subject to} &~~ \bW \succeq 0\, , \\
  & ~~\text{Trace}(\bW) = 1 \, ,\\
   &~~ \bW \geq 0\, .
\end{aligned}
\end{equation}
It is known \cite{burer09} that for $n \geq 5$ the completely positive cone is strictly included in the doubly non-negative cone 
$$\text{conv} \left \{ \bv \bv^\sT~:~\bv\in \reals_{\geq 0}^n\right \}
\subsetneq \{ \bW~:~\bW \geq 0~,~\bW \succeq 0\}~.$$
Hence in general this relaxation is not tight.
The solution is a symmetric non-negative matrix $\hat \bW$. We extract
the leading eigenvector $\bv_1(\hat \bW)$ and use its positive part as
our approximation for $\bv^+$.

In simulations we use CVX \cite{CVX}  to solve
\ref{eq:SDPrelaxation}, and compare the result to the output of AMP
stopped after $t =
50$ iterations. 
The interior point solver  of CVX forces us to consider small
problems. We use $n= 50$, $\beta = 1/\sqrt 2$, $\eps = 0.3$, and
average over $50$ instances.

On a 2.8 GHz Core 2 Duo with 8GB of RAM, CVX stops after about $40$ seconds and a Matlab
implementation of AMP after $2$ ms.  
On average, the convex relaxation method achieves scalar product
$\E\<\bvz,\bv_1(\hat \bW)_+ \>  = 0.54\pm 0.02$,
while denoting by $\bv^+_{\text{AMP}}$ the output of AMP, we obtain
$\E\<\bvz,\bv^+_{\text{AMP}}\> = 0.55 \pm 0.02$.  
In Figure \ref{fig:AMPvsCVX} we compare the values reached by each
algorithm over the 50 instances of the experiment with the 
predicted asymptotic value value $\F_V(\T_V(1/\sqrt 2)) \approx 0.53$.
 The plot suggests that indeed both methods solve to high accuracy the same problem.
\begin{figure}
\begin{center}
\includegraphics[width = 12cm]{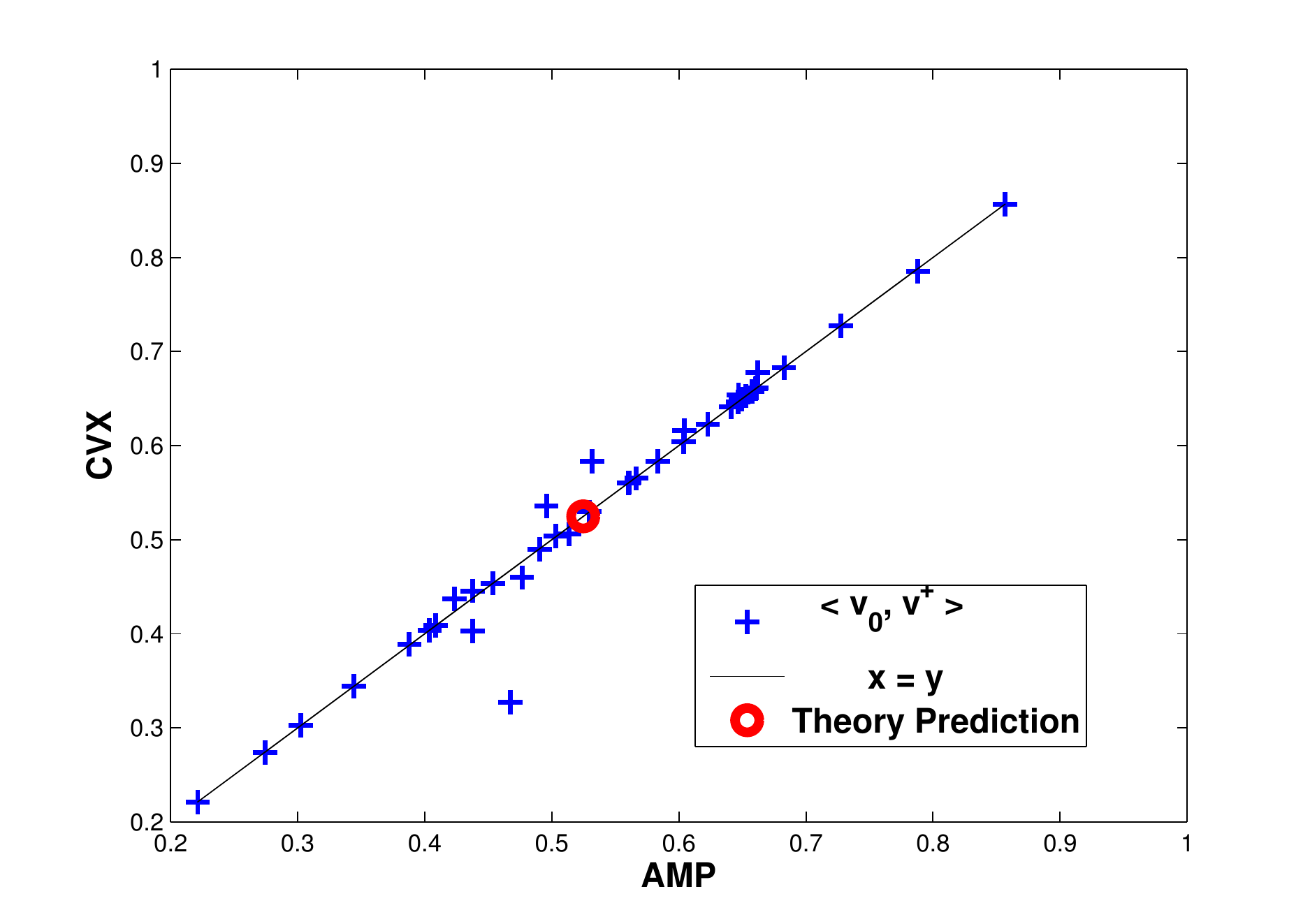}
\caption{Comparing the AMP estimator with the estimator obtained by
  convex relaxation. We plot $\<\bvz,\bv_1(\hat \bW)_+ \>$  (the
  correlation achieved by convex optimization) versus $\<\bvz,\bv^+_{\text{AMP}} \>$  (the
  correlation achieved by AMP), for $50$ random instances.}\label{fig:AMPvsCVX}
\end{center}
\end{figure}

%
%
%


\section{Proofs }\label{sec:proofUpperBounds}

Given a random variable $V$, with $\E(V^2)<\infty$, it  is useful to
define the function 
$\D_V:\reals\to\reals$, by
\begin{align}
\D_V(x) = \E\{(xV+G)_+^2\}\, .\label{eq:DDefinition}
\end{align}

\subsection{Preliminaries}\label{sec:preliminaryProof}

In this section we establish several useful properties of 
the functions $\F_V$, $\G_V$, $\R\sym$, $\R\rec$  introduced in 
Definition \ref{def:functionsFGRTS}. Throughout $V$ is a random
variable with law $\mu_V$ supported on $\reals_{\ge 0}$ and such that 
$\E(V^2) =\int x^2 \,\mu_V(\de x) = 1$. 
Note that, in particular, $V\neq 0$ with strictly positive probability. As before, we let $\eps = \eps(V) =
\prob(V\neq 0)$.

All statements concern
these functions in their domain, namely 
$\F_V$, $\G_V$, $\R\sym:\reals_{\ge 0}\to \reals$, $\R\rec:\reals_{\ge
0}\times\reals_{>0}\to\reals$. Given a function $x\mapsto f(x)$, we
will use $f'(x)$, $f''(x)$ to indicate its first and second derivatives.
\begin{lemma}\label{ref:FGvaluesMonotony} 
Both $\F_\bV$ and $\G_\bV$ are strictly positive, differentiable   and upper bounded by 1.
Further $\F_V$ is strictly  increasing on $\reals$, with $\F_V'(x)>0$ for all $x\in \reals $,   $\G_V$ strictly decreasing on $\reals_{\geq 0}$, with $\G_V'(x)<0$ for all
$x\ge 0$, and $\D_V$ is strictly convex on $\reals$.

Finally  $\F_V(0) = \E V /\sqrt{\pi}$.
and therefore $\F_V(0) \in (0, \sqrt {(\eps/\pi)}]$ and $\G_V(0) = 1 /
\sqrt 2$, $\lim_{x \to +\infty}\F_V(x) = 1$,
$\lim_{x\to-\infty}\F_V(x) = 0$, and $\lim_{x \to \infty} \G_V(x) = 0$. 
  \end{lemma}
\begin{proof}[Proof of Lemma \ref{ref:FGvaluesMonotony}] 
Positivity is immediate from the definition. The upper bound $1$
follows Cauchy-Schwarz inequality. 
To prove differentiability, we write $\F_V(x) = \Y(x) / \sqrt{\D_V(x)}$ 
with
\begin{align}
\Y(x) = \E V(xV+G)_+\, .
\end{align}
both differentiable (by dominated convergence) since  $V$ and $G$ have
bounded second moments, and strictly positive. Therefore  $\F_V$ is
differentiable. 

A direct calculation yields the following relations
\begin{align}
\frac{\de\D_V}{\de x}(x) &= 2\Y(x)\, , \\
\F_V(x) &= \frac{\de \phantom{x}}{\de x}\sqrt{\D_V(x)}\, ,\label{eq:FDerivative}\\
\frac{\de \Y}{\de x}(x)&= \E\{V^2 \one_{x V+G>0} \}\, ,\\
\D_V(x)\frac{\de \F_V}{\de x}(x) & = \D_V(x) \frac{\de\Y}{\de
  x} (x) -\frac{1}{2}\Y(x) \frac{\de}{\de x} \D_V(x)\, .
\end{align}
Using the last expression (and substituting the previous ones), we see that, 
to prove that $\F_V$ is increasing, it is sufficient to prove that 
\begin{align}
\left \{ \E V(xV+G)_+ \right \}^2 < \left (\E V^2 \one_{x V+G>0}
\right )  \left ( \E(xV+G)_+^2 \right )~~,
\end{align}
which directly follows from Cauchy-Schwarz inequality, even for $x<0$, and 
equality can not hold as $V$ and $G$ are independent. 

In order to show that $\G_V$ is decreasing on $\reals_{\geq 0}$ first observe that  for any
$x>0 $, $x~\F_\bV(x) + \G_\bV(x) = \sqrt{\D_V(x)}$.
Differentiating with respect to $x$ and using
Eq.~(\ref{eq:FDerivative}), we get
\begin{align}
x  \frac{\de\F_V}{\de x}(x) = - \frac{\de\G_V}{\de x} (x)\, .
\end{align}
Since $\F_V$ is strictly increasing, it follows that  $ \G_V$ is
strictly decreasing.

Finally, the values at $x=0$ are obtained by simple calculus. The
limits as $x\to \pm\infty$ follow by applying dominated convergence
both to the numerator and to the denominator of $\F_V(x)$ (or
$\G_V(x)$),  after dividing both by $x$.
\end{proof}

\begin{lemma}\label{lem:FGinnerproducts} 
Let $n,p \in \naturals$, $\bg \sim  \normal(0,\id_n), \bh \sim  \normal(0,\id_p)$ and, for each integer $p$, let $\bvz(p)\in
\reals^p$ be a deterministic vector with $\|\bvz(p)\|_2=1$ and such that $\{\sqrt{p}\,
\bvz(p)\}_{p\ge 0}$ converges  in empirical distribution to $\bV \in
\cP$. Similarly, for an integer $n$ let $\buz(n)\in
\reals^n$ be a deterministic vector such that $\|\buz(n)\|_2=1$ and $\{\sqrt{n}\,
\buz(n)\}_{n\ge 0}$ converges  in empirical distribution to $U$ with
$\E~ U^2 = 1$. Then, for any $b\in \reals$ there exists a sequence
$\{\delta_n(b)\}$, with $\delta_n(b)\to 0$ as $n\to\infty$ such that 
\begin{align}
\E\Big\{\Big\|\frac{1}{\sqrt{n}}\bg+b\,\buz
   \Big\|_2\Big\} &\le \sqrt{1+b^2}\, ,\label{eq:FirstScalar}\\
\E\Big\{\Big\|\Big(\frac{1}{\sqrt{p}}\bh+b\bvz\Big)_+\Big\|_2\Big\}& \le \sqrt{\D_V(b)}+\delta_n\, . \label{eq:SecondScalar}
\end{align}
\end{lemma}
\begin{proof}[Proof of Lemma \ref{lem:FGinnerproducts}]
For Eq.~(\ref{eq:FirstScalar}) note that
\begin{align}
\E\Big\{\Big\| \frac{1}{\sqrt{n}} \bg+b\,\buz
\Big\|_2\Big\}&\le 
\sqrt{\E\Big\{\Big\|\frac{1}{\sqrt{n}}\bg+b\,\buz
\Big\|_2^2\Big\}}\\
&= \sqrt{1+b^2\|\buz\|_2^2} = \sqrt{1+b^2}\, .
\end{align}

In order to prove Eq.~(\ref{eq:SecondScalar}), first note that
\begin{align}
\E\Big\{\Big\|\Big(\frac{1}{\sqrt{p}}\bh+b\bvz\Big)_+\Big\|_2\Big\}^2
& \le
\E\Big\{\Big\|\Big(\frac{1}{\sqrt{p}}\bh+b\bvz\Big)_+\Big\|_2^2\Big\}\, .
\end{align}
We then introduce the notation
$K(x) = \E\{(x+G)_+^2\}=(1+x^2)\Phi(x)+x\phi(x)$ and $H(x) = K(x)
-x_+^2$, and $\mu_p = \mu_{\bvz\sqrt{p}}$ for the empirical distribution
of $\{(\bvz))_i\sqrt{p}\}$.
Note that we get
\begin{align}
\E\Big\{\Big\|\Big(\frac{1}{\sqrt{p}}\bh+b\bvz\Big)_+\Big\|_2^2\Big\}& = \frac{1}{p}\sum_{i=1}^pK\big(b(\bvz)_i\sqrt{p}\big)\\
& = b^2 + \int H(b\,v) \mu_p(\de v)\, ,
\end{align}
and
\begin{align}
\Big|\E\Big\{\Big\|\Big(\frac{1}{\sqrt{p}}\bh+b\bvz\Big)_+\Big\|_2^2\Big\}-\D_V(b)\Big|&=
\Big|\int H(b\, v)\, \mu_p(\de v) - \int H(b\, v)\, \mu_V(\de v)
\Big|\, .
\end{align}
Since $x\mapsto H(x)$ is bounded and Lipschitz
continuous on $\reals$, and by assumption $\mu_p$ converges weakly to
$\mu$, the last expression tends to $0$ as $p\to\infty$, which proves
our claim.
\end{proof}

\begin{lemma}\label{lem:Twelldefined}
Each of the equations 
\begin{align}
\beta &= \frac{x}{\F_V(x)} \, ,\label{eq:FixPointSymm} \\
\beta & = \frac{x \sqrt{ 1+  \beta  \F_\bV(x / \sqrt \alpha )^2} }{
  \F_\bV( x / \sqrt \alpha )} \, ,\label{eq:FixPointRec} 
\end{align}
admits a unique non-negative solution for each $\alpha,\beta>0$,
which we denote by $ \T_V(\beta)$ (for Eq.~(\ref{eq:FixPointSymm}))
and $\S_V(\beta,\alpha)$ (for Eq.~(\ref{eq:FixPointSymm})).

Further, we have
\begin{align}
\frac{\de\F_V}{\de x} (\T_V(\beta)) \in (0,1 / \beta)  \,.\label{eq:ClaimDerivative}
\end{align}
\end{lemma}
\begin{proof}[Proof of Lemma \ref{lem:Twelldefined}]
Let us define the function $\q : x \mapsto \q(x) = \F_V(x)/x$. We
already know (by Lemma \ref{ref:FGvaluesMonotony}) that $\F_V(0)>0$,
so $\lim_{x \to 0} \q(x) = \infty$. 
Also, since  $\F_V(x)\le 1$, we have $\lim_{x \to \infty}\q(x) =
0$. Further $\F_V$ is differentiable and hence so is $\q$ on
$(0,\infty)$. It is therefore sufficient to prove that $\q$ is
strictly decreasing to prove existence and uniqueness of the solution
of Eq.~(\ref{eq:FixPointSymm}).

Recall that  (cf. Eq.~(\ref{eq:FDerivative})):
\begin{align}
\F_V(x) = \frac{\de\phantom{x}}{\de x} \sqrt{\D_V(x)}\quad\text{where}\quad \D_V(x) = \E\{
(xV +G)_+^2\}~. 
\end{align}
We will prove that $z \mapsto \D_V(\sqrt z)$ is concave. This
implies that $\q$ is decreasing: indeed, by the last equation we have
 $$ \q(x) = 2 \frac{ \de\phantom{x^2}}{ \de (x^2)} \sqrt{\D_V(x)}~. $$
Applying the change of variable $x = \sqrt z$, we get 
\begin{align}
\frac{\de\phantom{x}}{\de x}\q(x) = 4x
\frac{\de^2\phantom{x^2}}{\de(x^2)^2} \sqrt{\D_V(x)} = 
4 \sqrt z \frac{\de^2\phantom{z}}{\de z^2} \sqrt{\D_V(\sqrt z)}  =
2\sqrt z \left ( \frac{\frac{\de^2\phantom{z}}{\de z^2}\D_V(\sqrt z)}{\sqrt{\D_V( \sqrt z)}} - 
\frac{(\frac{\de\phantom{z}}{\de z}\D_V(\sqrt z))^2}{2\D_V(\sqrt z)^{3/2}}  \right ) ~~. 
\end{align}
%
This shows that the derivative of $\q(x)$ is strictly negative
provided that $\frac {\de^2\phantom{z}}{\de z^2} \D_V(\sqrt z)$ 
is non-positive, or $z \mapsto \D_V(\sqrt z)$ concave. Indeed the second
term in the last expression is strictly negative because $\frac
{\de \phantom{x}}{\de x} \D_V(x)>0$, cf. Lemma
\ref{ref:FGvaluesMonotony} and Eq.~(\ref{eq:FDerivative})

 We can write $\D_V(\sqrt z)$ as
$$\D_V(\sqrt z) =  \int \E_G \left \{ (\sqrt z v + G)_+^2\right \} \, 
\de\mu_V(v)~,$$
(where $\E_G$ denotes expectation with respect to $G\sim\normal(0,1)$) 
which, since $v\ge 0$ shows that our claim follows from concavity of
$z \mapsto \E_G \left \{ (\sqrt z  + G)_+^2\right \} \equiv K(\sqrt z)$, see Lemma \ref{lem:KsqrtcCCV}.

\begin{lemma}\label{lem:KsqrtcCCV}
The function $z\mapsto K(\sqrt{z})$ is concave on $\reals_{\ge 0}$.
\end{lemma}
\begin{proof} We have $ \E\{(G+x)_+^2\} = (x^2 + 1) \Phi(x) + x \phi(x)$ and $K(\sqrt z) = (z+1) \Phi(\sqrt z) + \sqrt z \phi(\sqrt z)$ so 
$$\frac{\de}{\de z} K(\sqrt z) = \Phi(\sqrt z) + \frac{1}{\sqrt z} \phi(\sqrt z)\quad \Rightarrow \quad \frac{\de^2}{\de z^2} K(\sqrt z) = - \frac 12 z^{-3/2} \phi(\sqrt z) < 0~.$$
\end{proof}
This concludes the proof that Eq.~(\ref{eq:FixPointSymm}) admits a
unique positive solution.

Consider now existence and uniqueness of solutions of Eq.~(\ref{eq:FixPointRec}). 
Note that this is equivalent to proving that for every $\beta>0$  there exists a unique $x>0$ such that 
$$ \frac {\sqrt{\alpha}}{\beta} = \q(x / \sqrt \alpha)\,  \frac{1}{\sqrt{1+\beta \F_V(x / \sqrt \alpha)^2}}~~.$$
We know that $\F_V$ is an increasing function, so $x \mapsto
1/\sqrt{1+\beta \F_V(x / \sqrt \alpha)^2}$ is a decreasing function
taking positive values. The result follows by using monotonicity of $\q$.

In order to prove Eq.~(\ref{eq:ClaimDerivative}), notice that the
lower bound follows from Lemma \ref{ref:FGvaluesMonotony}. For the
upper bound, observe that
$\q'(x) = \left (\F_V'(x) - \q(x)\right ) /x$. By evaluating it at
$\T_V(\beta)$, and using $\q(\T_V(\beta)) = 1 / \beta$, $\q'(x)\le 0$, we get $\beta \F_V'(\T_V(\beta)) < 1$. 
\end{proof}

\begin{lemma}\label{lem:Rvariations}
Let $\T_V(\beta)$ and $\S_V(\beta,\alpha)$ be defined as per
Eq.~(\ref{lem:Twelldefined}).
 
Then the function $x\mapsto \R\sym(x)$ is strictly increasing on
$(0,\T_V(\beta))$ and  strictly decreasing on
$(\T_V(\beta),+\infty)$. 
Similarly $ x\mapsto\R\rec(x ,\alpha)$ is  strictly increasing on
$(0,\S_V(\beta,\alpha))$  and 
strictly decreasing on $(\S_V(\beta,\alpha),+\infty)$. 
\end{lemma}

\begin{proof}[Proof of Lemma \ref{lem:Rvariations}]
Recall that letting $\D_V(x) \equiv \E \{ (x\bV+\bG)_+^2\}$ we
have, for any  $x\ge 0 $, $x\F_\bV(x) + \G_\bV(x) = \sqrt{\D_V(x)}$ and
$\F_V(x) = \frac{\de\phantom{x}}{\de x} \sqrt{\D_V(x)}$. As a
consequence $x  \F'_\bV(x) = -\G'_\bV(x)$, and therefore,
for all $\beta,x>0$, 
\begin{equation}\label{eq:derivativeRayleigh}
 \frac{\de\phantom{x}}{\de x} \R\sym(x)  = 2 x \left ( \beta~
   \frac{\F_\bV(x)}{x} - 1 \right ) \frac{\de\phantom{x}}{\de x}
 \F_\bV(x)\, .
 \end{equation}
Recall that, by Lemma \ref{ref:FGvaluesMonotony}, $\F_V'(x)>0$. 
Further, as per the proof of Lemma \ref{lem:Twelldefined}, 
$x\mapsto \q(x) = \F_V(x)/x$ is strictly decreasing with
$\q(\T_V(\beta))=1/\beta$. This immediately implies the claim for
$\R\sym$.

 The argument for  $\R\rec(\,\cdot\, ,\alpha)$ is completely
 analogous. 
We write the derivative of $\R\rec(x,\alpha)$ with respect to $x$: 
\begin{equation*}
 \frac{\partial\phantom{x}}{\partial x} \R\rec(x,\alpha)  =  x \left
    ( \beta~ \frac{\F_\bV(x / \sqrt \alpha )}{x \sqrt{1+\beta ~\F_V(x
        / \sqrt \alpha)^2}} - 1 \right ) \frac{\de\phantom{x}}{\de x}
  \F_\bV(x / \sqrt \alpha) ~.
\end{equation*}
The claim follows again from $\F_V'>0$, and using  the properties of $x \mapsto  {\F_\bV(x / \sqrt \alpha )} / \{x \sqrt{1+\beta~ \F_V(x / \sqrt \alpha)^2}\}$ already discussed in the proof of Lemma \ref{lem:Twelldefined}.
\end{proof}

\begin{lemma}\label{lem:convergenceFixedPointSymm} 
Let the state evolution sequence $\{ \tau_t\}_{t  \geq
  0}$ defined by $\tau_1 = \beta\E V$ and $\tau_{t+1} = \beta \F_V(\tau_t)$
for all $t\ge 1$. 
Then, for  any law $\mu_V$, there exist constants $c_0,c_1 >0, \gamma_0 \in
(0,1)$ such that, for all $t\ge 1$
\begin{align}
| \T_V(\beta) - \tau_t | \leq c_0 \gamma_0^t
\quad\quad \text{and}\quad\quad | \R\sym(\T_V(\beta)) -
\R\sym(\tau_t)| \leq c_1 \gamma_0^{2t} ~~.
\end{align}
\end{lemma}
\begin{proof}[Proof of Lemma \ref{lem:convergenceFixedPoint}]
We proved in Lemma \ref{lem:Twelldefined} that $x\mapsto \beta
\F_V(x)$
is monotone increasing with $\beta\F_V(x)>x$ if $x<\T_V(\beta)$ and
$\beta\F_V(x)<x$ if $x>\T_V(\beta)$. It follows that
$\tau_{t+1}>\tau_t$ if $\tau_t<\T_V(\beta)$ and $\tau_{t+1}>\tau_t$ if
$\tau_t>\T_V(\beta)$. Hence $\lim_{t\to\infty}\tau_t=\T_V(\beta)$
Convergence is exponentially fast, i.e.
$| \T_V(\beta) - \tau_t | \leq c_0 \gamma_0^t$, since, by Lemma
\ref{lem:Twelldefined}
$\beta\F'_V(\T_V(\beta))\in(0,1)$.

This proves the first second inequality. Note that $\T_V(\beta)$ is
the global  maximum  of $x\mapsto \R\sym(x)$ and hence, in a
neighborhood of $\T_V(\beta)$, 
$ | \R\sym(\T_V(\beta)) -
\R\sym(\tau_t)|\le c_*(\tau_t-\T_V(\beta))^2$
 \end{proof}
We state without proof the analogous result for the rectangular case. 
The argument is exactly the same as  for the symmetric case. 
\begin{lemma}\label{lem:convergenceFixedPoint} 
Let the state evolution sequence $\{\mu_t,\vth_t\}_{t\ge 0}$ 
be defined by the recursion \ref{eq:seqThetaMu} with the initial condition $\mu_0 = \sqrt \beta \E\bV$.
For any law $\mu_V$  there exist constants $k_0,k_1 >0, \kappa_0 \in
(0,1)$ such that, for all $t\ge 1$,
\begin{align}
\forall t\geq 0~,~| \S_V(\beta,\alpha) - \vth_t | \leq k_0
\kappa_0^t \quad\quad \text{and}\quad\quad |
\R\rec(\S_V(\beta,\alpha),\alpha) - \R\rec(\vth_t,\alpha)| \leq k_1
\kappa_0^{2t} ~~.
\end{align}
\end{lemma}

Our results  are stated in
terms of $\F_V$ and $\G_V$, and depend on the law of
$V$. However, when $\eps(V)  \to 0$, interestingly, two different
phenomena occur. First, our results can be stated independently of law
of $\bV$. Second, a phase transition occurs for a
specific value of the signal-to-noise ratio $\beta$. This is stated
formally below using the notion of uniform convergence introduced in Definition
\ref{def:UniformConvergence}.
\begin{lemma}\label{rk:sparseRegimeFG}
The following limits hold uniformly over the class $\cP$ of
probability distributions on $\reals_{\ge 0}$ with second moment equal
to $1$, and over $x\in [0,M]$ for any $M<\infty$:
\begin{align} 
\lim_{\eps(V)\to 0}\D_V(x) & = \frac{1}{2}+x^2\, ,\label{eq:DLimit}\\
\lim_{\eps(V) \to 0}\F_\bV(x) &= \frac x{ \sqrt{1/2+x^2}} \, ,\label{eq:Flimit}\\
\lim_{\eps(V) \to 0}\G_\bV(x) &= \frac {1/2}{
  \sqrt{1/2+x^2}}~~.\label{eq:Glimit}
\end{align}
Further, again uniformly over $\cP$, for any
$\beta,\alpha\in\reals_{\ge0}$
\begin{align}
\lim_{\eps(V)\to 0}\T_V(\beta) &= \begin{cases}\label{eq:TepsToZero}
0 & \mbox{ if $\beta\le 1/\sqrt{2}$,}\\
\sqrt{\beta^2-(1/2)} & \mbox{ otherwise.}
\end{cases}\\
\lim_{\eps(V)\to 0}\S_V(\beta,\alpha) &= \begin{cases}\label{eq:SepsToZero}
0 & \mbox{ if $\beta\le\sqrt{\alpha/2}$,}\\
\sqrt{\left( \beta^2 - \alpha / 2 \right ) / \left( 1 + \beta \right )} & \mbox{ otherwise.}
\end{cases}
\end{align}
\end{lemma}
\begin{proof}
In order to prove Eq.~(\ref{eq:DLimit}) note that, by taking first the
expectation over $G$ in $\D_V(x) \equiv \E\{(xV+G)_+^2\}$, we get
\begin{align}
\D_V(x) -\Big(\frac{1}{2}+x^2\Big) & = \E\Big\{(1+x^2V^2)\,\Phi(xV) +xV\, \phi(xV)\Big\}-\Big(\frac{1}{2}+x^2\Big) \\
& =
\E\Big\{\big[\Phi(xV)-\Phi(0)\big]+x^2V^2\big[\Phi(xV)-1\big]+xV\phi(xV)\Big\}\equiv
\E\{f(xV)\}\, ,
\end{align}
where $f(z) \equiv
[\Phi(z)-\Phi(0)\big]+z^2\big[\Phi(z)-1\big]+z\phi(z)$. Note that
$f(0) = 0$ and $f(z)$ is bounded, whence
\begin{align}
\Big|\D_V(x) -\frac{1}{2}-x^2\Big| & \le \E\{|f(xV)|\one_{\{V\neq
  0\}}\}\le \|f\|_{\infty}\, \eps\, ,
\end{align}
which yields the desired uniform convergence of $\D_V$.

Next recall that $\F_V(x) = \frac{\de\phantom{x}}{\de
  x}\sqrt{\D_V(x)}$,  cf. Eq.~(\ref{eq:FDerivative}) and, by Lemma \ref{ref:FGvaluesMonotony},
$\sqrt{\D_V(x)}$ is strictly convex. We hence have, for all
$\delta>0$, 
\begin{align}
\frac{1}{\delta}\inf_{\mu_V\in\cP_{\eps}}\big[\sqrt{\D_V(x)}-\sqrt{\D_V(x-\delta)}\big]
\le \inf_{\mu_V\in\cP_{\eps}}\F_V(x)\le
\sup_{\mu_V\in\cP_{\eps}}\F_V(x)\le
\frac{1}{\delta}\sup_{\mu_V\in\cP_{\eps}}\big[\sqrt{\D_V(x+\delta)}-\sqrt{\D_V(x)}\big]\, .
\end{align} 
The claim (\ref{eq:Flimit}) follows by taking the limit $\eps\to 0$ (using
Eq.~(\ref{eq:DLimit})) followed by $\delta\to 0$.
The expression of $\lim_{\eps(V) \to 0} \G_V(x)$ follows by taking the
limit on the identity $x~\F_V(x) + \G_V(x) = \sqrt{\D_V(x)}$.

In order to prove Eq.~(\ref{eq:TepsToZero}), let $\T_0(\beta)$ denote
the function on the right-hand side and assume by contradiction that 
there exists a sequence $\eps_n\to 0$, probability measures
$\mu_{V_n}\in\cP_{\eps_n}$ such that $\lim_{n\to\infty}\T_{V_n}(\beta)
>x_*  = \T_0(\beta)+\delta$ for some $\delta>0$. As shown in the proof
of Lemma \ref{lem:Twelldefined}, $x\mapsto x/\F_V(x)$ is monotone increasing.
Using the definition  we have, for all $n$ large enough
\begin{align}
\beta = \frac{\T_{V_n}(\beta)}{\F_{V_n}(\T_{V_n}(\beta))} \ge
\frac{x_*}{\F_{V_n}(x_*)} \ge \frac{x_*}{\sup_{\mu_V\in\cP_{\eps_n}}\F_V(x_*)}\, .
\end{align}
Taking the limit $n\to\infty$, and using Eq.~(\ref{eq:Flimit}), we get
\begin{align}
\beta\ge \sqrt{\frac{1}{2}+\big(\T_0(\beta)+\delta\big)^2}\, ,
\end{align}
that yields a contradiction by the definition of $\T_0$. Hence $\lim\sup_{\mu_V\in\cP_{\eps}}\T_{V}(\beta)\le\T_0(\beta)$.
The matching lower bound is proved in the same way.

Finally, the proof of Eq. (\ref{eq:SepsToZero}) follows along the same
lines. 
\end{proof}

\subsection{Upper bounds: Proof of Lemma \ref{th:upperBounds}}
\label{sec:ProofUpperBound}

In this section we prove Lemma \ref{th:upperBounds}. As mentioned
before, the proof of Lemma \ref{th:upperBoundsSymmetric} is completely
analogous and omitted.

For $\mu\in [0,1]$, we define
\begin{align}
\cW_{\mu} &\equiv \left \{ (\bu,\bv) \in \reals^n\times \reals^p~:~\| \bu\|_2 = 1
  ~,\| \bv \|_2 = 1 ~,~ \bv\geq 0, 
  \<\bv,\bvz\> =\mu\right \}\,, \\
M_{\bX}(\mu) & \equiv\max\big\{\<\bu,\bX\bv\>:\;\, 
(\bu,\bv)\in\cW_{\mu}\big\}\, ,\\
\oM(\mu) & \equiv\E M_{\bX}(\mu) = \E \max\big\{\<\bu,\bX\bv\>:\;\, 
(\bu,\bv)\in\cW_{\mu}\big\}\, .
\end{align}
Note that 
\begin{align}
\sigma^+(\bX) & = \max_{\mu\in [0,1]}M_{\bX}(\mu)= 
M_{\bX}(\<\bv^+,\bvz\>) \, . \label{eq:SigmaPlus}
\end{align}
The function  $\bX\mapsto M_{\bX}(\mu)$ is Lipschitz continuous with
Lipschitz constant  $1$ (namely $|M_{\bX}(\mu)-M_{\bX'}(\mu)|\le
\|\bX-\bX'\|_F$). Hence, by Gaussian isoperimetry, we have 
\begin{align}
\prob \Big\{ \big|M_{\bX}(\mu) - \oM(\mu)\big| \ge t \Big\}
  \le 2\, 
  e^{-nt^2/2}\, . \label{eq:Isoperimetry}
\end{align}
Further we claim that $\mu\mapsto M_{\bX}(\mu)$ is uniformly
continuous for $\mu\in [0,1]$. Indeed  if $\bv(\mu)
=\sqrt{1-\mu^2}\bv_{\perp}(\mu) +\mu\bv_0$ realizes the maximum
over $\cW_{\mu}$ (with $\<\bv_{\perp}(\mu),\bv_0\>=0$), we have
\begin{align}
M_{\bX}(\mu_1) &=\big\|\bX\bv(\mu_1)\big\|_2\ge \Big\|\bX\Big(
\sqrt{1-\mu_1^2}\,\bv_{\perp}(\mu_0) +\mu_1\bv_0\Big)\Big\|_2\\
&\ge M_{\bX}(\mu_0) - C\|\bX\|_2(\mu_1-\mu_0)^{1/2}\, .
\end{align}
Recall that $\prob\{\|\bX\|\ge
C_1\}\le C_2e^{-n/C_2}$ for some constants $C_1(\alpha), C_2(\alpha)$
\cite{Guionnet}.
Hence, with probability at least $1-C_2e^{-n/C_2}$ we have, for all $\mu_0,
\mu_1\in [0,1]$
\begin{align}
\big|M_{\bX}(\mu_1)-M_{\bX}(\mu_0)\big| &\le C' \, |\mu_1-\mu_0|^{1/2}\, ,
\big|\oM(\mu_1)-\oM(\mu_0)\big| &\le C' \, |\mu_1-\mu_0|^{1/2}\, .
\end{align}
Let ${\cI}_n\equiv \{0,1/n, 2/n,\dots\}\cap [0,1]$ be a grid. 
By the above uniform continuity, we have, with probability at least
$1-C_2e^{-n/C_2}$,
\begin{align}
\sup_{\mu\in [0,1]}\big|M_{\bX}(\mu) - \oM(\mu)\big| \le \sup_{\mu\in
  \cI_n} \big|M_{\bX}(\mu) - \oM(\mu)\big| +C''n^{-1/2}\, .
\end{align}
Using Eq.~(\ref{eq:Isoperimetry}) and union bound over $\cI_n$, 
we conclude that
\begin{align}
\prob \Big\{\max_{\mu\in [0,1]} \big|M_{\bX}(\mu) - \oM(\mu)\big| \ge t \Big\}&
  \le 2\, n\, \exp\Big\{-\frac{n}{2}(t-C''n^{-1/2})^2\Big\} +
  C_2e^{-n/C_2} \le Cn \, e^{-nt^2/4}\, ,
\end{align}
where the last inequality holds for all   $t\le t_0$ with
$t_0$ a suitable constant.
In particular, by Borel-Cantelli we have, almost surely and in expectation,
\begin{align}
\lim_{n\to\infty}\max_{\mu\in [0,1]} \big|M_{\bX}(\mu) -
\oM(\mu)\big| =0\, . \label{eq:MoM}
\end{align}

In order to upper bound $\oM(\mu)$, we apply Vitale's extension of Sudakov-Fernique inequality (see e.g.
\cite[Theorem 1]{vitale2000some} and \cite[Theorem 1]{chatterjee05}
for a quantitative version) 
to the two processes $\{\cX(\bu,\bv)\}$, $\{\cY(\bu,\bv)\}$ indexed by
$(\bu,\bv) \in \cW_{\mu}$ defined as follows:
\begin{align}
\cX(\bu,\bv) & \equiv \<\bu, \bX \bv \>= \sqrt \beta \langle \buz ,
\bu \rangle \langle \bvz, \bv\rangle+\<\bu, \bZ \bv \>\, ,\\
\cY(\bu,\bv) & \equiv \sqrt \beta \langle \buz , \bu \rangle \langle
\bvz, \bv\rangle+\frac 1 {\sqrt n} \left ( \langle \bg , \bu\rangle +
  \langle \bh,\bv\rangle \right )\, ,
\end{align}
for independent random  vectors $\bg \sim \normal(0,\id_n), \bh \sim
\normal(0,\id_p)$.
It is easy to see that $\E\cX(\bu,\bv)= \E\cY(\bu,\bv)$ and 
\begin{align}
\E\big\{\big[\cX(\bu_1,\bv_1)-\cX(\bu_2,\bv_2)\big]^2\big\}&  =\left \{ \E\cX(\bu_1,\bv_1) - \cX(\bu_2,\bv_2)\right \}^2 + 
\frac{2}{n}
\big(1-\<\bu_1,\bu_2\>\<\bv_1,\bv_2\>\big)\, ,\\
\E\big\{\big[\cY(\bu_1,\bv_1)-\cY(\bu_2,\bv_2)\big]^2\big\}&  =  \left \{ \E\cY(\bu_1,\bv_1) - \cY(\bu_2,\bv_2)\right \}^2 + \frac{2}{n}
\big(2-\<\bu_1,\bu_2\>-\<\bv_1,\bv_2\>\big)\, .
\end{align}
Hence $\E\big\{\big[\cX(\bu_1,\bv_1)-\cX(\bu_2,\bv_2)\big]^2\big\}\le
\E\big\{\big[\cY(\bu_1,\bv_1)-\cY(\bu_2,\bv_2)\big]^2\big\}$
(this follows from $1-ab\le 2-a-b$ for $a,b\in [-1,1]$). We conclude
that
\begin{align}
\oM(\mu)
&\le\E \max \left \{ \sqrt \beta \mu\langle \buz , \bu \rangle +\frac 1 {\sqrt n} \left ( \langle \bg , \bu\rangle
    + \langle \bh,\bv\rangle \right ) ~:~(\bu,\bv) \in \cW_{\mu}
\right \}\\
& \le \E \max \left \{ \langle\frac{1}{\sqrt{n}}\bg+\sqrt{\beta}
  \,\mu\buz , \bu \rangle
    + \langle \frac{1}{\sqrt{n}}\bh+\vartheta\bvz,\bv\rangle -\vartheta\mu ~:~(\bu,\bv) \in \cW \right \}
\, .
\end{align}
where  last inequality holds for any $\vartheta\in\reals$, setting $\cW
\equiv \cup_{\mu}\cW_{\mu}$.

The maximum in the last expression is achieved for  
 \begin{align}
 \bu = \frac{\bg +\mu \sqrt {\beta n} \buz }{\|   \bg +\mu \sqrt{
     \beta n}
   \buz\|_2}\quad , \quad \bv = \frac{\left ( \bh +\vartheta {\sqrt n}
     \bvz \right )_+ }{\|\left ( \bh +\vartheta{\sqrt n} \bvz \right
   )_+\|_2}~~. 
\end{align}
Hence, by Lemma \ref{lem:FGinnerproducts}, 
there exists a deterministic sequence $\delta_n=\delta_n(\alpha,\beta,\vartheta)$
independent of $\mu\in [0,1]$, such that $\lim_{n\to\infty}\delta_n=
0$ for any $\vartheta\in\reals$ and
\begin{align}
\oM(\mu)&\le \E\left \{\Big\|\frac{1}{\sqrt{n}}\bg+\sqrt{\beta}
  \,\mu\buz \Big\|_2
    + \Big\| \Big(\frac{1}{\sqrt{n}}\bh+\vartheta\bvz\Big)_+\Big\|_2-\vartheta\mu \right \}\\
&\le \sqrt{1+\beta\mu^2}+
\sqrt{\alpha\D_V(\vartheta/\sqrt{\alpha})}-\vartheta\mu + \delta_n\, ,\label{eq:UpperBoundLimit}
\end{align}
where we recall that $\D_V(x) \equiv \E\{(xV+G)_+^2\}$.

We next fix $\vartheta=\vartheta_*(\alpha,\beta) =
\S_V(\beta,\alpha)$, which is also the unique maximizer of $x\mapsto
\R\rec(x,\alpha)$, as shown in  Lemma \ref{lem:Rvariations}.
Note that Eq.~(\ref{eq:UpperBoundLimit}) is strictly concave in $\mu\in
[0,1]$, with unique maximum at $\mu_*
=(\vartheta_*/\beta)(1-\vartheta_*^2/\beta)^{-1/2}$. 
Substituting in Eq.~(\ref{eq:UpperBoundLimit}), we get
\begin{align}
\max_{\mu\in[0,1]}\oM(\mu)&\le \sqrt{1-\vartheta_*^2/\beta}+ 
\sqrt{\alpha\D_V(\vartheta_*/\sqrt{\alpha})}\\
& = \R\rec(\vartheta_*,\alpha) +\delta_n\, ,\label{eq:MaxM}
\end{align}
where the last equality follows from the identity $\sqrt{\D_V(x)} =
x\F_V(x)+\G_V(x)$, and from the equation $\vartheta_* =
\beta\F_V(1+\beta\F_V)^{-1/2}$ with $\F_V =
\F_V(\vartheta_*/\sqrt{\alpha})$ that holds by definition of
$\vartheta_* = \S_V(\beta,\alpha)$.

From Eq.~(\ref{eq:SigmaPlus}), (\ref{eq:MoM}) and (\ref{eq:MaxM}) we finally get 
\begin{align}
\lim\sup_{n\to\infty} \E~\sigma^+(\bX) &\le
\lim\sup_{n\to\infty}\max_{\mu\in[0,1]}\oM(\mu) \\
&\le  \R\rec(\vartheta_*, \alpha ) =
\max_{\vartheta\in\reals}\R\rec(\vartheta , \alpha )\, ,\label{eq:LastUpperboundSigma}
\end{align}
which coincides with 
claim (\ref{eq:upperBoundGeneralProposition}).

Next reconsidering Eq.~(\ref{eq:UpperBoundLimit}) with $\vartheta =
\vartheta_*$,   we see that since the right-hand side is strictly
concave in $\mu\in [0,1]$, we can strengthen Eq.~(\ref{eq:MaxM}) to
\begin{align}
\oM(\mu)&\le \R\rec(\vartheta_*,\alpha) -c_*(\mu-\mu_*)^2+\delta_n\, ,\label{eq:StrictConvM}
\end{align}
for some $c_*>0$. We call $H(x) = c_*x^2$.

By Eq.~(\ref{eq:SigmaPlus}) and (\ref{eq:MoM}) we have, almost surely,
\begin{align}
\lim\inf_{n\to\infty}\sigma^+(\bX)
=\lim\inf_{n\to\infty}M_{\bX}(\<\bv^+,\bvz\>) \le
\lim\inf_{n\to\infty}\oM(\<\bv^+,\bvz\>) \, .
\end{align}
We then use Eq.~(\ref{eq:StrictConvM}) to deduce that
\begin{align}
\lim\inf_{n\to\infty}\sigma^+(\bX)& \le  \R\rec(\vartheta_*,\alpha)
-\lim\sup_{n\to\infty} H(|\<\bv^+,\bvz\>-\mu_*|) \\
& =  \R\rec(\vartheta_*,\alpha)
-H\big(\lim\sup_{n\to\infty} |\<\bv^+,\bvz\>-\mu_*|\big) 
\, .
\end{align}
This implies immediately
Eq.~(\ref{eq:upperBoundNNinnerProductNonSymDense}) with $\Delta = H^{-1}$,
since (as shown above) $\vartheta_* = \S_V(\beta,\alpha)$, and $\mu_* = \F_V
  (\S_V(\beta,\alpha) / \sqrt \alpha)$.

\subsection{Lower bounds: Proofs of Theorem
 \ref{th:lowerBoundsAMPsymmetric} and Theorem
  \ref{th:lowerBoundsAMPnonsymmetric}}
\label{sec:ProofLowerBounds}

In this section we prove lower bounds on the non-negative eigenvalue
(singular value) that follows from the analysis of the AMP algorithm,
namely  Theorem \ref{th:lowerBoundsAMPsymmetric} for symmetric
matrices and Theorem \ref{th:lowerBoundsAMPnonsymmetric} for rectangular matrices.
The proofs are very similar in the two cases, hence we will provide
details only in the case of rectangular matrices, and limit ourselves
to pointing out differences arising in the symmetric setting.

\subsection{Proof of Theorem \ref{th:lowerBoundsAMPnonsymmetric}}

  Define 
\begin{align}
  r_t(n) \equiv \<\hbu^t,\bX\, \hbv^t\> =  \frac 1n \langle g(\bu^t) , \bX f(\bv^t)
  \rangle\, ,
\end{align}
  and observe, using \ref{eq:AMPnonsym},
  \begin{align}
 r_t(n) & = \frac 1n \langle g(\bu^t) , \bu^{t} + \ons_t\, g(\bu^{t-1})  \rangle \\
  & = \frac 1n \langle g(\bu^t) , \bu^t \rangle + \frac{\ons_t}{n}
  \langle  g(\bu^t) ,  g(\bu^t)  \rangle +
 \frac{\ons_t}{n} \langle  g(\bu^t) ,  \big(g(\bu^{t-1}) -   g(\bu^t)  \big) \rangle \\
  & = \frac 1n \langle g(\bu^t) , \bu^t \rangle +  \ons_t  + E_t(n)~~. \label{eq:Rt}
\end{align}
where
\begin{align}
|E_t(n)| &= \frac{\ons_t}{n}\Big| \langle  g(\bu^t) ,  \big(g(\bu^{t-1}) -
 g(\bu^t) \big) \rangle \Big|\\
&\le \frac{\ons_t}{\sqrt{n}}\, 
\big\|g(\bu^{t-1}) - g(\bu^t) \big\|_2\\
& \le 4\ons_t\,
\frac{\|\bu^{t-1}-\bu^t\|_2}{\|\bu^{t-1}\|_2+\|\bu^{t}\|_2}\, .\label{eq:E-UpperBound}
\end{align}
The last step follows from triangular inequality.

By Proposition \ref{prop:SEgeneral} applied to $\psi(x,y) = x^2$, we
have, almost surely,
\begin{align}
\lim_{n\to\infty}\frac{1}{n}\|\bu^t\|_2^2 &= \E\{(\mu_t U+G)^2\} =
1+\mu_t^2\, ,\label{eq:U-2norm}
\end{align}
and therefore
\begin{align}
\lim_{n\to\infty}\frac{1}{n}\<g(\bu^t),\bu^t\> &=
\lim_{n\to\infty}\frac{1}{\sqrt{n}} \, \|\bu^t\|_2 = \sqrt{1+\mu_t^2}\\
& =   \sqrt{1+\beta \F_V(\vartheta_{t-1} /\sqrt \alpha  )^2} \, . \label{eq:Gu}
\end{align}
By applying the same proposition to $\psi(x,y) = x_+^2$ we have 
\begin{align}
\lim_{n\to\infty}\frac{1}{p}\|(\bv^t)_+\|_2^2 &=
\E\{(\vth_t/\sqrt{\alpha}\, V+G)^2_+\} \, .\label{eq:Vp-2norm}
\end{align}
Further, letting $\psi(x,y) =\ind(x>0) = 1-\ind(x\le 0)$ we get 
\begin{align}
\lim_{n\to\infty}\frac{1}{p}\|(\bv^t)_+\|_0&=\lim_{n\to\infty}\frac{1}{p}\sum_{i=1}^p\ind(\bv^t_i>0)
= \prob\big(\vth_t/\sqrt{\alpha}\, V+G>0\big)\\
&= \E\big\{G(\vth_t/\sqrt{\alpha}\, V+G)_+\big\}\, ,\label{eq:Vp-0norm}
\end{align}
where the last equality follows from Stein's lemma \cite{Ste72}.
Using together Eq.~(\ref{eq:Vp-2norm}) and Eq.~(\ref{eq:Vp-0norm}), we get
\begin{align}
\lim_{n\to\infty} \ons_t(n) = \lim_{n\to \infty}\sqrt{\alpha}\,
\frac{\|(\bv^t)_+\|_0/p}{\|(\bv^t)_+\|_2/p} = \sqrt \alpha~\G_V \left
  (\frac{\vartheta_t}{\sqrt \alpha } \right )\, . \label{eq:limbt}
\end{align}
Using Eq.~(\ref{eq:U-2norm}) and Eq.~(\ref{eq:limbt}) in the
upper bound (\ref{eq:E-UpperBound}), we get
\begin{align}
\lim_{t\to\infty} \lim_{n\to\infty} |E_t(n)| = 0\, .
\end{align}
Finally, substituting this result together with Eq.~(\ref{eq:Gu}) and
(\ref{eq:limbt}) in Eq.~(\ref{eq:Rt}), we obtain
\begin{align}
\lim_{t\to\infty} \lim_{n\to\infty} r_t(n) =\lim_{t\to\infty} \R\rec(\vth_t,\alpha)\, .
\end{align}
The claim (\ref{eq:limAMPRayleighNonSym}) follows by Lemma
\ref{lem:convergenceFixedPoint}.

Consider next Eq. ~(\ref{eq:AMPScalProdURect}). We have
\begin{align}
\lim_{n\to\infty}  \langle \hbu^t,\buz\rangle  &= \lim_{n\to\infty} \frac{\<\bu^t,\buz\>}{\|\bu^t\|_2}\\
& =\frac{\E\{U(\mu_tU+G)\}}{\sqrt{\E\{(\mu_t U+G)^2\}}} = \frac{\mu_t}{\sqrt{1+\mu_t^2}}\, ,
\end{align}
where the second equality follows by applying Proposition
\ref{prop:SEgeneral} to $\psi(x,y) = xy$ (for the numerator) and using
Eq.~(\ref{eq:U-2norm}) (for the denominator). Finally, the claim
(\ref{eq:AMPScalProdURect})
follows by taking $t\to\infty$, and using Lemma
\ref{lem:convergenceFixedPoint}.

The proof of claim (\ref{eq:AMPScalProdVRect}) follows by the same
argument and we omit it.

\subsubsection{Proof of Theorem \ref{th:lowerBoundsAMPsymmetric}}

The proof in the symmetric case is very similar to the one for
rectangular matrices, see Theorem \ref{th:lowerBoundsAMPnonsymmetric}.
We limit ourselves to sketching the first steps.  We have, using \ref{eq:AMPsym},
\begin{align}
  \rho_t(n) &\equiv \<\hbv^t,\bX\, \hbv^t\> \\
  & =  \frac 1n \langle f(\bv^t) , \bX f(\bv^t)
  \rangle\, \\
& = \frac 1n \langle f(\bv^t) , \bv^{t+1} + \ons_t\, f(\bv^{t-1})  \rangle \\
  & = \frac 1n \langle f(\bv^t) , \bv^t \rangle + \frac{\ons_t}{n}  \langle  f(\bv^t) ,   f(\bv^t)  \rangle+  \frac 1n \langle f(\bv^t) , \bv^{t+1}  - \bv^t \rangle  + \frac{\ons_t}{n} \langle  f(\bv^t) ,  \big(f(\bv^{t-1}) -   f(\bv^t)  \big) \rangle  \\
  & = \frac 1n \langle f(\bv^t) , \bv^t \rangle +  \ons_t  + \tilde E_t(n)~~. 
\end{align}
and we are left with a term $\tilde E_t = \frac 1n \langle f(\bv^t) ,
\bv^{t+1}  - \bv^t \rangle  + \frac{\ons_t}{n} \langle  f(\bv^t) ,
\big(f(\bv^{t-1}) -   f(\bv^t)  \big) \rangle  $ which we treat
similarly to $E_t(n)$ of Theorem \ref{th:lowerBoundsAMPnonsymmetric}.
Namely,  by using  Proposition \ref{prop:SEgeneral} and Proposition
\ref{lem:diffGoesToZero}, we prove that
\begin{align}
 \lim_{t\to \infty} \lim_{n\to \infty} \tilde E_t(n) = 0~~.
\end{align}
In addition, it follows from  Proposition \ref{prop:SEgeneral} that
 \begin{align}
 \lim_{n\to \infty}  \frac 1n \langle f(\bv^t) , \bv^t \rangle &=
 \tau_t \F_V(\tau_t) + \G_V(\tau_t) = \beta \F_V(\tau_{t-1})
 \F_V(\tau_t) + \G_V(\tau_t)~,\\
\lim_{n\to \infty} \ons_t &= \G_V(\tau_t)~.
\end{align}
This terminates the proof sketch.
%
%
\subsection{Minimax analysis: proof of Theorems \ref{th:Fworstcase} and \ref{th:FworstcaseRec}}\label{sec:worstCase}

In this section we prove that the least favorable vectors $\bvz$ are
--asymptotically-- of the following form:
$(\bvz)_i = 1/\sqrt{|S|}$ for all $i\in S$, and $(\bvz)_i = 0$
otherwise, for some support $S\subseteq [p]$. Further, we characterize
the least favorable size of the support $|S|$.

The proofs proceed by analyzing the expression in Theorem
\ref{th:mainSym} and applying strong duality to a certain linear
program over probability distributions, that is related to the function
$\mu_V\mapsto \F_V(x)$. We start  with some preliminary facts and
definitions in Section \ref{subsec:Preliminary}. The key step is to
reduce ourselves to two points mixtures: this is achieved in Section
\ref{subsec:Reduction}.
Finally, in Sections \ref{subsec:WorstSymm} and \ref{subsec:WorstRec},
we use these results to prove Theorems  \ref{th:Fworstcase} and \ref{th:FworstcaseRec}.
Since the proof of Theorem \ref{th:FworstcaseRec} is completely
analogous to the one of Theorem  \ref{th:Fworstcase}, we will limit
ourselves to mentioning the main differences.

\subsubsection{Preliminary definitions}
\label{subsec:Preliminary}

For $\eps\in (0,1]$ and $v\in\reals_{\ge 0}$, we
 define the 2-points mixture
\begin{align}
\mu_{\eps,v} \equiv (1-\eps)\delta_0+\eps\delta_v\, ,
\end{align}
In particular, when $v=1/\sqrt{\eps}$ (and hence the above
distribution has second moment equal to $1$), we write $\mu_{\eps} =
\mu_{\eps,1/\sqrt{\eps}}$. We also write --with a slight abuse of
notation-- $\F_\eps(x)$ instead of $\F_V(x)$ when
$V\sim\mu_{\eps}$. Explicitly 
\begin{align}
\F_{\eps}(x) = \frac{\E\{(x+\sqrt{\eps}G)_+\}}{\sqrt{(1-\eps)/2+
    \E\{(x+\sqrt{\eps}G)_+^2\}}}\, .
\end{align}
Even more explicitly 
\begin{align}
\F_{\eps}(x) &= 
\frac{\eps\,  B(x/\sqrt{\eps})/x}
{\sqrt{(1-\eps)/2+\eps(B(x/\sqrt{\eps})+\Phi(x/\sqrt{\eps}))} }\, ,\label{eq:FepsDefinition}\\
B(w) & \equiv w^2\Phi(w)+w\,\phi(w)\, .
\end{align}
We will also adopt the shorthand $\T_{\eps}(\beta) = \T_{V}(\beta)$
when $V\sim \mu_{\eps}$.

We wil next establish two calculus lemmas that are useful for the following.
\begin{lemma}\label{lemma:TwoSols}
For any given $a,b\in \reals$, the equation
\begin{align}
\frac{\phi(v)}{v} + b \Phi(v) = a\, ,\label{eq:SimpleEq}
\end{align}
in the unknown $v\in\reals_{>0}$  has at most two solutions $v_1,v_2$.
\end{lemma}
\begin{proof}
Let $h_b(v) = (\phi(v)/v)+b\Phi(v)$ denote the left-hand side of
Eq.~(\ref{eq:SimpleEq}).
Then
\begin{align}
h'_b(v) = -\Big(1-b+\frac{1}{v^2}\Big)\, \phi(v)\, .
\end{align}
If $b\le 1$, then we conclude that $h'_b(v)<0$ for all $v>0$ and hence
the equation $h_b(v) = a$ has at most one positive solution. If --on the
other hand-- $b>1$, then $h'_b(v) <0$ for $v<v_* \equiv (b-1)^{-1/2}$ and
$h'_b(v) >0$ for $v>v_*$. It follows that the equation
$h_b(v) = a$ has at most one solution in $(0,v_*]$ and at most one in $(v_*,\infty)$.
\end{proof}

\begin{lemma}\label{lemma:ElementaryInequality}
Let $b:\reals\to \reals$ be defined as $b(x) =
x^2(\Phi(x)-1)+x\phi(x)$.
Then, for every $x\in (0,\infty)$, we have
\begin{align}
\phi(x)\, b(x) >\Big(\Phi(x) -\frac{1}{2}\Big)\, b'(x) \label{eq:Inequality}\, .
\end{align}
\end{lemma}
\begin{proof}
By simple calculus, we get $b(0) =0$, and the derivatives
\begin{align}
b'(x) & = \phi(x)-2x\big(1-\Phi(x)\big)\, ,& b'(0) = \phi(0)\, ,\\
b''(x) & = x\phi(x)-2\big(1-\Phi(x)\big)\, , &b''(0) = -1\, ,\\
b'''(x) & = (3-x^2)\, \phi(x)\, , & b'''(0) = 3\phi(0)\, . 
\end{align}
Let us further recall the inequalities (valid for $x>0$)
\begin{align}
\frac{\phi(x)}{x}\Big(1-\frac{1}{x^2}\Big)< 1-\Phi(x) <
\frac{\phi(x)}{x}\, ,
\end{align}
which immediately imply for all $x>0$
\begin{align}
0&< b(x) < \frac{\phi(x)}{x}\, .
\end{align}
Therefore the left-hand side of Eq.~(\ref{eq:Inequality}) is always
strictly positive. Consider the right-hand side. 
By consulting special values of the normal distribution, we see that 
$b'(1) = \phi(1)-2(1-\Phi(1))< -0.07<0$.
By a change of
variables we know that $x(1-\Phi(x))/\phi(x) =
\int_{0}^{\infty}\exp(-z-z^2/(2x^2))\de z$ or, equivalently
\begin{align}
b'(x) = \phi(x)\Big\{1-2\E\big[e^{-Z^2/(2x^2)}\big]\Big\}\, .
\end{align}
Since the term in curly brackets is decreasing in $x$, and is negative
at $x=1$, we have $b'(x)<0$ for all $x\ge 1$.
Therefore the right-hand side of Eq.~(\ref{eq:Inequality})  is non-positive for $x\ge
1$. This proves the claim for $x\ge 1$, and we will assume hereafter
$x\in (0,1)$.

Next notice that $0\le b'''(x) \le 3\phi(0)$ for
$x\in(0,1)$. Therefore, by Taylor expansion and intermediate value
theorem, we get, for $x\in (0,1)$, 
\begin{align}
b'(x) \le \phi(0) -x +\frac{3\phi(0)}{2}\, x^2\, .
\end{align}
The right-hand side is negative for $x\in (x_0,x_1)$ where
\begin{align}
x_{1/0} =\frac{1\pm\sqrt{1-6\phi(0)^2}}{3\phi(0)} \, .
\end{align}
In particular $x_0\le 2/3$, and $x_1>1$.  It follows that the
right-hand side of Eq.~(\ref{eq:Inequality}) is non-positive for $x\ge
x_0$. 

We will therefore restrict ourselves to considering $x\in
(0,x_0)\subseteq (0,2/3)$. Note that our claim can be equivalently
written as
\begin{align}
b(x) \ge \Big(\Phi(x)-\frac{1}{2}\Big)\,
\Big(1-2x\frac{1-\Phi(x)}{\phi(x)}\Big)\, .
\end{align}
We will next develop, for $x\in (0,x_0)$, a lower bound on the
left-hand side, to be denoted by $l(x)$, and an upper bound on the
right-hand side, to be denoted by $u(x)$ and prove that $l(x)\ge
u(x)$.
For the left hand side note that $b'''(x)\ge 0$ for $x\in (0,x_0)$ and
hence, again by Taylor expansion
\begin{align}
b(x) \ge \phi(0)x-\frac{1}{2}x^2\equiv l(x)\, .
\end{align}
For the right hand side note that $\Phi(x) -(1/2)\le\phi(0)x$.
Further $x\mapsto (1-\Phi(x))/\phi(x)$ is monotone decreasing. 
We therefore define
\begin{align}
d_0 \equiv \frac{2(1-\Phi(x_0))}{\phi(x_0)}\, ,
\end{align} 
and obtain the upper bound $u(x) = \phi(0)x\, (1-d_0x)$.
Hence
\begin{align}
l(x) -u(x) = \Big(d_0\phi(0)-\frac{1}{2}\Big)x^2\, ,
\end{align}
It is a straightforward exercise to check that indeed
$d_0\phi(0)>(1/2)$ thus completing the proof.
\end{proof}

\subsubsection{Reduction to two points mixtures}
\label{subsec:Reduction}

The main theorem of this Section shows that $\F_V(x)$ is minimized by
probability measures $\mu_V$ that are mixture of at most two point
masses.
\begin{theorem}\label{thm:LeastFavorableF}
Fix $x\ge 0$. Then for any random variable $V$ with probability
distribution $\mu_V\in\cP_{\beps}$, we have
\begin{align}
\F_V(x) \ge \min_{\eps\in (0,\beps]} \F_{\eps} (x)\, .
\end{align}
\end{theorem}
The proof of this theorem is presented at the end of the
section. Before getting to it,  we'll introduce a related problem.
Note that
\begin{align}
\F_V(x) \ge \inf_{y\in \reals_{>0}} \frac{\cF(x,y)}{y}\, ,\label{eq:FcF}
\end{align}
where $\cF(x,y)$ is the value of a constrained optimization problem:
\begin{align}
\text{minimize}\;\;\;\; & \E\{V(xV+G)_+\}\, ,\nonumber\\
\text{subject to}\;\;\;\; & \mu_V\in\cP_{\beps}, \label{eq:ConstrainedProblem}\\
& \E\{(xV+G)_+^2\} =y^2\, .\nonumber
\end{align}
Here it is understood that $\cF(x,y) = \infty$ if this problem is
unfeasible. 
\begin{lemma}\label{lemma:Constrained}
Let $x,y\in\reals_{>0}$ be such that the problem
(\ref{eq:ConstrainedProblem}) is
feasible. Then there exist $\eps\le \beps$ such that $\mu_{\eps}$ is
feasible and $q\in [0,1]$, such that, letting $v_*^2= (1-q)/\eps$,
we have
\begin{align}
\cF(x,y) &=  qx+\int\E\{v(xv+G)_+\}\, \mu_{\eps,v_*}(\de v)\, ,\label{eq:CalF}\\
y^2 &=  qx^2+\int\E\{(xv+G)_+^2\}\, \mu_{\eps,v_*}(\de v)\, .\label{eq:YF}
\end{align}
\end{lemma}
\begin{proof}
By a rescaling of the objective function, and letting $W=xV$, we
can rewrite the problem  (\ref{eq:ConstrainedProblem}) as
\begin{align}
\text{minimize}\;\;\;\; & \E\{W(W+G)_+\}\, ,\nonumber\\
\text{subject to}\;\;\;\; & \mu_W(\{0\})\ge 1-\beps, \\
&\E\{W^2\} = x^2\, ,\nonumber\\
& \E\{(W+G)_+^2\} =y^2\, .\nonumber
\end{align}
Now, for fixed $w\in\reals$, let
\begin{align}
f(w) &\equiv \E\{w(w+G)_+\} = w^2\Phi(w)+w\,\phi(w)\, ,\label{eq:f}\\
g(w) &\equiv \E\{(w+G)_+^2\} = (1+w^2)\Phi(w) +w\,\phi(w)\, ,\label{eq:g}
\end{align}
and write $\mu_W = (1-\beps)\delta_0+(1/g(w))\mu$ with $\mu$ a measure on
$\reals_{\ge 0}$. Then we can rewrite the optimization problem as the
following
(with decision variable $\mu$)
\begin{align}
\text{minimize}\;\;\;\; & \int \frac{f(w)}{g(w)} \mu(\de w)\, ,\nonumber\\
\text{subject to}\;\;\;\; & \int\, \frac{1}{g(w)}\mu(\de w)=\beps\, , \label{eq:MuProblem}\\
& \int \,\frac{w^2}{g(w)}\,\mu(\de w) = x^2\, ,\nonumber\\ 
&\int\, \,\mu(\de w) =y^2\, .\nonumber
\end{align}
The corresponding value is $x\cF(x,y)$.
Note that each of the functions $(f(w)/g(w))$, $1/g(w)$, $w^2/g(w)$ is
bounded and Lipschitz continuous, with a finite limit as $w\to\infty$. This implies that the value
$x\cF(x,y)$ is achieved by a measure $\mu_*$ on the
completed real line $[0,\infty]$, with total mass $y^2$. Indeed the family of normalized
distributions on $[0,\infty]$ is tight and both the objective and the
constraints are continuous in the weak topology. Hereafter, we shall assume this
holds. Functions on $[0,\infty)$ are extended by continuity to $+\infty$.

By introducing Lagrange
multipliers, we obviously have, for any $\alpha,\beta,\gamma\in\reals$
\begin{align}
x\cF(x,y) \ge \beps\alpha+\beta\, x^2+\gamma\, y^2+ \inf_{\mu} \int
\Big\{\frac{f(w)-\alpha-\beta\, w^2-\gamma\,g(w)}{g(w)}\Big\} \mu(\de w)\, ,\label{eq:Duality}
\end{align}
where the infimum is over measures $\mu$ on $\reals_{\ge 0}$.
  By
strong duality (which follows, for instance,  from the Kneser-Kuhn 
minimax theorem \cite{kneser1952theoreme}, see also 
\cite[Theorem A.1]{JohnstoneBook}), there exists\footnote{In general,
  the mentioned theorem only imply that equality is achieved
  asymptotically, along a sequence $\alpha_n,\beta_n,\gamma_n$. In the
present case, it is not had to show that, letting
$P(\alpha,\beta,\gamma;\mu)$ denote the right hand side of
Eq.~(\ref{eq:Duality}), the $\sup_{\alpha,\beta,\gamma}[\inf_{\mu}P]$
is actually achieved at finite $\alpha_*,\beta_*,\gamma_*$, by showing
that the sequence must remain bounded and using standard compactness arguments.} $\alpha_*,\beta_*,\gamma_*\in\reals$
such that the above holds with equality. Note that for such choice
$f(w)-\alpha_*-\beta_*\, w^2-\gamma_*\,g(w)\ge 0$ for all
$w\in [0,\infty]$, because otherwise the infimum is
$-\infty$. Under this condition, the infimum term in
Eq.~(\ref{eq:Duality}) is zero, and hence we must have
\begin{align}
\alpha_* = \inf_{w\in\reals_{\ge 0}}\Big[f(w)-\beta_*\,
w^2-\gamma_*\,g(w)\Big]\, ,\label{eq:GammaStar}
\end{align}
because otherwise we could increase the lower bound
Eq.~(\ref{eq:Duality}) by increasing $\alpha$. Further, since the
right-hand side is an analytic function of $w$, the infimum in
Eq.~(\ref{eq:GammaStar}) is achieved on a finite set $S_*\in
[0,\infty]$, and the minimizer $\mu_*$ of problem (\ref{eq:MuProblem})
has support $\supp(\mu_*) \subseteq S_*$ because otherwise the
infimum in the lower bound (\ref{eq:Duality}) would not be achieved.

Next we claim that $S_* \subseteq \{0,w_*,\infty\}$ for some finite
$a\in\reals_{>0}$. Indeed, let $h(w) \equiv f(w)-\beta_*\,
w^2-\gamma_*\,g(w)$.   It follows from Eqs.~(\ref{eq:f}) and
(\ref{eq:g})
that 
\begin{align}
h'(w) = (1-2\gamma_*)\phi(w) +2(1-\gamma_*)w\Phi(w) -2\beta_*\, w\, .
\end{align}
Assume $\gamma_*\neq 1/2$. We then have $h'(w) = 0$ for some finite $w\in \reals_{>0}$ if and only
if
\begin{align}
\frac{\phi(w)}{w} + \Big(\frac{2-2\gamma_*}{1-2\gamma_*}\Big)\Phi(w) =
\Big(\frac{2\beta_*}{1-2\gamma_*}\Big)\, w\, .
\end{align}
By Lemma \ref{lemma:TwoSols} this has at most two solutions
$w_1,w_2$. If on the other hand $\gamma_* = 1/2$, then the above
equation reduces to $\Phi(w) = 2\beta_*$ which has at most one
solution. 
In both cases, at most one solution --call it $w_*$-- is a local minimum of $h$.

We conclude that $S_* \subseteq \{0,w_*,\infty\}$, and therefore the
value of the problem  (\ref{eq:MuProblem}) is achieved by a measure of
the form 
\begin{align}
\mu_* = p_0 \, \delta_0 + p_1\, \delta_{w_*} + p_{2}\delta_{\infty}
\end{align}
The three constraints imply the following relations
\begin{align}
2p_0 +p_1\, \frac{1}{g(w_*)} & = \beps\, ,\\
p_1\, \frac{w_*^2}{g(w_*)}+p_2&= x^2\, ,\\
p_0+p_1+p_2 & = y^2\, ,\label{eq:yProof}
\end{align}
and the value is 
\begin{align}
x\cF(x,y) = p_1\, \frac{f(w_*)}{g(w_*)} + p_2\, .\label{eq:cFProof}
\end{align}
The proof is completed by the change of variables $p_1 = g(w_*)\,
\eps$, $p_0 = (\beps-\eps)/2$, $p_2=qx^2$, $w_* = v_*x$, $p_2 = qx^2$.
With these substitutions  Eq.~(\ref{eq:yProof}) yields (\ref{eq:YF}),
and Eq.~(\ref{eq:cFProof}) yields (\ref{eq:CalF}).
\end{proof}

We are now in position to prove Theorem \ref{thm:LeastFavorableF},
that is the main result in this section.
\begin{proof}[Proof of Theorem  \ref{thm:LeastFavorableF}]
By Lemma  \ref{lemma:Constrained} and Eq.~(\ref{eq:FcF}), we have,
for any $\mu_V\in\cF_{\beps}$,
\begin{align}
\F_V(x) \ge \inf_{q,v, \eps} \frac{qx+\int\E\{v(xv+G)_+\}\,
  \mu_{\eps,v_*}(\de v)}
{\sqrt{qx^2+\int\E\{(xv+G)_+^2\}\, \mu_{\eps,v_*}(\de v)}}\, .
\end{align}
where the infimum is over $q\in [0,1]$, $\eps\in (0,\beps]$, $v_* =
\sqrt{(1-q)/\eps}$. Our claim is equivalent to saying that the infimum
on the right hand side is achieved when $q=0$.

 Since $x>0$ is given, we will can regard the right-hand side as a
function of $w=v_*x$ and $\eps$, and substitute $qx^2 = x^2-\eps
w^2$. 
We then define the function
\begin{align}
G(w,\eps) = \frac{x^2-\eps\,w^2+\eps\E\{w(w+G)_+\}}
{\sqrt{x^2-\eps\,w^2+(1-\eps)/2+\eps\E\{(w+G)_+^2\}}}\, .
\end{align}
More explicitly 
\begin{align}
G(w,\eps) & =\frac{x^2+\eps\,  b(w)}
{\sqrt{x^2+(1/2)+\eps(b(w)+\Phi(w)-(1/2))} }\, ,\\
b(w) & \equiv w^2\big(\Phi(w)-1\big)+w\,\phi(w)\, ,\label{eq:bw_definition}
\end{align}
which needs to be optimized over $\eps\in (0,\beps]$, and $w\in
[0,x/\sqrt{\eps}]$. Our claim is equivalent to saying that the minimum
cannot be in the interior of this domain.

Since $G$ is analytic in the mentioned domain, a minimum in the
interior must satisfy $\partial_{w} G (w,\eps) = \partial_{\eps}
G(w,\eps)=0$.
Simple calculus shows that these two conditions are equivalent
--respectively-- to:
\begin{align}
2\eps b'(w) \,\Big[x^2+\frac{1}{2}+\eps
\Big(b(w)+\Phi(w)-\frac{1}{2}\Big)\Big] &=
\big[x^2+\eps\, b(w)\big]\,\eps\big[b'(w)+\phi(w)\big]\, ,\\
2b(w) \,\Big[x^2+\frac{1}{2}+\eps
\Big(b(w)+\Phi(w)-\frac{1}{2}\Big)\Big] &=
\big[x^2+\eps\, b(w)\big]\,\Big[b(w)+\Phi(w)-\frac{1}{2}\Big]\, .
\end{align}
Taking the ratio of these equations, we obtain the necessary
condition
\begin{align}
\frac{b'(w)}{b(w)} = \frac{b'(w)+\phi(w)}{b(w)+\Phi(w)-(1/2)}\, ,
\end{align} 
or equivalently
\begin{align}
\Big(\Phi(x)-\frac{1}{2}\Big)b'(x) = \phi(x)\, b(x)\, .
\end{align} 
Lemma \ref{lemma:ElementaryInequality} establishes that this equation
does not have any solution in $\reals_{>0}$ and hence $G(w,\eps)$ does
not have stationary points in domain $\eps\in(0,\beps]$, $w\in
[0,x/\sqrt{\eps}]$. This finishes our proof.
\end{proof}

\begin{figure}
\begin{center}
\includegraphics[width = 12cm]{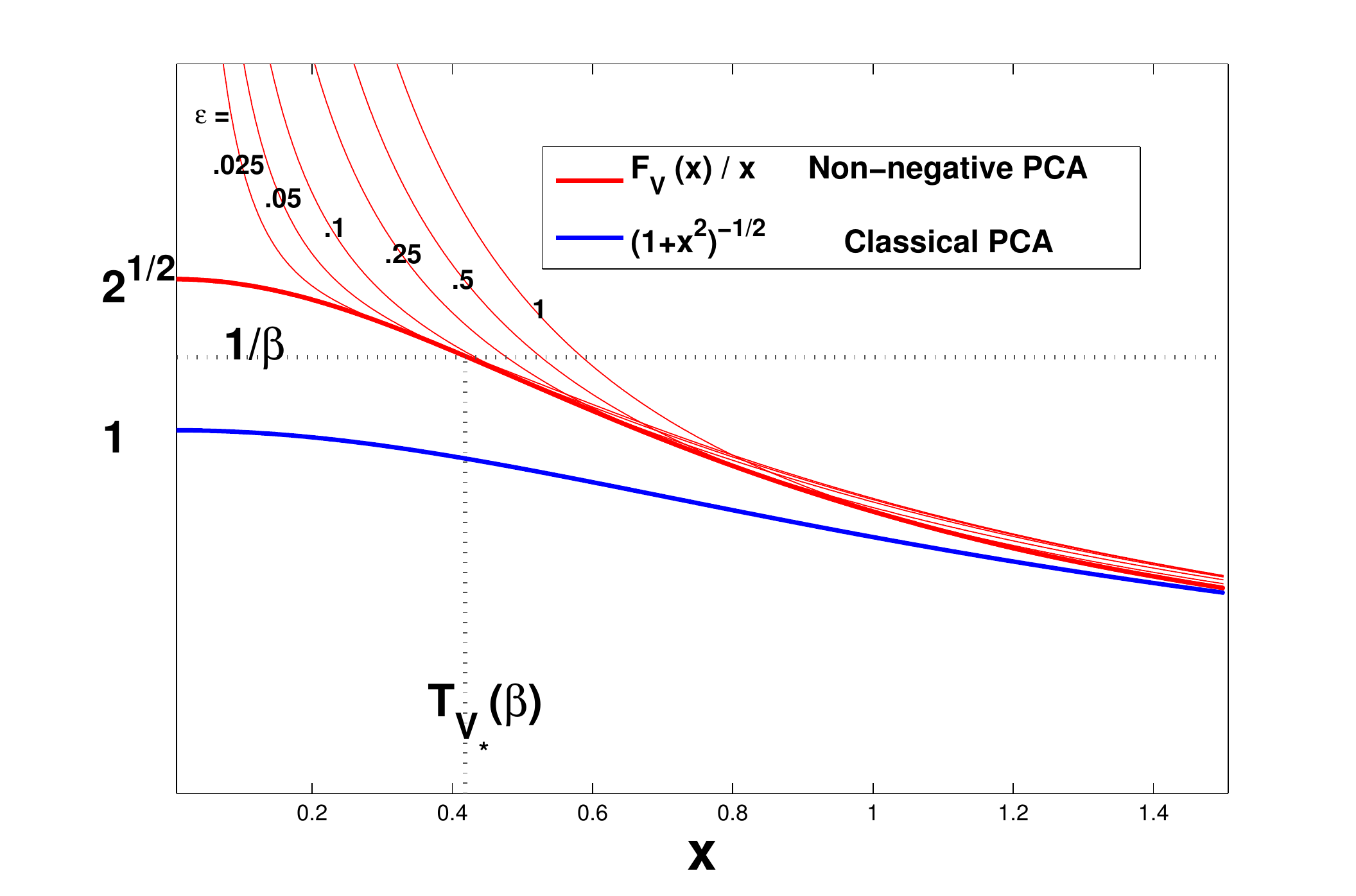}
\caption{The function $\F_V(x) / x$ where $V$ is a mixture of two Dirac $\delta$s at $0$ and at $1/\sqrt \varepsilon$ with various values of $\varepsilon$, and the worst case curve $\F_*(x)/x$. The analogue curve in the case of classical PCA is drawn in blue and the construction of $\T_*(\beta)$ for $\beta \in (1/\sqrt 2,1)$ is illustrated with dashed lines.}\label{fig:symWorstCase}
\end{center}
\end{figure}
We conclude with a Corollary of Theorem \ref{thm:LeastFavorableF}.
(Figure \ref{fig:symWorstCase} provides an illustration of the
argument used in the proof.)
\begin{corollary}\label{coro:LeastFavorableT}
Fix $\beta\in (0,\infty)$. Then for any random variable $V$ with probability
distribution $\mu_V\in\cP_{\beps}$, we have
\begin{align}
\T_V(\beta) \ge \inf_{\eps\in (0,\beps]} \T_{\eps} (\beta)\, . \label{eq:WorstCaseT}
\end{align}
Further, for any $\beta>1/\sqrt{2}$, the infimum on the right-hand
side is achieved at some $\eps_*\in (0,\beps]$.
\end{corollary}
\begin{proof}
Assume the claim (\ref{eq:WorstCaseT}) does not hold. Then there
exists $\mu_V\in \cP_{\beps}$  such that
$\T_V(\beta)<\T_{\eps}(\beta)$ for all $\eps\in (0,\beps]$. Now, on the
one hand, by definition we have 
\begin{align}
\frac{1}{\beta} = \frac{\F_V(\T_V(\beta))}{\T_V(\beta)} =
\frac{\F_\eps(\T_\eps(\beta))}{\T_\eps(\beta)}  \, .
\end{align}
On the other hand, by Theorem \ref{thm:LeastFavorableF}, there exists
$\eps_0\in (0,\beps]$ such that $\F_V(x)\ge\F_{\eps_0}(x)$ for $x=
\T_V(\beta)\in \reals_{>0}$. Using this fact, the contradiction
assumption $\T_V(\beta)<\T_{\eps_0}(\beta)$, and the fact that
$x\mapsto \F_{\eps_0}(x)/x$ is strictly decreasing on $\reals_{>0}$ as
shown in the proof of 
Lemma \ref{lem:Twelldefined}, we get
\begin{align}
\frac{1}{\beta} = \frac{\F_V(\T_V(\beta))}{\T_V(\beta)} \ge
 \frac{\F_{\eps_0}(\T_V(\beta))}{\T_V(\beta)}>
 \frac{\F_{\eps_0}(\T_{\eps_0}(\beta))}{\T_{\eps_0}(\beta)}
 =\frac{1}{\beta}\, .
\end{align}
We therefore reached a contradiction, which proves the claim
(\ref{eq:WorstCaseT}).

In order to prove that the infimum is achieved at some $\eps_*\in
(0,\beps]$, note that $\eps\mapsto\T_{\eps}(\beta)$ is clearly
continuous and, by Lemma \ref{rk:sparseRegimeFG},
\begin{align}
\lim_{\eps\to 0}\T_{\eps}(\beta) = \sqrt{\beta^2-(1/2)}\, .
\end{align}
It is therefore sufficient to show that $\eps\to\T_{\eps}(\beta)$ is
decreasing for $\eps$ small enough. By an argument similar to the
above, this follows if we show that $\eps\mapsto \F_{\eps}(x)$ is
decreasing for $x=\T_0(\beta) = \sqrt{\beta^2-(1/2)}$ and  $\eps$
small enough. Indeed using the definition (\ref{eq:FepsDefinition})
and recalling that $\Phi(w) = 1-O(\phi(w))$ as $w\to\infty$, we get,
for every fixed $x>0$
\begin{align}
\F_{\eps}(x) = \frac{x}{\sqrt{\frac{1+\eps}{2}+x^2}}
+O\big(\phi(x/\sqrt{\eps})\big)\, ,
\end{align}
which is of course decreasing in $\eps$ for $\eps\in [0,c(x)]$ with $c(x)>0$.
\end{proof}

\subsubsection{Proof of Theorems \ref{th:Fworstcase}}
\label{subsec:WorstSymm}

First let $\beta \in [0, 1/\sqrt 2]$. 
We then set $\ell = \lfloor n\eps\rfloor$ and
\begin{align}
(\bvz)_i = \begin{cases}
1/\sqrt{\ell} & \mbox{ for } i\in \{1,2,\dots,\ell\},\\
0 & \mbox{ for } i\in\{\ell+1,\dots,n\} \, .
\end{cases}
\end{align}
Then of course $\{\bvz(n)\}_{n\ge 0}$ converges in empirical
distribution to $\mu_{\eps}$ and, by Theorem \ref{th:mainSym} 
\begin{align}
\lim_{n\to\infty}\<\bv^+,\bv_0\>&=
\F_\eps(\T_\eps(\beta))\,   ,
\end{align}
with $\T_{\eps}(\beta)$ the only non-negative solution of $x =
\beta\F_{\eps}(x)$. By Lemma \ref{rk:sparseRegimeFG}
(cf. Eqs.~(\ref{eq:Flimit}), (\ref{eq:TepsToZero})), we have
$\lim_{\eps\to 0} \F_\eps(\T_\eps(\beta)) = 0$, and hence
\begin{align}
\lim_{\eps\to 0}\lim_{n\to\infty}\<\bv^+,\bv_0\>&= 0\, .
\end{align}
The claim (\ref{eq:WorstCaseLowSNR}) then follows by replacing $\eps$,
by sequence  $\{\eps_n\}_{n\ge 1}$ with $\eps_n\downarrow 0$
sufficiently slowly. The limit vanishes in this case as well by a
standard argument.

Next consider the claim (\ref{eq:WorstCaseHighSNR}).
We let $\eps_*$ be the value achieving the infimum  in
Eq.~(\ref{eq:WorstCaseT}), which exists by Corollary \ref{coro:LeastFavorableT}.
 It is obvious (by another application of Theorem \ref{th:mainSym}) that equality holds
for the stated choice of $\bvz(n)$. 
Assume by contradiction that the
inequality (\ref{eq:WorstCaseHighSNR}) does not hold for some sequence
$\{\bvz(n)\}$. Then, by tightness, there exists a subsequence along
which the limit on the left hand side exists, and that converges in
empirical distribution to a certain probability measure
$\mu_V\in\cP_{\beps}$. Hence, using Theorem \ref{th:mainSym}, it
follows that (using the
definition of $\T_V(\beta)$)
\begin{align}
\T_V(\beta)<T_{\eps_{*}}(\beta) = \inf_{\eps\in (0,\beps]}\T_{\eps}(\beta)\, .
\end{align}
This contradicts corollary   \ref{coro:LeastFavorableT}, hence proving
our claim.

\subsubsection{Proof of Theorem \ref{th:FworstcaseRec}}
\label{subsec:WorstRec}

The proof of Theorem \ref{th:FworstcaseRec} is very similar to the
proof of Theorem \ref{th:Fworstcase}, and therefore we will only
sketch the first steps.

First if $\beta \in [0,\sqrt{\alpha /   2})$, 
we set $p= \lfloor n\eps\rfloor$
\begin{align}
(\bvz)_i = \begin{cases}
1/\sqrt{\ell} & \mbox{ for } i\in \{1,2,\dots,\ell\},\\
0 & \mbox{ for } i\in\{\ell+1,\dots,p\} \, .
\end{cases}
\end{align}
Then  $\{\bvz(p)\}_{p\ge 0}$ converges in empirical
distribution to $\mu_{\eps}$ and, by Theorem \ref{th:mainRec},
\begin{align}
\lim_{p\to\infty}\<\bv^+,\bvz\>&=
\F_\eps(\S_\eps(\beta,\alpha)/\sqrt{\alpha})\,   .
\end{align}
with $\S_{\eps}(\beta,\alpha)\equiv \S_V(\beta,\alpha)$ for
$V\sim\mu_{\eps}$ is given by Definition \ref{def:functionsFGRTS},
i.e. is the only positive solution $x$ of
\begin{align}
x =
\frac{\beta\F_{\eps}(x/\sqrt{\alpha})}{\sqrt{1+\beta\F_{\eps}(x/\sqrt{\alpha})^2}}\, .
\end{align}
 By Lemma \ref{rk:sparseRegimeFG}
(cf. Eqs.~(\ref{eq:Flimit}), (\ref{eq:SepsToZero})), we have
$\lim_{\eps\to 0} \F_\eps (\S_\eps(\beta,\alpha)/\sqrt{\alpha})= 0$, and hence
\begin{align}
\lim_{\eps\to 0}\lim_{p\to\infty}\<\bv^+,\bv_0\>&= 0\, .
\end{align}
The claim follows by taking $\eps=\eps(p)\to 0$ slowly enough.

Next consider $\beta>\sqrt{\alpha/2}$. By the same argument as in
Corollary \ref{coro:LeastFavorableT}, we have, for any $\mu_V\in\cP_{\beps}$,
\begin{align}
\S_V(\beta,\alpha) \ge \inf_{\eps\in (0,\beps]} \S_{\eps} (\beta,\alpha)\, . \label{eq:WorstCaseS}
\end{align}
Further, for any $\beta>1/\sqrt{2}$, the infimum on the right-hand
side is achieved at some $\eps_*\in (0,\beps]$. 
We then take $V_{*}\sim \mu_{\eps_{*}}$.

Assuming  that the claim (\ref{eq:WorstCaseRecS}) is
false, we can construct by the same tightness argument used in the
previous section, a probability distribution $\mu_V$, such that
$\S_V(\beta,\alpha) < \S_{\eps_*} (\beta,\alpha)$. This contradicts
Eq.~(\ref{eq:WorstCaseS}), which proves our claim.

\section*{Acknowledgements}

This work was
partially supported by  the NSF grant CCF-1319979 and
the grants AFOSR/DARPA FA9550-12-1-0411 and FA9550-13-1-0036.
%
%
\appendix

\section{State evolution: Proofs of Proposition \ref{prop:AMPsym} and
Proposition \ref{lem:diffGoesToZeroSymm}}
\label{app:StateEvolution}

In this appendix we characterize the high-dimensional behavior 
of AMP as per Proposition \ref{prop:AMPsym} and
Proposition \ref{lem:diffGoesToZeroSymm}. The analogous 
results for rectangular matrices (namely, Propositions \ref{prop:SEgeneral}
and \ref{lem:diffGoesToZero}) follow from very similar arguments which
we omit here.  

It is convenient to first state two simple facts.  The first one
allows to control small perturbations of a given iterative scheme.
\begin{lemma}\label{lemma:PerturbedAMP}
Let $\bX$ be as in the statement of Proposition \ref{prop:AMPsym}, and
the sequences $\{\bu^t\}_{t\ge 0}$, $\{\btu^t\}_{t\ge 0}$ be defined by
the recursions
\begin{align}
\bu^{t+1} &= \bX\, g_t(\bu^t)-\onsa_tg_{t-1}(\bu^{t-1})\, ,\\
\btu^{t+1} &= \bX\, g_t(\btu^t)-\onsa_tg_{t-1}(\btu^{t-1})+\bDelta^t\, ,
\end{align}
where $\onsa_t\in\reals$ and $g_t:\reals^n\to\reals^n$.

Assume that $\lim_{n\to\infty}\|\bu^0-\btu^{0}\|_2/\sqrt{n} = 0$, $\lim\sup_{n\to\infty}\|\bu^0\|_2^2/n<\infty$ and,
for every $t\in \{0,\dots,T\}$, we have the following, almost surely
\begin{enumerate}
\item $\lim_{n\to\infty}\|\bDelta^t\|_2^2/n =0$.
\item $\lim\sup_{n\to\infty}|\onsa^t| <\infty$.
\item $g_t$ is Lipschitz continuous with bounded Lipschitz
  constant. Namely, there exists $L_t\in \reals$ independent of $n$ such that
$\|g_t(\bu)-g_t(\bu')\|_2\le L_t\|\bu-\bu'\|_2$ for all
$\bu,\bu'\in\reals^n$.
\end{enumerate}
Then, for all $t\in\{0,1,\dots,T+1\}$ we have
\begin{align}
\lim_{n\to\infty}\frac{1}{n}\|\bu^t-\btu^t\|_2^2 &= 0\, ,\label{eq:PerturbedAMP}\\
\lim\sup_{n\to\infty}\frac{1}{n}\|\bu^t\|_2^2 &<\infty\, .\label{eq:BoundedAMP}
\end{align}
\end{lemma}
\begin{proof}
The proof is immediate by induction over $t$.
We will prove Eq.~(\ref{eq:PerturbedAMP}): Eq.~(\ref{eq:BoundedAMP})
follows by a similar argument. The case $t=0$ holds by
assumption.
In order to prove the induction step, note that
$\|\bX\|_2\le \beta+\|\bZ\|_2\le \beta+3$ with probability larger that
$1-c^{-1}e^{-c\, n}$ for some $c>0$ \cite{Guionnet}.
By triangular
inequality 
\begin{align}
\|\bu^{t+1}-\btu^{t+1}\|_2 &\le
\|\bX\|_2\big\|g_t(\bu^{t})-g_t(\btu^{t})\big\|_2+
|\onsa_t|\,
\big\|g_{t-1}(\bu^{t-1})-g_{t-1}(\btu^{t-1})\big\|_2+\|\bDelta^t\|_2\\
&\le L(\beta+3)
\big\|\bu^{t}-\btu^{t}\big\|_2+|\onsa_t|\,\big\|\bu^{t-1}-\btu^{t-1}\big\|_2
+\|\bDelta^t\|_2\, ,
\end{align}
where the second inequality holds with probability at least $1-c^{-1}e^{-c\, n}$.
The induction claim follows by dividing the above inequality by $\sqrt{n}$.
\end{proof}

The second remark allows to establish limit results as in
Proposition \ref{prop:AMPsym}, once they have been established for a
perturbed sequence.
\begin{lemma}\label{lemma:Perturb}
Assume that the sequences of vectors $\bu=\bu(n)$, $\btu= \btu(n)$
satisfy
\begin{align}
\lim_{n\to\infty}\frac{1}{n}\|\bu(n)-\btu(n)\|_2^2 &= 0\, ,\label{eq:PerturbAssumption1}\\
\lim\sup_{n\to\infty}\frac{1}{n}\|\btu(n)\|_2^2  &<\infty\, ,\label{eq:PerturbAssumption2}
\end{align} 
and further assume $\buz = \buz(n)$ be such that $\sup_n\|\buz(n)\|_2<\infty$.
If $\lim_{n\to\infty}n^{-1}\sum_{i=1}^n\psi(\btu_i,\sqrt{n}(\buz)_i)$
exists for some  pseudo-Lipschitz function $\psi$,
then
\begin{align}
\lim_{n\to\infty}\frac{1}{n}\sum_{i=1}^n\psi(\bu_i,\sqrt{n}(\buz)_i) =
\lim_{n\to\infty}\frac{1}{n}\sum_{i=1}^n\psi(\btu_i,\sqrt{n}(\buz)_i)\, .
\end{align}
\end{lemma}
\begin{proof}
Using the pseudo-Lipschitz property of $\psi$, and Cauchy-Schwartz,
we get
\begin{align}
\frac{1}{n}\sum_{i=1}^n
\Big|\psi(\bu_i,\sqrt{n}(\buz)_i) -\psi(\btu_i,\sqrt{n}(\buz)_i) \Big|
&\le 
\frac{L}{n}\sum_{i=1}^n\big(1+2\sqrt{n}|(\buz)_i|+|\bu_i|
+|\btu_i|\big)\, |\bu_i-\btu_i|\\
&\le
\frac{L}{n}\big(\sqrt{n}+2\sqrt{n}\|\buz\|_2+\|\bu\|_2+\|\btu\|_2\big)
\,\|\bu-\btu\|\, .
\end{align}
By Eqs.~(\ref{eq:PerturbAssumption1}) and (\ref{eq:PerturbAssumption2}), 
$\lim\sup_{n\to\infty}\frac{1}{n}\|\btu^t\|_2^2 <\infty$. Using this
fact together with the other assumptions, we get from the last
inequality
\begin{align}
\lim\sup_{n\to\infty}\frac{1}{n}\sum_{i=1}^n
\Big|\psi(\bu_i,\sqrt{n}(\buz)_i) -\psi(\btu_i,\sqrt{n}(\buz)_i) \Big|
=0\, ,
\end{align}
which proves our claim.
\end{proof}

\subsection{Proof of Proposition \ref{prop:AMPsym}}
\label{app:AMPsym}

The proof consists in modifying the AMP sequence $\{\bv^t\}_{t\ge 0}$ 
as to reduce ourselves to the setting of \cite{BM-MPCS-2011}.
The first step consists in  introducing a sequence $\{\bw^t\}_{t\ge
  0}$ defined by $\bw^0 = (1,1,\dots,1)^{\sT}$, $\bw^{-1} = 0$ and letting,
for all $t\ge 0$, 
\begin{align}
\bw^{t+1} = \bX\,(\bw^t)_+-\tons_t\, (\bw^{t-1})_+\, ,\label{eq:SimplifiedAMP}
\end{align}
where $\tons_t = \|(\bw^t)_+\|_0/n$.
The relation between this recursion  and the original one is quite
direct: they differ only by a normalization factor.
\begin{lemma}\label{lemma:VvsW}
Let $\{\bw^t\}_{t\ge 0}$ be defined per Eq.~(\ref{eq:SimplifiedAMP})
and $\{\bv^t\}_{t\ge 0}$ be the AMP sequence, as per
(\ref{eq:AMPsym}). Then, for all $t\ge 1$ we have 
\begin{align}
\bv^{t} = \sqrt{n}\, \frac{\bw^t}{\|(\bw^{t-1})_+\|_2}\, .
\end{align}
\end{lemma}
\begin{proof}
The proof is by induction over the number of iterations. Let us first assume that it holds 
for all iterations until $t$, and  prove it for iteration
$t+1$. Multiplying Eq.~(\ref{eq:SimplifiedAMP}) by
$\sqrt{n}/\|(\bw^{t})_+\|_2$, we get
\begin{align}
\sqrt{n}\, \frac{\bw^{t+1}}{\|(\bw^{t-1})_+\|_2} = \bX\,
\frac{(\bw^t)_+\sqrt{n}}{\|(\bw^{t})_+\|_2}-
\frac{1}{\sqrt{n}}\,
\frac{\|(\bw^t)_+\|_0}{\|(\bw^t)_+\|_2}\,(\bw^{t-1})_+\, .\label{eq:InductionWV}
\end{align}
Note that the induction hypothesis implies $\bv^t = c\,\bw^t$ for some
constant $c$, and hence 
\begin{align}
\sqrt{n}\frac{ (\bw^t)_+}{\|(\bw^t)_+\|_2}= \sqrt{n}\frac{ (\bv^t)_+}{\|(\bv^t)_+\|_2} =f(\bv^t) \,. \label{eq:FirstWV}
\end{align}
By the same argument and using $\|(\bv^t)_+\|_2 =
\sqrt{n}\|(\bw^t)_+\|_2/\|(\bw^{t-1})_+\|_2$, we get
\begin{align}
\frac{1}{\sqrt{n}}\,
\frac{\|(\bw^t)_+\|_0}{\|(\bw^t)_+\|_2}\,(\bw^{t-1})_+ &=
\frac{1}{n}\,
\frac{\|(\bw^t)_+\|_0\|(\bw^{t-1})_+\|_2}{\|(\bw^t)_+\|_2}\,f(\bv^{t-1})
\\
&= \frac{1}{\sqrt{n}}\,
\frac{\|(\bv^t)_+\|_0}{\|(\bv^t)_+\|_2}\,f(\bv^{t-1}) =\ons_t\,
f(\bv^{t-1})\, .\label{eq:SecondWV}
\end{align}
Using Eqs.~(\ref{eq:FirstWV}) and (\ref{eq:SecondWV}) in
Eq.~(\ref{eq:InductionWV}), we obtain
\begin{align}
\sqrt{n}\, \frac{\bw^{t+1}}{\|(\bw^{t-1})_+\|_2} =
\bX\,f(\bv^t)-\ons_t\,f(\bv^{t-1})\, .
\end{align}
The induction step is completed by comparing this with
(\ref{eq:AMPsym}). The base case follow easily by a similar argument.
\end{proof}

As a second step, we introduce a sequence $\{\bs^t\}_{t\ge 0}$ defined
as  follows. First , we let $\mu_t$, $\sigma_t$ be scalars given by
\begin{align}
\mu_{t+1} & = \beta\, \E\big\{V(\mu_t V+\sigma_t G)_+\big\}\,,\label{eq:StateEvolutionSplit1}\\
\sigma^2_{t+1} & = \E\big\{(\mu_t V+\sigma_t G)_+^2\big\}\,. \label{eq:StateEvolutionSplit2}
\end{align}
with initial conditions $\mu_1 = \beta\E(V)$ and $\sigma_1 = 1$. 
Note that by Cauchy-Schwartz
$\mu_{t+1}\le\beta\sqrt{\mu_t^2+\sigma_t^2}$ and $\sigma_{t+1}^2\le
\mu_t^2+\sigma_t^2$, whence $\mu_t,\sigma_t<\infty$ for all $t$. 
Further, since $G\ge 0$ with probability $1/2$, we also have
$\mu_{t+1}\ge \beta\mu_t/2$, $\sigma^2_{t+1}\ge \mu_t^2/2$,
whence $\mu_t,\sigma_t\in (0,\infty)$ for all $t$. 

Using
these quantities (and recalling that $\bX = \beta\, \bvz\bvz^{\sT} +\bZ$ with
$(\bZ)_{ij}\sim\normal(0,1)$, i.i.d. for $i<j$), we define
\begin{align}
\bs^{t+1} &= \bZ\, h_t(\bs^t,\bvz\sqrt{n})-\onsd_t\, h_{t-1}(\bs^{t-1},\bvz\sqrt{n})\, ,\\
& h_t(x,y) \equiv (x+\mu_ty)_+ \, ,\label{eq:IterationS}\\
&\onsd_t \equiv \frac{1}{n}\|(\bs^t+\mu_t\bvz\sqrt{n})_+ \|_0\, ,
\end{align}
As usual, here  $h_t(\bs^t,\bvz\sqrt{n})$ is interpreted as the
component-wise application of $h_t$. The initial condition is $\bs^1 =
\bw^1-\mu_1\bv_0$. This iteration is in the form of 
\cite[Theorem 1]{javanmard2013state} (and analogous to \cite[Theorem 4]{BM-MPCS-2011}), which implies immediately the
following.
\begin{lemma}\label{lemma:StateEvolutionS}
For any $t\ge 1$ and any pseudo-Lipshitz function
$\psi:\reals\times\reals\to\reals$, we have, almost surely
\begin{align}
\lim_{n\to\infty}\frac{1}{n}\sum_{i=1}^n\psi\big(\bs^t_i,\sqrt{n}(\bvz)_i\big)
=\E\big\{\psi(\sigma_tG,V)\big\}\, ,
\end{align}
where expectation is with respect to $G\sim\normal(0,1)$ independent
of $V$.
\end{lemma}

The sequences $\{\bs^t\}_{t\ge 0}$ and $\{\bw^t\}_{t\ge 0}$ are in
fact closely related as we show next.
\begin{lemma}\label{lemma:StateEvolutionW}
For any $t\ge 1$ and any pseudo-Lipshitz function
$\psi:\reals\times\reals\to\reals$, we have
\begin{align}
\lim_{n\to\infty}\frac{1}{n}\big\|\bw^t-\mu_t\bvz-\bs^t\big\|_2&=0\,
,\\
\lim_{n\to\infty}\frac{1}{n}\sum_{i=1}^n\psi\big(\bw^t_i,\sqrt{n}(\bvz)_i\big)
&=\E\big\{\psi(\mu_tV+\sigma_tG,V)\big\}\, ,
\end{align}
where expectation is with respect to $G\sim\normal(0,1)$ independent
of $V$.
\end{lemma}
\begin{proof}
Define $\bts^t=\bw^t-\mu_t\bvz\sqrt{n}$. 
Then Eq.~(\ref{eq:SimplifiedAMP})
implies immediately
\begin{align}
\bts^{t+1}   &= \bZ\, h_t(\bs^t,\bvz\sqrt{n})-\onsd_t\,
h_{t-1}(\bs^{t-1},\bvz\sqrt{n})+\bDelta^t\, ,\label{eq:IterationStilde}\\
\bDelta^t & \equiv
\big(\beta\<\bvz,(\mu_t\bvz\sqrt{n}+\bts_t)_+\>-\mu_{t+1}\sqrt{n}\big)\, \bvz
+(\onsd_t-\tons_t)\, (\mu_{t-1}\bvz\sqrt{n}+\bts^{t-1})_+\, .
\end{align}

Next note that  our claim is equivalent to the following holding for every
pseudo-Lipshitz $\psi$ and every iteration number $\ell$:
\begin{align}
\lim_{n\to\infty}\frac{1}{n}\big\|\bts^{\ell}-\bs^{\ell}\big\|_2^2&=0\,
,\label{eq:SvsStilde}\\
\lim_{n\to\infty}\frac{1}{n}\sum_{i=1}^n\psi\big(\bts^{\ell}_i,\sqrt{n}(\bvz)_i\big)
&=\E\big\{\psi(\sigma_{\ell}G,V)\big\}\, .\label{eq:SEInduction}
\end{align}
We prove this by induction over the iteration number.
Assume that the claim indeed holds for all $\ell\in\{1,\dots,t$. Then comparing 
Eq.~(\ref{eq:IterationS}) and Eq.~(\ref{eq:IterationStilde}) we obtain
--by Lemma \ref{lemma:PerturbedAMP}-- that Eq.~(\ref{eq:SvsStilde})
holds for
all $\ell\in\{1,\dots,t,t+1\}$ provided we can prove that
\begin{align}
\lim_{n\to\infty}\frac{1}{n}\|\bDelta^\ell\|_2^2 = 0
\end{align}
for all $\ell\in\{1,\dots,t\}$. Now we have
\begin{align}
\frac{1}{n}\|\bDelta^{\ell}\|_2^2&\le D^{\ell}_1+D^{\ell}_2\, ,\label{eq:BoundDelta}\\
D^{\ell}_1& \equiv 2\big(\beta\<\bvz,(\mu_{\ell}\bvz+\bts_{\ell}
n^{-1/2})_+\>-\mu_{\ell+1}\big)^2\, ,\\
D^{\ell}_2&\equiv 4
(\onsd_\ell-\tons_\ell)^2\Big(\mu_{\ell-1}^2+\frac{1}{n}\|\bts^{\ell-1}\|_2^2\Big)\,
.
\end{align}
Now, using $\psi(x,y) = (\mu_{\ell}y+x)_+y$ in
Eq.~(\ref{eq:SEInduction}) we get for all $\ell\in\{1,\dots,t\}$,
almost surely
\begin{align}
\lim_{n\to\infty}\beta \<\bvz,(\mu_{\ell}\bvz+\bts_{\ell}
n^{-1/2})_+\>& =\beta\,\lim_{n\to\infty} \frac{1}{n}\sum_{i=1}^n \psi\big(\bts^{\ell}_i,\sqrt{n}(\bvz)_i\big)\\
&=\beta\E\big\{V(\mu_{\ell}V+\sigma_{\ell}G)_+\big\} = \mu_{\ell+1}\, .
\end{align}
In other words $D^{\ell}_1\to 0$ almost surely.

Using again Eq.~(\ref{eq:SEInduction}) we have
$n^{-1}\|\bts^{\ell-1}\|_2^2\to\sigma_{\ell-1}^2$ almost surely.
Therefore,
since
$\mu_{t}$ is finite for all $t$, we get
$|D^{\ell}_2|\le C|\onsd_\ell-\tons_\ell|$ for some constant $C$
bounded uniformly in $n$. Finally,
fix $\delta>0$ and let
\begin{align}
\psi_{\delta}(x) =\begin{cases}
1 & \mbox{ if $x>\delta$,}\\
x/\delta &\mbox{ if $x\in (0,\delta)$,}\\
0 & \mbox{ otherwise.}
\end{cases}\label{eq:PsiDelta}
\end{align}
Then
\begin{align}
\onsd_\ell-\tons_\ell&=
\frac{1}{n}\sum_{i=1}^n\Big[\ind(\bs^{\ell}_i+\mu_{\ell}(\bvz)_i\sqrt{n}>0)-
\ind(\bts^{\ell}_i+\mu_{\ell}(\bvz)_i\sqrt{n}>0)\Big]\\
& \le \frac{1}{n}\sum_{i=1}^n\Big[\psi_{\delta}(\bs^{\ell}_i+\mu_{\ell}(\bvz)_i\sqrt{n}+\delta)-
\psi_{\delta}(\bts^{\ell}_i+\mu_{\ell}(\bvz)_i\sqrt{n})\Big]\\
& \le \frac{1}{n}\sum_{i=1}^n\Big[\psi_{\delta}(\bs^{\ell}_i+\mu_{\ell}(\bvz)_i\sqrt{n}+\delta)-
\psi_{\delta}(\bs^{\ell}_i+\mu_{\ell}(\bvz)_i\sqrt{n}) +
\frac{1}{\delta}\|\bs^{\ell}_2-\bts^{\ell}_i|\Big]\\
& \le  \frac{1}{n}\sum_{i=1}^n\Big[\psi_{\delta}(\bs^{\ell}_i+\mu_{\ell}(\bvz)_i\sqrt{n}+\delta)-
\psi_{\delta}(\bs^{\ell}_i+\mu_{\ell}(\bvz)_i\sqrt{n})\Big] +
\frac{1}{\delta\sqrt{n}}\|\bs^{\ell}-\bts^{\ell}\|_2\, .
\end{align}
Taking $n\to\infty$ and using Eqs.~(\ref{eq:SvsStilde}),
(\ref{eq:SEInduction}), we conclude that
\begin{align}
\lim\sup_{n\to\infty} (\onsd_\ell-\tons_\ell)&\le
\E\Big\{\psi_{\delta}(\mu_{\ell}V+\sigma_{\ell}G+\delta)-\psi_{\delta}(\mu_{\ell}V+\sigma_{\ell}G)\Big\}\,,
\end{align}
And since $G$ has a density with respect to Lebesgue measure, this
implies, by letting $\delta\to 0$,
$\lim\sup_{n\to\infty}(\onsd_\ell-\tons_\ell)\le 0$.
A lower bound is obtained by a similar argument  yielding
\begin{align}
\lim_{n\to\infty} (\onsd_\ell-\tons_\ell) = 0\, .
\end{align}
and hence $D^\ell_2\to 0$.
By Eq.~(\ref{eq:BoundDelta}) we have $\|\bDelta^{\ell}\|_2^2/n\to 0$
and hence Eq.~(\ref{eq:SvsStilde}) holds for $\ell=t+1$. Finally,
Eq.~(\ref{eq:SEInduction}) follows for $\ell=t+1$ by Lemma
\ref{lemma:Perturb},
Lemma \ref{lemma:StateEvolutionS} and Eq.~(\ref{eq:SvsStilde}) (for $\ell=t+1$).
\end{proof}

Finally, the proof of Proposition \ref{prop:AMPsym} follows
immediately from Lemma \ref{lemma:StateEvolutionW}, using Lemma
\ref{lemma:VvsW}.
Indeed, by applying \ref{lemma:StateEvolutionW} to $\psi(x,y) =
(x)_+^2$, we get, almost surely
\begin{align}
\lim_{n\to\infty} \frac{1}{\sqrt{n}}\|(\bw^t)_+\|_2 =
\sqrt{\E\{(\mu_tV+\sigma_tG)^2_+\}} = \sigma_{t+1}\in(0,\infty)\, .
\end{align}
Hence, for any pseudo-Lipshitz function $\psi$, almost surely
\begin{align}
\lim_{n \to \infty} \frac 1n \sum_{i=1}^n \psi( \bv_i^t ,\sqrt n
 (\bvz)_i) &= 
\lim_{n \to \infty} \frac 1n \sum_{i=1}^n \psi\Big( \sqrt{n}\frac{\bw_i^t}{\|(\bw^{t-1})_+\|_2} ,\sqrt n
 (\bvz)_i) \\
& = \lim_{n \to \infty} \frac 1n \sum_{i=1}^n \psi\Big(\frac{\bw_i^t}{\sigma_{t+1}} ,\sqrt n
 (\bvz)_i\Big) \\
&   \E \left \{ \psi \Big( \frac{\mu_t}{\sigma_t} V + G , V\Big)
\right \} \, .
\end{align}
We conclude by noting that --by comparison of
Eq.~(\ref{eq:SymmetricStateEvolution})
with Eqs.~(\ref{eq:StateEvolutionSplit1}) and
(\ref{eq:StateEvolutionSplit2}) --
it follows that $\tau_t=\mu_t/\sigma_t$ for all $t$.

Finally the claim $\psi(x,y) = \ind(x\le a)$, follows by a standard
argument already used above. Namely, we use the
bounds $\psi_{\delta}(x)\le\ind(x\le a)\le \psi_{\delta}(x+\delta)$,
with the definition in Eq.~(\ref{eq:PsiDelta}), apply the previous
result to the Lipschitz functions $\psi_{\delta}(x)$,
$\psi_{\delta}(x+\delta)$
and eventually let $\delta\to 0$.

\subsection{Proof of Proposition \ref{lem:diffGoesToZeroSymm}}
\label{app:DiffGoesToZeroSymm}

The proof is based on a version of state evolution that 
describes the asymptotic joint distribution of $\bv^t$, $\bv^s$ for
two distinct times $t,s$. 
For this purpose, we define a function $\H_V:[-1,1]\times\reals_+\times
\reals_+\to\reals$
as follows
\begin{align}
\H_V(Q;\tau_1,\tau_2) = \frac{\E\{(\tau_1V+G_1)_+(\tau_2V+G_2)_+\}}
{\sqrt{\E\{(\tau_1V+G_1)_+^2\}\E\{(\tau_2V+G_2)_+^2\}}}\, , 
\end{align}
where expectation is with respect to the centered Gaussian vector 
$(G_1,G_2)$ with $\E\{G_1^2\} = \E\{G_2^2\}=1$, $\E\{G_1G_2\} = Q$,
independent of $V$. 

Let the state evolution sequence $\{\tau_t\}_{t\ge 0}$ be given as per
Eq.~(\ref{eq:SymmetricStateEvolution}), and define recursively
$\{Q_{t,s}\}_{t,s\ge 0}$ by letting
\begin{align}
Q_{t+1,s+1} = \H_V(Q_{t,s};\tau_t,\tau_s)\, .\label{eq:TwoTimesSE}
\end{align}
with initial condition $Q_{1,1}=1$ and, for $t\ge 2$,
\begin{align}
Q_{t,1} =
\frac{\E\{(\tau_{t-1}V+G)_+\}}{\sqrt{\E\{(\tau_{t-1}V+G)_+^2\}}}\, .
\end{align}
Then we have the following extension of state evolution.
\begin{lemma}\label{lemma:TwoTimes}
With the above definitions, let $\psi:\reals^3\to\reals$ be a
pseudo-Lipschitz function. Then, for any $t,s\ge 1$, we have,
almost surely
\begin{align}
\lim_{n\to\infty}\frac{1}{n}\sum_{i=1}^n\psi(\bv^t_i,\bv^s_i,\sqrt{n}(\bvz)_i)
=\E\big\{\psi(\tau_tV+G_t,\tau_sV+G_s,V)\big\}\, ,
\end{align}
where expectation is with respect to the centered Gaussian vector 
$(G_t,G_s)$ with $\E\{G_t^2\} = \E\{G_s^2\}=1$, $\E\{G_tG_s\} = Q_{t,s}$,
independent of $V$. 
\end{lemma}
\begin{proof}
Very similar statements were proven, for instance in \cite[Theorem
4.2]{BM-MPCS-2011} or \cite[Lemma C1]{donoho2013high}. The
construction is always the same, and we will only sketch the first
steps. Thanks to Lemma \ref{lemma:StateEvolutionW}, it is sufficient
to prove that, for $\{\bs^t\}_{t\ge 0}$ defined per
Eq.~(\ref{eq:IterationS}), we have
\begin{align}
\lim_{n\to\infty}\frac{1}{n}\sum_{i=1}^n\psi(\bs^t_i,\bs^s_i,\sqrt{n}(\bvz)_i)
=\E\big\{\psi(\Gamma_t,\Gamma_s,V)\big\}\, ,
\end{align}
where $(\Gamma_t,\Gamma_s)$ is a centered Gaussian vector with
$\E\{\Gamma_t^2\} = \sigma_t^2$, $\E\{\Gamma_s^2\} = \sigma_s^2$,
$\E\{\Gamma_t\Gamma_s\} = \sigma_t\sigma_s\, Q_{t,s}$.

In order to prove the last claim, we fix a maximum time $T$, and
consider all $t,s\in\{0,1,\dots,T-1\}$. We then define
$\br^t\in(\reals^{T})^n$, $t\in\{0,1,\dots,T-1\}$ that we can think of either as a vector of
lenfth $n$, with entries in $\reals^T$, or as a matrix with dimensions
$n\times T$. With the last interpretation in mind, $\br^t$ is defined
as a matrix whose first $t+1$ columns ar $\bs^0$, $\bs^1$, \dots,
$\bs^t$, and the others vanish, namely
\begin{align}
\br^t = \Big[\bs^0\Big|\bs^1\Big|\cdots\Big|\bs^t\Big|0\cdots 0\Big]\, .
\end{align}
Define $\bh_t:\reals^{T+1}\to\reals^{T}$ by letting
\begin{align}
\bh_t(s_0,s_1,\dots,s_{T-1};v) \equiv
\big(s_0,h_0(s_0;v),\dots,h_{t-1}(s_{t-1};v),h_t(s_t;v), 0,\dots,0\big)
;%
\end{align}
Then it is easy to see that Eq.~(\ref{eq:IterationS}) implies
\begin{align}
\br^{t+1} = \bZ\bh_t(\br^t;\sqrt{n}\bvz)
-\bh_t(\br^t;\sqrt{n}\bvz)\onsD_t\, ,
\end{align}
for a certain sequence of matrices $\onsD_t\in\reals^{T\times T}$. The
proof then follows by applying \cite[Theorem 1]{javanmard2013state} to
$\{\br^t\}_{t\ge 0}$.
\end{proof}

The next lemma provides the basic tool for applying the state
evolution method to prove our claim.
\begin{lemma}\label{lemma:TwoTimesTinfty}
Let $\{Q_{t,s}\}_{t,s\ge 1}$ be defined as above using the two times
state evolution recursion (\ref{eq:TwoTimesSE}). Then 
\begin{align}
\lim_{t\to\infty}Q_{t,t+1} = 1\, .
\end{align}
\end{lemma}
Before proving this Lemma, we state a useful general fact (which
appeared already in specific forms in
\cite{BayatiMontanariLASSO,donoho2013high}.
\begin{lemma}\label{lemma:ConvexityTwoTimes}
Let $h:\reals^2\to\reals$ be a Borel function, $W, Z_1,Z_2$  random
variables, and $\prob_q$ a probability distribution such that --under $\prob_q$--
$(Z_1,Z_2)$ is a centered Gaussian vector independent of $W$, with 
covariance given $\E_q(Z_1^2)=\E_q(Z_2^2)=1$ and $E_q\{Z_1,Z_2\}=
q$. Assume $\E\{h(Z_1,W)^2\}<\infty$ and define 
\begin{align}
\cH(q) \equiv \E_q\{h(Z_1,W)h(Z_2,W)\}\, .
\end{align}
Then $q\mapsto\cH(q)$ is non-decreasing and convex on $[0,1]$.
Further, unless $h(x,y)$ is affine in $x$, it is strictly convex.
Finally assuming $h$ is weakly differentiable, and denoting by 
$\partial_1 h$ its derivative with respect to the first argument, we have
\begin{align}
\left.\frac{\de\cH}{\de q}\right|_{q=1} = \E\{[\partial_1h(Z,W)]^2\}\,
.\label{eq:DerivativeTwoTimes}
\end{align}
\end{lemma}
\begin{proof}
First consider the case of $h(x,y)=h(x)$ independent of the
second argument. 
Let $\{X_t\}_{t\ge 0}$ be the stationary Ornstein--Uhlenbeck process with
covariance $\E(X_0X_t) = e^{-t}$. Then 
\begin{align}
\cH(q) = \E\{h(X_0) h(X_t)\}\Big|_{t= \log(1/q)}\, ,
\end{align}
Then we have the spectral representation (for $t= \log(1/q)$
and $c_{\ell} = \<\phi_{\ell},h\>$, $\phi_{\ell}$ the $\ell$-th
eigenfunction of the Ornstein--Uhlenbeck generator
\begin{align}
\cH(q) = \sum_{\ell = 0}^{\infty}c_{\ell}^2\, e^{-\ell\, t} =
\sum_{\ell = 0}^{\infty}c_{\ell}^2\, q^{\ell}\, ,
\end{align}
whence the $\cH$ is non-decreasing and convex.
Strict convexity follows since $c_{\ell}\neq 0$ for some $\ell\ge 2$
as long as $h(x)$ is non-linear. 

Finally, if $h$ depends on its second argument as well,
we have $\cH(q) = \E\{\cH_W(q)\}$, with
$\cH_W(q) \equiv \E_q\{h(Z_1,W)h(Z_2,W)|W\}$. Using independence of
$(Z_1,Z_2)$ and $W$, the previous proof applies to $\cH_W$ for almost
every $W$ and, by linearity, to $\cH(q)$.

Equation~(\ref{eq:DerivativeTwoTimes}) follows by writing
\begin{align}
\cH(q) = \E_q\{h(Z,W)^2\} -\frac{1}{2} \E_q\big\{\big[h(Z_1,W)
-h(Z_2,W)\big]^2\big\}\, ,
\end{align}
with $Z\sim\normal(0,1)$. The claim follows by 
using the representation $Z_1 = aX+bY$, $Z_2= aX-bY$, with $X,Y$
independent standard normal, $a=\sqrt{(1+q)/2}$, $b=\sqrt{(1-q)/2}$,
and Taylor expanding the right hand side in $b$.
\end{proof}

We are now in position to prove Lemma \ref{lemma:TwoTimesTinfty}.
\begin{proof}[Proof of Lemma \ref{lemma:TwoTimesTinfty}]
Recall that  $\lim_{t\to\infty}\tau_t=\T_V(\beta)\in(0,\infty)$,
cf. Lemma \ref{lem:convergenceFixedPointSymm}. Letting $\tau_*\equiv
\T_V(\beta)$, we define $\H_V^*(Q) \equiv \H_V(Q;\tau_*,\tau_*)$, i.e.
\begin{align}
\H_V^*(Q) = \frac{\E\{(\tau_*V+G_1)_+(\tau_*V+G_2)_+\}}
{\sqrt{\E\{(\tau_*V+G_1)_+^2\}\E\{(\tau_*V+G_2)_+^2\}}}\, , 
\end{align}
with $(G_1,G_2)$ a centered Gaussian vector with
$\E(G_1^2)=\E(G_2^2)=1$ and $\E\{G_1G_2\}= Q$. 
By Lemma \ref{lemma:ConvexityTwoTimes}, the function $Q\mapsto
\H_V^*(Q)$ is  strictly convex and monotone increasing
in $[0,1]$. Further we have  $\H_V^*(1) = 1$ and, for $G\sim\normal(0,1)$,
\begin{align}
\left.\frac{\de\phantom{Q}}{\de Q} \H_V^*(Q)\right|_{Q=1} =
\frac{\prob\big(\tau_*V+G\ge
  0\big)}{\E\big\{\big(\tau_*V+G\big)_+^2\big\}}\, .
\end{align}
Note that 
\begin{align}
\E\big\{\big(\tau_*V+G\big)_+^2\big\} &=
\E\big\{\big(\tau_*V+G\big)\big(\tau_*V+G\big)_+\big\} \\
&= \tau_*\E\big\{V(\tau_*V+G)_+\big\}+ \E\big\{G(\tau_*V+G)_+\big\}\\
& \ge\prob\big(\tau_*V+G\ge0\big)\, ,
\end{align}
where the last inequality follows since $V\ge 0$, and applying Stein's
Lemma to the second term. We therefore have
\begin{align}
\left.\frac{\de\phantom{Q}}{\de Q} \H_V^*(Q)\right|_{Q=1} \le 1\, ,
\end{align}
and therefore, by convexity, $\H_V^*(Q)>Q$ for all $Q\in [0,1)$.

Now, for ease of notation, let $Q_t \equiv Q_{t,t+1}$. Note
that $\H_V(Q;\tau_1,\tau_2)\in [0,1]$ for all $Q\in [0,1]$:
indeed, for $Q\ge 0$, the random variables $(\tau_1V+G_1)_+$
and $(\tau_2V+G_2)_+$ are non-decreasing functions of positively
correlated ones, and hence are positively correlated. Therefore
$Q_t\in[0,1]$ for all $t$. Assume by contradiction that $Q_t$ does not
converge to $1$, and let $Q_* \equiv \lim\inf_{t\to\infty} Q_t$. Let
$\{t(k)\}_{k\in\naturals}$  be a subsequence with
$\lim_{k\to\infty}Q_{t(k)} = Q_*$.
 Since $\tau_t\to\tau_*$,  $\H_V$
is continuous and $\H_V^*$ is non-decreasing, we have
\begin{align}
Q_* &= \lim_{k\to\infty} Q_{t(k)}\\
& =\lim\inf_{k\to\infty}\H_V(Q_{t(k)-1};\tau_{t(k)},\tau_{t(k)-1}) \\
&= \lim\inf_{k\to\infty} \H^*_V(Q_{t(k)-1}) \\
& \ge \H_V^*(Q_*)\,.
\end{align}
This contradicts the previous remark that $\H_V^*(Q)>Q$ for all $Q\in
[0,1)$,
and hence proves the claim that $Q_t\to 1$.
\end{proof}

We are now in position to prove our claim
(\ref{lem:diffGoesToZeroSymm}). 
First note that --by triangular inequality-- it is sufficient to
consider the case $\ell=1$.
Using Lemma \ref{lemma:TwoTimes} 
for $\psi(x,y,z) = (x-y)^2$ and $s=t+1$, we get, almost surely
\begin{align}
\lim_{n\to\infty}\frac{1}{n}\|\bv^t-\bv^{t+1}\|_2^2 = 
\E\big\{\big(\tau_t V+G_t-\tau_{t+1}V-G_{t+1}\big)^2\big\} =
(\tau_t-\tau_{t+1})^2 + 2(1-Q_{t,t+1})\, .
\end{align}
Since by Lemma \ref{lem:convergenceFixedPointSymm} the sequence
$\tau_t$ converges to a finite limit as $t\to\infty$, 
we have $\lim_{t\to\infty}(\tau_t-\tau_{t+1}) = 0$. Hence taking the
limit $t\to\infty$ in the last expression and using Lemma \ref{lemma:TwoTimesTinfty}, we obtain the desired result.

\section{Proof of Theorems \ref{th:mainSym} and \ref{th:mainRec}}\label{sec:ProofOfMainTheorems}

\begin{proof}[Proof of Theorem \ref{th:mainSym}] 
For the sake of clarity, we will note the dimension index $n$ for a
matrix $\bX_n\in\reals^{n\times n}$, distributed according to the
\ref{eq:SymmetricModel}. In order to prove
 Eq.~(\ref{eq:lambdaPlusSim}) (i.e. $\lim_{n \to \infty}
 \lambda^+(\bX_n) = \R\sym(\T_V(\beta))$ almost surely),
we need to prove:
\begin{align}
\prob\left [\liminf_{n \to \infty} \lambda^+(\bX_n) \geq
  \R\sym(\T_V(\beta)) \right ] = 1~~\text{and}~~ \prob\left
  [\limsup_{n \to \infty} \lambda^+(\bX_n) \leq \R\sym(\T_V(\beta))
\right ] = 1~.
\end{align}
\begin{itemize}
\item  Theorem \ref{th:lowerBoundsAMPsymmetric} states that there exists a deterministic sequence $\{\delta_t\}_t$ such that $\lim_t \delta_t = 0$ and 
\[ \prob\left [\lim_{n \to \infty} \<\hbv^t, \bX_n \hbv^t \> \geq \R\sym(\T_V(\beta)) - \delta_t \right ] = 1~~.\]
It follows, using $\lambda^+(\bX_n) \geq \<\hbv^t, \bX_n \hbv^t \>$, and taking the intersection of these events for $t \in \naturals$, that  
\[ \prob\left [\liminf_{n \to \infty} \lambda^+( \bX_n) \geq \R\sym(\T_V(\beta)) \right ] = 1~~.\]
\item  Since the function $\bX_n \mapsto \max \left \{ \<\bv,\bX_n\bv\>~:~\bv\geq 0~,~\|\bv\|_2 \leq 1\right \}$ is 1-Lipschitz continuous, then using the upper bound of Lemma \ref{th:upperBoundsSymmetric} and Gaussian isoperimetry, for any $s>0$ we have, with probability at least $1-\exp\{-ns^2 / 2\}$, 
 \[   \lambda^+(\bX_n)  \leq \R\sym(\T_V(\beta)) +s~~.  \]
 Taking $s = \sqrt{(4 \log n)/n}$, with probability at least $1-n^{-2}$, 
 \[   \lambda^+(\bX_n)\leq\R\sym(\T_V(\beta)) + \frac{4 \log n}n~~.\]
 Hence $\limsup_{n\to \infty}\lambda^+(\bX_n) \leq
 \R\sym(\T_V(\beta))$ almost surely by Borel-Cantelli.
\end{itemize}
This concludes the proof of Eq.~(\ref{eq:lambdaPlusSim}).  
Equation (\ref{eq:innerProductSim}) follows immediately from Lemma
\ref{th:upperBoundsSymmetric} since $\lim_{x\to0}\Delta(x) = 0$, and
we know that the sequence $\lambda^+(\bX)$ converges almost surely to
to $\R\sym(\T_V(\beta))$.

 In order to prove the limit behavior as $\eps \to 0$ of Eqs. (\ref{eq:RlimitEpsZeroSym}) and (\ref{eq:FlimitEpsZeroSym}) we refer to Lemma \ref{rk:sparseRegimeFG} in Section \ref{sec:preliminaryProof} that establish the limit behavior of functions of interest $\T_V,\F_V,\G_V$ uniformly over the class of probability distributions $\cP$. We know, thanks to Definition \ref{def:UniformConvergence}, Lemma \ref{rk:sparseRegimeFG}, and uniform continuity on the interval $[0, 1]$ of  the square function $x\mapsto x^2$, and on $\reals_{\geq 0}$ of
 \begin{align*}
\F_0~:~ x \mapsto \frac x{ \sqrt{1/2+x^2}}~,\quad \G_0~:~ x\mapsto \frac {1/2}{  \sqrt{1/2+x^2}}~\text{and}\quad\T_0~:~ \beta \mapsto \begin{cases}
0 & \mbox{ if $\beta\le 1/\sqrt{2}$,}\\
\sqrt{\beta^2-(1/2)} & \mbox{ otherwise,}
\end{cases}\\
 \end{align*}
 that for any $\kappa>0$, one can find $\eps_0 = \eps_0(\kappa)$ such that for any $\eps<\eps_0$ and $\mu_V \in \cP_\eps$, we have, for $\beta\geq 0$, $|\F_V(\T_V(\beta)) -\F_0(\T_0(\beta))|\leq\kappa $ and
 $$|\beta \F_V(\T_V(\beta))^2 +2 \G_V(\T_V(\beta)) -\left [\beta \F_0(\T_0(\beta))^2 +2 \G_0(\T_0(\beta)) \right ]| \leq \kappa~~.  $$
 This proves uniform convergence of $\F_V(\T_V(\cdot))$ to $\F_0(\T_0(\cdot))$ and of $\beta \F_V(\T_V(\cdot))^2 +2 \G_V(\T_V(\cdot))$ to $\beta \F_0(\T_0(\cdot))^2 +2 \G_0(\T_0(\cdot))$. Since we have 
\begin{align*}
\beta \F_0(\T_0(\beta))^2 +2 \G_0(\T_0(\beta)) &=\begin{cases}
\sqrt 2 & \mbox{ if $\beta\le 1/\sqrt{2}$,}\\
\beta  + 1/(2\beta) & \mbox{ otherwise,}
\end{cases} 
\end{align*}
and
\begin{align*}
\F_0(\T_0(\beta))&=\begin{cases}
0 & \mbox{ if $\beta\le 1/\sqrt{2}$,}\\
\sqrt{ 1 - 1/(2\beta^2)} & \mbox{ otherwise,}
\end{cases} 
\end{align*}
the result is proved.
\end{proof}

\begin{proof}[Proof of Theorem \ref{th:mainRec}]
In order to prove Theorem \ref{th:mainRec} we proceed as for Theorem
\ref{th:mainSym}. We consider a sequence of random matrices $\{
\bX_n\}_{n\geq 1}$ of size $n \times p$, generated according to the
\ref{eq:SpikedMatrix}. We use  Lemma \ref{th:upperBounds} and
Gaussian isoperimetry for the $1$-Lipschitz function  $$\max \left \{
  \< \bu, \bX_n \bv\>~:~\|\bu\|_2\leq 1~,~\|\bv\|_2\leq 1~,~\bv\geq
  0\right \}~~,$$ to conclude that with probability at least
$1-n^{-2}$,  
\begin{align}
\sigma^+(\bX_n)\leq \R\rec(\S_V(\beta,\alpha)/\sqrt\alpha)
+\frac{\log n}{n}~~.
\end{align}
This proves, using Borel-Cantelli Lemma, that 
$\limsup_{n \to \infty} \sigma^+(\bx_n)\leq
\R\rec(\S_V(\beta)/\sqrt\alpha)$ almost surely.

By Theorem \ref{th:lowerBoundsAMPnonsymmetric} there exists a
deterministic sequence $\{\delta_t\}_{t}$ such that 
\begin{align}
 \prob \left [\< \hbu^t, \bX \hbv^t\> \geq \R\rec(\S_V(\beta,\alpha)) -
   \delta_t \right ] = 1~~.
\end{align}
Since $\sigma^+(\bX_n) \geq \< \hbu^t, \bX_n \hbv^t\>  $, and by
taking the intersection over $t \in \naturals$, we get 
$\liminf_{n \to \infty} \sigma^+(\bX_n)  \geq \R\rec(\S_V(\beta))= 1$ almost surely.
This concludes the proof of Eq.~(\ref{eq:MainRrec}),  i.e. $\lim_{n \to \infty}
\sigma^+(\bX_n) =  \R\rec(\S_V(\beta,\alpha))$. 

Together with Lemma \ref{th:upperBounds}, and using  $\lim_{x \to 0}
\Delta(x) = 0$, this implies Eq.~(\ref{eq:MainFrec}), i.e. 
$\lim_{n \to \infty} \< \bvz,\bv^+\> = \F_V(\S_V(\beta,\alpha) / \sqrt
\alpha)$ almost surely.

Finally the proof of Eqs.~(\ref{eq:RlimitEpsZeroNonSym}) and
(\ref{eq:FlimitEpsZeroNonSym}) follows from Lemma
\ref{rk:sparseRegimeFG} as in the symmetric case.
 \end{proof}

\bibliographystyle{amsalpha}

\newcommand{\etalchar}[1]{$^{#1}$}
\providecommand{\bysame}{\leavevmode\hbox to3em{\hrulefill}\thinspace}
\providecommand{\MR}{\relax\ifhmode\unskip\space\fi MR }
\providecommand{\MRhref}[2]{%
  \href{http://www.ams.org/mathscinet-getitem?mr=#1}{#2}
}
\providecommand{\href}[2]{#2}

\end{document}